%% file: main.tex
\documentclass[nonblindrev]{Informs3customR}
\usepackage{appendix}
\usepackage{lmodern}  
\usepackage{bold-extra} 
\usepackage{bm} 
\let\theoremstyle\relax

\usepackage{amsthm}  
\usepackage{amsmath}
\usepackage{color}
\usepackage{booktabs}
\usepackage{graphicx}
\usepackage[table]{xcolor} 
\usepackage{amsfonts}
\usepackage[utf8]{inputenc}
\usepackage[colorlinks=true, linkcolor=blue, citecolor=blue, urlcolor=blue, hypertexnames=false]{hyperref}
\usepackage{mathrsfs}  
\usepackage{natbib}
\bibpunct[, ]{(}{)}{,}{a}{}{,}%
 \def\bibfont{\small}%
 \def\bibsep{\smallskipamount}%
\usepackage[english,activeacute]{babel}
\usepackage{graphicx}
\usepackage{amssymb}
\usepackage{setspace}
\usepackage[capitalise,noabbrev]{cleveref}
\usepackage[shortlabels]{enumitem}
\usepackage{hhline}
\usepackage[table]{xcolor}
\usepackage{booktabs}
\usepackage{dsfont} 

\usepackage{bbm}    

\usepackage[T1]{fontenc}

\usepackage{tikz}
\usetikzlibrary{arrows.meta, positioning}
\usepackage{setspace}
\usepackage{float}
\usepackage{pgfplots}

\pgfplotsset{compat=newest} 

\pgfplotsset{compat=1.18}
\usepgfplotslibrary{groupplots}

\usepackage{mathtools, algorithm, algpseudocode}
\usepackage{enumitem}
\usepackage{bbm}
\usepackage{amsfonts}
\usepackage{comment}
\usepackage{float}

\TheoremsNumberedThrough     
\ECRepeatTheorems
\EquationsNumberedThrough

\makeatletter
\def\footnoterule{\kern-3\p@
  \hrule \@width 6.5in \kern 2.6\p@} 
\makeatother
\renewcommand{\baselinestretch}{1.2}
\definecolor{orange}{rgb}{1,0.5,0}
\definecolor{OliveGreen}{cmyk}{0.64,0,0.95,0.40}
\definecolor{BrickRed}{cmyk}{0,0.89,0.94,0.28}

\definecolor{colA}{RGB}{198,30,30}
\definecolor{colB}{RGB}{20,145,50}
\definecolor{colC}{RGB}{25,80,195}
\definecolor{colD}{RGB}{215,120,0}
\definecolor{colE}{RGB}{125,20,185}

\newcommand{\cliff}[1]{}
\newcommand{\avi}[1]{}
\newcommand{\rene}[1]{}
\newcommand{\yichen}[1]{}
\newcommand{\prem}[1]{}

\newcommand{\tim}[1]{\scalebox{0.54}{$#1$}}

\newcommand{\ti}[1]{\scalebox{0.45}{\mbox{\rm #1}}} 

\newcommand{\tpsi}{\tilde{\psi}}
\newcommand{\eps}{\epsilon}
\newcommand{\var}{\mathbb{V}\mbox{\rm ar}}

\newcommand{\D}{\mbox{\rm d}}

\newcommand{\M}{\ti M}
\newcommand{\I}{\ti I}

\newcommand{\e}{\mathbb{E}}
\newcommand{\BN}{\ti{BN}}
\newcommand{\MP}{\ti{MP}}
\newcommand{\Ad}{{\cal A}}
\newcommand{\iid}{{\ti{IID}}}

\newcommand{\join}{\operatorname{join}}

\theoremstyle{definition}
\newtheorem{result}{Result}

\newcommand{\C}[1]{\mathcal{#1}}

\begin{document}
\RUNAUTHOR{Caldentey, Giloni, Hurvich, Talwai and Zhang}
\RUNTITLE{Binomial Smoothing for Inventory and Information Control}
\TITLE{\bf Binomial Smoothing for Inventory \\ and Information Control in Supply Chains\\}
\ARTICLEAUTHORS{%
\AUTHOR{Ren\'e Caldentey}
\AFF{Booth School of Business, The University of Chicago, Chicago, IL 60637,
\EMAIL{rene.caldentey@chicagobooth.edu}}
\AUTHOR{Avi Giloni}
\AFF{Sy Syms School of Business, Yeshiva University, New York, NY 10033,
\EMAIL{agiloni@yu.edu}}
\AUTHOR{Clifford Hurvich}
\AFF{Stern School of Business, New York University,  New York, NY 10012,
\EMAIL{churvich@stern.nyu.edu}}
\AUTHOR{Prem Talwai}
\AFF{Operations Research Center, Massachusetts Institute of Technology, Cambridge, MA 02142, \EMAIL{talwai@mit.edu}}
\AUTHOR{Yichen Zhang}
\AFF{Daniels School of Business, Purdue University,  West Lafayette, IN 47907,
\EMAIL{zhang@purdue.edu}}
}

\date{}
\ABSTRACT{\vspace{0.2cm}

In many decentralized supply chains, upstream firms do not observe market demand directly and instead infer downstream conditions from the order stream. A retailer's replenishment policy therefore plays a dual role: it governs inventory replenishment and shapes the information available for upstream forecasting. This creates a fundamental trade-off. Smoother orders improve upstream predictability, but delaying the response to demand can increase downstream inventory costs. We study how a retailer should optimally smooth demand in a two-tier supply chain with one retailer and one manufacturer when the manufacturer forecasts future orders from the retailer's order history.

We propose \emph{Binomial Smoothing}, a class of replenishment policies that implements delayed demand response by spreading each unit of demand over a finite horizon using binomial weights. The class is interpretable, easy to calibrate, and analytically tractable. Under weakly stationary Gaussian demand satisfying mild regularity conditions, we show that, for any fixed smoothing horizon, the Binomial policy minimizes the manufacturer's forecast error among all policies with the same degree of smoothing. It remains invertible, so the manufacturer can recover demand history from observed orders. More generally, Binomial Smoothing achieves a constant-factor approximation guarantee relative to an optimal policy.

Our results yield a broader insight: replenishment policies should be designed not merely to reduce order variance, as in the traditional bullwhip measure, but to reduce the \emph{unpredictable} component of orders. Carefully designed smoothing can improve supply-chain performance and partially substitute for information sharing, providing a concrete mechanism for \emph{coordination without collaboration}.\vspace{0.1cm}

}\vspace{0.1cm}

\KEYWORDS{Inventory control; order smoothing;  supply chains; forecasting of weakly stationary processes. \vspace{0.2cm}}
\maketitle

\parindent 0em
\vspace{-2.2cm}
\input{Management_Science/Introduction_MS}

\input{Management_Science/Literature_MS}

\input{Management_Science/Overview}

\input{Management_Science/Model_MS}

\input{Management_Science/Preliminaries_MS}
\input{Management_Science/ApproximationPolicies_MS}

\input{Management_Science/Binomial_MS}

\input{Management_Science/Benchmark-Analysis_MS}

\input{Management_Science/Non-IID-Demand}

\input{Management_Science/Extensions_MS}
\input{Management_Science/Conclusions}



\bibliographystyle{ormsv080}

\begingroup
\small
\renewcommand{\baselinestretch}{0.9}\selectfont
\setlength{\bibsep}{0pt plus 0.2ex}
\bibliography{ref}
\endgroup

\newpage
\begin{appendices}
 \renewcommand{\thesection}{Appendix \Alph{section}}

\input{Management_Science/Appendix_MS}

\input{Management_Science/Appendix_Cost_Criterion_MS}

\input{Management_Science/Appendix_Benchmark_Policies}
\end{appendices}
\end{document}

%% file: Management_Science/Introduction_MS.tex
\section{Introduction}\label{sec:Introduction}

A defining feature of many decentralized supply chains is that upstream firms do not observe market demand directly, but instead infer it from the orders they receive \citep{simchi2008designing,HauLee1997,fisher1997what,croson2003impact}. This makes retailers' replenishment policies central to supply chain performance. Orders are not merely operational decisions; they are also information signals. Consequently, retailers' ordering rules govern both how inventory moves through the supply chain and how effectively upstream firms can infer downstream demand conditions for production and replenishment planning. \smallskip

This dual role underscores the importance of {\em order smoothing} in the efficient operation of decentralized supply chains. By order smoothing, we mean inventory replenishment policies that spread the response to demand variability over time, so that orders adjust more gradually to changes in market demand \citep{disney2003bullwhip,MiyaokaHausman2004}. Although the literature has long recognized that smoother orders can mitigate the bullwhip effect and improve coordination \citep{lee1997distortion,HauLee1997,LST2000,Li-Lee2009,WangDisney16}, smoothing is not unambiguously beneficial: delaying the response to realized demand may improve upstream planning while simultaneously increasing downstream inventory-related costs. The central issue, therefore, is not whether demand should be smoothed, but how to smooth it optimally.\smallskip

Most of the existing literature on order smoothing takes a pragmatic approach by focusing on simple heuristic policies, such as moving averages and exponential smoothing, that dampen short-run order variability and are easy to implement \citep{BGP2004,disney2004variance,dejonckheere2004impact}. Yet these rules are not derived from the underlying trade-off between downstream responsiveness and upstream forecastability. As a result, the literature offers limited theoretical guidance on how smoothing should be structured when upstream firms must infer demand conditions from order histories rather than observe market demand directly \citep{zhao2002impact,gavirneni2006information,li2008confidentiality}. \smallskip

This paper addresses this gap by characterizing how demand should be smoothed so as to balance downstream responsiveness and upstream forecastability, and thereby improve supply chain performance. To this end, we study the optimal design of inventory replenishment policies in a two-tier supply chain with one retailer and one manufacturer. The retailer faces weakly stationary demand and chooses a replenishment policy that maps realized demand into current and future orders, while the manufacturer observes only the order stream and replenishes inventory using optimal forecasts. Our central research question is therefore: \emph{How should a retailer optimally smooth demand when orders are both replenishment decisions and the manufacturer's main source of information?} This formulation turns the retailer’s ordering rule into an information-design problem: it jointly determines downstream inventory risk and upstream forecastability. We formalize this interaction through a cost criterion that trades off the retailer's inventory cost against the manufacturer's forecast-error cost, and we ask which policies best navigate this trade-off.\smallskip

We answer this question in three steps. First, we show that there is an irreducible tension between forecastability and responsiveness: no admissible policy can make the manufacturer's orders arbitrarily easy to forecast without generating unbounded inventory volatility for the retailer. Second, we show that the shape of the smoothing profile matters fundamentally. In particular, commonly used benchmark policies can be arbitrarily far from optimal in worst case, implying that smoothing \emph{per se} is not enough. Third, we identify a new class of policies -- \emph{Binomial Smoothing} -- that spreads each unit of demand over a finite horizon using binomial weights. These policies introduce a deliberate, structured delay that is easy to interpret operationally and closely aligned with the manufacturer's forecasting problem.\smallskip

Our main result is that Binomial Smoothing delivers robust near-optimal performance. Under mild regularity conditions, this class achieves a constant-factor guarantee relative to the optimal policy. Thus, Binomial Smoothing provides an implementable answer to the optimal-smoothing question: managers can restrict attention to a simple, finite-dimensional family of delayed-response rules while retaining a provable performance guarantee. The analysis also yields a broader conceptual takeaway. The relevant upstream operational objective is not the total order variance, as in the traditional variance-based measure of the bullwhip effect \citep{HauLee1997,chen2000quantifying}. What matters upstream is the component of orders that is \emph{unpredictable} from the order history, so a good smoothing rule must be evaluated not only by its effect on raw variability, but also by its effect on forecast error. In this sense, carefully designed smoothing can substitute for part of the value of explicit information sharing, even when downstream firms do not communicate demand data directly \citep{LST2000,Li-Lee2009,simchi2005value}.\smallskip

These findings lead to three concrete managerial implications. First, demand smoothing should be viewed as a strategic design choice, not as a mechanical variance-reduction device. In particular, decision makers should evaluate replenishment policies by how they affect both the predictable and unpredictable components of orders, since the latter directly determines the manufacturer's forecasting burden. Second, delaying replenishment can improve system performance, but only when the pattern of delay is calibrated to the upstream forecasting problem. Managers should therefore choose not only how much smoothing to introduce, but also how the replenishment of each unit of demand is distributed over time. Third, standard textbook-like smoothing heuristics can be misleading guides for practice when evaluated only through order variance: two policies may appear similarly smooth, yet induce very different forecastability and inventory outcomes. Managers should therefore compare smoothing policies using forecast-error-based performance measures, rather than relying only on traditional bullwhip or order-variance metrics. In other words, the structure of the response profile is central to supply-chain performance.\smallskip

Our analysis also contributes methodologically. We build on \cite{caldentey2024information}, which characterizes $\epsilon$-optimal policies for a related i.i.d.\ model. That work provides an important theoretical benchmark, but the resulting policies are difficult to interpret in the time domain and not designed as low-dimensional implementable replenishment rules. We complement it with a class of policies that has a transparent time-domain interpretation, can be tuned with a small number of parameters, and extends naturally beyond i.i.d.\ demand. In doing so, we connect the information-sharing literature and the order-smoothing literature through a common lens: optimal smoothing as the design of an order stream that is simultaneously operationally feasible downstream and informative upstream.\smallskip

The remainder of the paper is organized as follows. \Cref{sec:literature} reviews the related literature and positions our work within it. \cref{sec:overview} provides a high-level non-technical overview of the main results and their managerial interpretation. \cref{sec:Model} introduces the model of our two-tier supply chain system and the associated optimization problem. Sections~\ref{sec:preliminaries}--\ref{sec:extensions} develop the analysis, present the Binomial Smoothing class, compare it with benchmark policies, and provide some policy refinements. Finally, \cref{sec:conclusion} offers concluding remarks and directions for future research.

%% file: Management_Science/Literature_MS.tex
\section{Related Literature}\label{sec:literature}

There is a substantial literature in supply chain management dedicated to investigating the performance of inventory policies and their connection to order smoothing and information sharing. We organize our discussion around three themes that are most closely related to our work and motivate our modeling and analytical approach: (a) the bullwhip effect, information sharing, and order variability; (b) order-smoothing heuristics and control-theoretic replenishment rules; and (c) forecast-error-based information design and invertibility. For a more comprehensive review of this literature, we refer the reader to \cite{stadtler2005supply,GEARY20062,HaTang2017} and the references therein. \smallskip

\paragraph{\sf (a) Bullwhip effect, information sharing, and order variability.}
A central theme in this literature is that locally optimal replenishment rules can amplify variability upstream. The classical studies of \cite{lee1997distortion,LST2000,chen2000quantifying,GGS2005} analyze periodic-review order-up-to policies, motivated by their optimality in important single-stage inventory models and by their practical simplicity; see \cref{footnote-base-stock}. However, these policies can perform poorly from a system-wide perspective when the retailer does not share demand information or when the manufacturer relies on imperfect forecasting rules \citep{Raghunathan}. Subsequent work extends the analysis to richer demand processes: \cite{zhang2004impact} and \cite{GGS2005} study ARMA($p,q$) demand, while \cite{Li-Lee2009} use a generalized Martingale Model of Forecast Evolution (MMFE) to study the joint value of information sharing and order smoothing. Among these papers, \cite{Li-Lee2009} is particularly close to our perspective. They show that sharing projected future orders can allow the supplier to recover much of the value of demand information without knowing the retailer’s demand model or ordering policy. They also emphasize that order variability and order uncertainty are distinct: when projected orders are shared, the supplier’s cost is driven by uncertainty in order revisions, not necessarily by the unconditional variance of orders. We build on this distinction but study a complementary setting in which neither projected orders nor market-demand information are shared with the manufacturer. The manufacturer must forecast future orders from the realized order history alone, so the retailer’s replenishment rule becomes an implicit information-design mechanism.\smallskip

Taken together, this stream of work shows that information sharing and order-variability reduction are complementary, but much of this literature measures the manufacturer's exposure through the variance of the retailer's orders or through suboptimal forecasting rules. The gap for our purposes is therefore a forecast-error-based evaluation of smoothing policies: we formulate supply-chain performance in terms of the retailer's inventory volatility and the manufacturer's one-step-ahead mean square forecast error.\smallskip

\paragraph{\sf (b) Order-smoothing heuristics and control-theoretic replenishment rules.}
A second stream proposes replenishment rules that deliberately smooth orders to reduce upstream variability.\footnote{This notion of order smoothing is distinct from the production-smoothing concept studied by \cite{BrayMendelson2015}, who focus on the manufacturer's deliberate stabilization of its own production rather than on the retailer's smoothing of replenishment orders.} \cite{BGP2004} introduce order-smoothing policies, including simple moving average and exponential smoothing rules, and show that reducing order variability can lower total system cost even when it increases the retailer's inventory cost. Related control-theoretic approaches include APIOBPCS ({\em Automatic Pipeline, Inventory and Order-Based Production Control System}), developed and analyzed by \cite{TOWILL01111982,disney2003effect,dejonckheere2004impact,disney2004variance}. These policies adjust orders using demand forecasts, deviations from target inventory, and deviations from target pipeline inventory, thereby providing a flexible way to stabilize production and inventory dynamics. Their performance, however, depends on tuning parameters such as forecast responsiveness and inventory-adjustment speeds, and inappropriate parameter choices can still generate upstream variability. This suggests the need for a smoothing family that is simple and implementable, but also supported by general performance guarantees rather than by model-specific or numerical tuning; we address this need through binomial smoothing policies with uniform approximation guarantees for general weakly stationary demand.\smallskip

\paragraph{\sf (c) Forecast-error-based information design and invertibility.}
A third stream, closer to our approach, treats the retailer's order process not only as a replenishment decision but also as an information signal observed by the manufacturer. This perspective makes invertibility central. \cite{zhang2004impact} and \cite{hosoda2006variance} characterize upstream order processes under invertibility assumptions, while \cite{GGS2005} emphasize that non-invertibility can arise upstream even when the retailer's market demand is invertible. In such cases, the manufacturer's forecasting problem cannot generally be reduced to the time-domain formulas that would apply under invertibility, and the forecasting approach in \cite{GGS2005} does not necessarily coincide with the best linear forecast. Building on this observation, \cite{caldentey2024information} allow the manufacturer to use the full history of retailer orders and characterize optimal forecasts through Kolmogorov's one-step-ahead prediction formula \citep[Section 5.8]{BrockwellDavis}. They reformulate the retailer's optimization problem in the $z$-transform domain, with decision variables in the Hardy space $\mathbb{H}^2$, and construct $\varepsilon$-optimal inventory policies for the case of i.i.d.\ demand. Relatedly, \cite{talwai2026informationdelay} introduce the notion of {\em information delay}, connect implicit information transmission to the group delay of the retailer's ordering transfer function, and develop tractable invertible ARMA policies that approximate the limiting optimal filter. The remaining challenge is to translate these forecast-error and Hardy-space insights into a transparent replenishment rule that applies beyond i.i.d.\ demand and avoids singular or difficult-to-implement limiting filters. \smallskip

\paragraph{\bf Summary of Contributions and Positioning.}
Our paper contributes to the literature on supply-chain information sharing, order smoothing, and the bullwhip effect by studying how replenishment policies can control not only the physical flow of inventory, but also the information conveyed upstream through orders. Existing work has emphasized the value of explicit information sharing, the propagation of order variability, and the performance of specific smoothing heuristics such as moving-average and exponential-smoothing rules. We instead formulate order smoothing as an operational information-design problem: the retailer chooses how demand is embedded in the order stream to balance downstream inventory responsiveness and upstream forecastability. This perspective shifts attention from whether smoothing reduces order variance to how the shape of the smoothing profile affects the manufacturer's forecast error and the retailer's inventory volatility.

Methodologically, this paper builds on \cite{caldentey2024information}, which developed the Hardy-space and inner--outer factorization framework for studying how inventory-replenishment policies control the information content of orders. That work provides the theoretical foundation for evaluating the value of information sharing and for constructing $\epsilon$-optimal policies under no information sharing. Our contribution is different in both objective and output. Rather than characterizing limiting optimal policies that are largely frequency-domain objects and difficult to interpret operationally, we ask whether this theory can lead to a simple, implementable smoothing rule. We answer this question by identifying Binomial Smoothing, a finite-history class with a direct time-domain interpretation: each unit of demand is spread over a finite horizon according to binomial weights. We show that this class minimizes the manufacturer's forecast error within each finite MA$(q)$ class, remains invertible, extends beyond the i.i.d.\ demand setting, and delivers a uniform constant-factor performance guarantee relative to the optimal admissible policy. Thus, the paper moves the theory from an exact but abstract characterization of optimal information control toward an operational prescription for achieving \emph{coordination without collaboration} through the information embedded in replenishment orders.

%% file: Management_Science/Overview.tex
\section{An Overview of Main Results}\label{sec:overview}

Before turning to the formal analysis, this section summarizes the setting, the key trade-offs, the main results, and the managerial insights that emerge from our study. Rather than presenting the full mathematical model, we provide a high-level discussion of the core findings and their operational interpretation. Precise assumptions, formal statements, and proofs are deferred to the subsequent sections and appendices.\smallskip

The supply chain system we study consists of a single retailer and a single manufacturer, as illustrated in \cref{fig:twotierSC}.
\input{Figures/Fig_SC_System_MS_Overview}
Market demand $D_t$ evolves as a weakly stationary Gaussian process (not necessarily i.i.d.), and the retailer serves this demand from a finished-goods inventory $I_t$, replenished periodically through orders $O_t$ placed with the manufacturer. The manufacturer fulfills those orders from its own inventory, replenished according to a state-dependent base-stock policy $S_t$, and uses expedited production whenever necessary to maintain full service. We consider a standard setting with linear holding and backlogging costs at the retailer, linear holding and expediting costs at the manufacturer, full backlogging, and a constant replenishment one-period lead time.\smallskip

We represent the retailer's replenishment policy as a linear {\em impulse response} function that smooths market demand over time before transmitting it upstream:
\[\tag{Retailer's Demand-Smoothing Policy}
O_t = \sum_{n=0}^\infty \varphi_n\, D_{t-n},
\]
where the coefficient sequence $\varphi=\{\varphi_n\}_{n\ge 0}$ characterizes the policy and satisfies $\sum_{n=0}^\infty \varphi_n=1$. This stability condition ensures that average orders match average demand, $\e[O_t]=\e[D_t]$, and prevents the retailer's inventory from drifting without bound.\smallskip

Central to our analysis is the operational meaning of the impulse-response coefficients. Each coefficient $\varphi_n$ represents the fraction of current demand that is translated into orders $n$ periods later. Thus, $\varphi$ defines a {\em delayed demand-response} rule: it determines how realized demand is converted into future orders and, in doing so, shapes both the physical flow of product and the information conveyed to the manufacturer.\smallskip

This interpretation makes the main trade-off transparent. A policy that places most of its weight on $\varphi_0$ responds quickly to demand and keeps the retailer closely aligned with current market conditions, but it also passes demand fluctuations upstream almost immediately. A policy that spreads its coefficients over a longer horizon smooths the order stream and improves its predictability for the manufacturer, but it also makes the retailer replenish more gradually and exposes its net inventory to larger demand fluctuations. Thus, smoothing has the potential to improve upstream forecastability, but it may also make downstream inventory management more difficult.\smallskip

\Cref{fig:impulseresponse} illustrates this interpretation.
\input{Figures/Fig_ImpulseResponse}
The height of each bar represents the coefficient $\varphi_n$, the fraction of demand observed $n$ periods earlier that is incorporated into the current order. In the example shown, $\varphi_0=0.4$ and $\varphi_1=0.25$, so $40\%$ of demand is ordered immediately, $25\%$ one period later, and the remaining $35\%$ over subsequent periods. More generally, a front-loaded impulse response corresponds to a reactive policy, whereas a more dispersed impulse response corresponds to a smoother and more delayed demand-response rule. This observation points to a central theme of the paper: the shape of the impulse response function is critical. Our proposed Binomial Smoothing policy has a distinctive, and not obvious {\em a priori}, impulse-response shape that captures an optimal delayed response in a sense made precise later; see Figure~\ref{fig:transferfunction} for an illustration.\smallskip

We formulate the supply-chain problem from the perspective of the retailer choosing the replenishment policy $\varphi$. Specifically, we assume that the retailer acts as a Stackelberg leader: it selects $\varphi$ to maximize its expected long-run average payoff while  guaranteeing the manufacturer's participation,
\[ \tag{Retailer's Inventory Problem}
\max_{\varphi}\ \Pi^{\ti R}(\varphi)
\qquad
\text{subject to}
\qquad
\Pi^{\ti M}(\varphi)\geq \bar\pi^{\ti M},
\]
where $\Pi^{\ti R}(\varphi)$ and $\Pi^{\ti M}(\varphi)$ denote the retailer's and manufacturer's expected long-run average payoffs, respectively, and $\bar\pi^{\ti M}$ is the manufacturer's outside option or participation requirement. This formulation captures the main economic tension of the model: by changing its replenishment policy, the retailer affects not only its own inventory performance, but also the predictability and cost of the order stream faced by the manufacturer.\smallskip

The effect of $\varphi$ on supply-chain performance and the retailer's inventory management problem is governed by two quantities. On the retailer side, long-run average inventory cost is proportional to the square root of the stationary variance of net inventory,
\[\tag{Retailer's Inventory Volatility}
\sigma_{\I}^2(\varphi)=\lim_{t\to\infty}\var[I_t].
\]
On the manufacturer side, long-run average replenishment and holding cost is proportional to the square root of the one-step-ahead mean squared forecast error (MSFE) of the order process,
\[\tag{Manufacturer's MSFE}
\sigma_{\ti M}^2(\varphi)=\var\bigl(O_{t+1}\mid{\cal F}^{\ti M}_t\bigr),
\]
where ${\cal F}^{\ti M}_t:=\sigma\bigl(O_\tau:\tau\le t\bigr)$ denotes the manufacturer's information at time $t$.\footnote{The fact that $\sigma_{\ti M}^2(\varphi)$ does not depend on $t$ follows from the stationarity of the retailer's order process.} The retailer's policy therefore affects supply-chain performance through two linked channels: downstream inventory volatility and upstream order forecastability.\smallskip

As we will see, solving the retailer's inventory problem boils down to understanding the trade-off between $\sigma_{\I}(\varphi)$ and $\sigma_{\ti M}(\varphi)$. A key contribution of the paper is to show that this trade-off can be analyzed through the shape of the impulse-response function. In particular, we clarify that the issue is not simply how much to smooth, but how to smooth: different smoothing profiles may generate similar reductions in order variability while having very different implications for upstream forecastability and downstream inventory volatility.\smallskip

A convenient reformulation of the retailer's inventory problem is obtained by expressing it as an equivalent cost-minimization problem in which the trade-off between $\sigma_{\I}(\varphi)$ and $\sigma_{\ti M}(\varphi)$ enters through the linear cost functional
\[\tag{Cost Criterion}
{\cal C}(\varphi;\kappa)=\kappa\,\sigma_{\I}(\varphi)+\sigma_{\ti M}(\varphi),
\]
where $\kappa>0$ measures the relative weight placed on the retailer's inventory cost. It is worth noting that, beyond yielding an equivalent formulation of the Stackelberg retailer's inventory problem, this cost-minimization criterion also admits a centralized interpretation: for an appropriate value of $\kappa$, it corresponds to a planner minimizing aggregate long-run inventory-related costs in the supply chain. More generally, varying $\kappa$ allows the analysis to interpolate between objectives that place greater emphasis on retailer inventory volatility, manufacturer forecastability, or overall system performance.\smallskip

Our first result identifies a fundamental limit on smoothing.\smallskip

\begin{result}[Trade-off between forecastability and inventory volatility]
No admissible replenishment policy can make the manufacturer's order process arbitrarily easy to forecast without generating unbounded inventory volatility for the retailer. In particular, along any sequence of admissible policies for which $\sigma_{\ti M}(\varphi)\downarrow 0$, we must have $\sigma_{\I}(\varphi)\uparrow\infty$.
\end{result}\smallskip

Practically, this result implies that demand smoothing is inherently costly for the retailer: a supply chain cannot make upstream orders arbitrarily predictable without eventually generating prohibitive downstream inventory volatility. The design question is therefore not whether smoothing is desirable, but what form of smoothing best balances these competing forces.\smallskip

Our main analytical result identifies a simple family of policies, which we call \emph{Binomial Smoothing}, that provides an effective answer to this design problem. Under Binomial Smoothing, each unit of demand is distributed over a finite horizon of degree $q \in \mathbb{N}_0=\{0,1,2,\dots\}$ according to binomial weights:
\[\tag{Binomial Policy}
\varphi^{\BN}_n=\frac{1}{2^q}\binom{q}{n}, \qquad n=0,\dots,q.
\]
The parameter $q$ controls the degree of delayed demand response. The class is operationally interpretable, analytically tractable, easy to calibrate, and has a defining optimality property: among all admissible policies whose impulse response has degree $q$, the Binomial Smoothing policy minimizes the manufacturer's one-step-ahead MSFE.\smallskip

\begin{result}[Optimality within the class of degree-$q$ policies]
Fix $q\in\mathbb{N}_0$ and consider the class of admissible policies whose impulse response has degree $q$. Among all such policies, the Binomial Smoothing policy of degree $q$ minimizes the manufacturer's root mean squared forecast error, $\sigma_{\ti M}(\varphi)$.
\end{result}\smallskip

This result gives a sharp interpretation of the binomial weights. The degree $q$ determines how far the response to demand is spread into the future, while the binomial profile determines how that delayed response should be allocated across periods to make the order stream as forecastable as possible. Thus, once the smoothing horizon is fixed, Binomial Smoothing is the most favorable policy for upstream forecasting.\smallskip

A second defining feature is {\it invertibility}. In our setting, invertibility means that the history of the retailer's orders contains enough information for the manufacturer to recover the underlying history of market demand, even though demand is not directly observed upstream. 
Thus, Binomial Smoothing does not reduce forecast error by discarding information. It preserves the informational content of demand while reorganizing when that information appears in the order stream. In this sense, structured smoothing can partially replicate one specific benefit of explicit information sharing: making the manufacturer's future order stream more predictable, without requiring the retailer to disclose proprietary demand data directly. \smallskip

\begin{result}[Invertibility of Binomial Smoothing]
Every Binomial Smoothing policy is invertible. Equivalently, by observing the history of the retailer's orders, the manufacturer can recover the unobserved history of market demand.
\end{result}\smallskip

Taken together, these last two properties make Binomial Smoothing especially appealing. For a fixed degree of smoothing, the binomial profile is the best possible choice for reducing the manufacturer's forecast error, and this gain does not come from destroying demand information. Most importantly, the family also delivers a uniform approximation guarantee relative to the optimal policy.\smallskip

\begin{result}[Uniform Performance Guarantee]
For any choice of $\kappa$, one can select a Binomial Smoothing policy whose cost is within the same fixed factor of the minimum cost achievable by any admissible replenishment policy.
\end{result}\smallskip

In this sense, Binomial Smoothing provides a uniform near-optimal guarantee across the full range of objectives considered in our model. Managers can therefore rely on a simple and easy-to-calibrate family of replenishment rules to achieve robust supply-chain performance without designing highly customized or operationally complex smoothing policies.\smallskip

\Cref{fig:transferfunction} illustrates Binomial Smoothing and compares it with three familiar benchmarks: the Myopic policy, Simple Moving Average, and Exponential Smoothing.
\input{Figures/Fig_Policies_Comparison}
These policies span a natural range from full responsiveness to substantial smoothing: the Myopic policy places all weight on current demand, Simple Moving Average distributes demand uniformly over a fixed number of periods, and Exponential Smoothing assigns geometrically decaying weights that front-load the response while still delaying part of it. In the example, the Binomial and Simple Moving Average policies have degree five, and the Exponential Smoothing policy has smoothing parameter $0.5$.\smallskip

Each row in \Cref{fig:transferfunction} corresponds to a different policy. The left column shows the same demand sequence $\{D_t\}$, the middle column shows the impulse-response coefficients $\varphi_n$, and the right column shows the resulting order sequence, $O_t=\sum_n \varphi_n D_{t-n}$. Each order bar is decomposed into colored sub-bars of height $\varphi_n D_{t-n}$, with colors matching the demand realizations in the left panel. For each order bar, only the top segment in the new color, corresponding to same-period demand, represents information not previously available to the manufacturer; this is the manufacturer's forecast error.\smallskip

The figure highlights the central trade-off. The Myopic policy is the most responsive, but it generates the least forecastable order stream. The Simple Moving Average produces the smoothest orders in the variance sense, but it introduces substantial delay. Exponential Smoothing responds more quickly, but its geometric decay is not especially well suited to upstream forecasting. Binomial Smoothing places more weight on intermediate lags and provides an effective compromise: it greatly improves manufacturer forecastability while keeping the increase in retailer inventory volatility comparatively moderate.\smallskip

\begin{result}[Binomial Smoothing versus benchmark policies]
The best Binomial Smoothing policy attains a supply-chain cost within a fixed constant factor of the optimum. By contrast, neither the Myopic policy nor the Simple Moving Average and Exponential Smoothing families admit an analogous uniform guarantee. Equivalently, there is no finite constant that bounds their worst-case relative performance.
\end{result}\smallskip

To make the comparison concrete, \Cref{tab:policycomparison_iid} reports the corresponding performance metrics for i.i.d.\ demand. Both Binomial Smoothing and Simple Moving Average generate relatively high retailer inventory volatility, but Binomial Smoothing delivers by far the largest improvement in manufacturer forecastability.\smallskip

\begin{table}[h!]
\centering
\setlength{\arraycolsep}{10pt}
\begin{tabular}{lccc}
\toprule
Policy & $\var(O_t)/\var(D_t)$ & ${\rm MSFE}(O_t)/\var(D_t)$ & $\var(I_t)/\var(D_t)$ \\
\midrule
Myopic & 1.0000 & 1.0000 & 1.0000 \\
Simple MA & 0.1667 & 0.0278 & 2.5278 \\
Exponential Smoothing & 0.3333 & 0.2500 & 1.3333 \\
Binomial & 0.2461 & 0.0010 & 2.8848 \\
\bottomrule
\end{tabular}
\caption{Normalized metrics for the four benchmark replenishment policies in the special case of i.i.d.\ demand.}
\label{tab:policycomparison_iid}
\end{table}

The table also connects our results to the classical bullwhip literature. The first column, $\var(O_t)/\var(D_t)$, is the traditional bullwhip measure, which compares order variance with demand variance. By this metric, Simple MA appears best, reducing the variance ratio to $0.1667$ versus $0.2461$ under Binomial Smoothing. The second column, ${\rm MSFE}(O_t)/\var(D_t)$, gives a different ranking. It measures the unpredictable component of orders relative to demand variability and therefore captures a {\em forecast-adjusted} notion of the bullwhip effect: from the manufacturer's standpoint, the relevant burden is not total order variance {\em per se}, but the part that cannot be anticipated from past orders.\footnote{The distinction between variability and uncertainty (i.e., lack of predictability) in market demand and orders, in the context of the bullwhip effect, has been repeatedly recognized and discussed in the supply chain management literature (see \citealp{Aviv2001, hosoda2006variance, Li-Lee2009, chenlee2012bullwhip, BrayMendelson2015}). The specific forecast-adjusted bullwhip measure, defined as the MSFE of orders divided by the MSFE of demand, appears in \cite{CaldenteyGiloniHurvich2020OrderSmoothing}, who show that it can differ sharply from the traditional variance-based measure and that commonly used replenishment policies can be ranked quite differently depending on the measure used; in some cases, the same policy can generate bullwhip according to one measure while dampening it according to the other. Independently, \cite{Saoud25} propose the same metric, which they call the {\em ratio of forecast uncertainty}, and provide a comprehensive discussion of the literature on the bullwhip effect and its measurement.} Under this measure, Binomial Smoothing reduces normalized forecast error to $0.0010$, compared with $0.0278$ under Simple MA. Thus, while Simple MA is preferred under the traditional variance-based metric, Binomial Smoothing dampens the bullwhip effect more strongly under a forecast-adjusted measure that better reflects the manufacturer's operational burden.\smallskip


\begin{result}[Forecast-Adjusted Bullwhip]
Among all admissible policies with a finite smoothing horizon of degree $q\in\mathbb{N}_0$, the Binomial Smoothing policy minimizes the forecast-adjusted bullwhip measure.
\end{result}\smallskip

In summary, the retailer's policy should be viewed not only as an inventory replenishment rule, but also as a mechanism that jointly shapes downstream inventory risk and the predictability of the manufacturer's order stream. From this perspective, the central operational problem is to balance demand responsiveness against upstream forecastability. The discussion above points to Binomial Smoothing as a simple and effective way to navigate this trade-off, especially along dimensions that are not fully captured by traditional smoothing policies or variance-based bullwhip measures. The following sections make this message precise by formalizing the model, connecting the retailer's inventory problem to the linear cost criterion, and developing the analytical results that support it.

%% file: Figures/Fig_SC_System_MS_Overview.tex
\begin{figure}[h]
\begin{center}
\begin{tikzpicture}[scale=1.1]
\draw [rounded corners] (1.95,-0.7) rectangle (4.95,0.7);
\node at (3.45,0.45)    {\bf Manufacturer};
\node at (3.45,0.05)   {\footnotesize (Base-Stock Policy)};
\node at (3.45,-0.43)  {\small $S_t=m_t+\sigma_{\ti M}\,\zeta^{\ti M}$};
\draw[<-,>=stealth,line width=1pt] (5.0,0.5) -- (8.95,0.5);
\node at (6.975,0.3)  {\footnotesize Retailer's Orders};
\node at (6.975,0.73) {\small $\displaystyle O_t$};
\draw[->,>=stealth,line width=1pt] (5.0,-0.5) -- (8.95,-0.5);
\node at (6.975,-0.27) {\small $O_{t-1}$};
\node at (6.975,-0.79) {\footnotesize Manufacturer's Fulfillment};
\node at (6.975,-1.1)  {\tiny (100\% Service Level)};
\draw [rounded corners] (9.0,-0.7) rectangle (12.5,0.7);
\node at (10.8,0.45)    {\bf Retailer};
\node at (10.8,0.05)   {\footnotesize (Net Inventory)};
\node at (10.8,-0.43)  {\small $I_t=I_{t-1}+O_{t-1}-D_t$};
\draw[<-,>=stealth,line width=1pt] (12.51,0) -- (14.15,0);
\node at (13.5, 0.3)  {\small $D_t$};
\node at (13.5,-0.2)  {\footnotesize Demand};
\draw (15.1,0) ellipse (0.9cm and 0.7cm);
\node at (15.1,0) {\bf Market};

\def\shifty{-1.6}
\def\shiftx{0.8}
\def\shiftD{-0.5}


\node[anchor=west] at (\shiftx,\shifty)
    {\footnotesize \sf Manufacturer's mean forecast $m_t=\e[O_{t+1}|{\cal F}^{\ti M}_t]$ and MSFE $\sigma_{\ti M}^2=\var[O_{t+1}|{\cal F}^{\ti M}_t]$, where ${\cal F}^{\ti M}_t=\sigma\big(O_\tau : \tau \le t\big)$.};
\end{tikzpicture}
\end{center}
\caption{A two-tier supply chain system.}
\label{fig:twotierSC}
\end{figure}

%% file: Figures/Fig_ImpulseResponse.tex
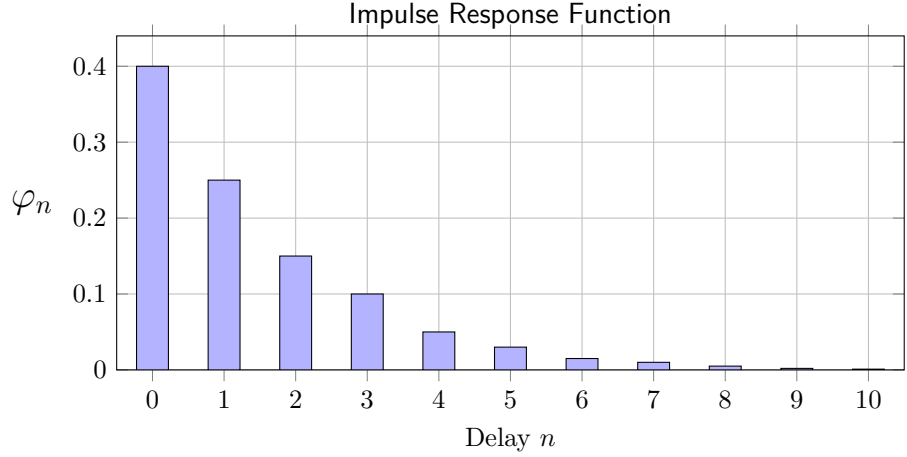
\begin{figure}[h!]
\begin{center}
\begin{tikzpicture}
  \begin{axis}[
    ybar,
    /pgf/bar width=12pt,
    ymin=0,
    xlabel={Delay $n$},
    ylabel={\large $\varphi_n$},
    ylabel style={rotate=-90}, 
    xtick=data,
    enlarge x limits=0.05,
    enlarge y limits={upper, value=0.1},
    grid=major,
    width=12cm,
    height=6cm,
    ticklabel style={font=\small},
    label style={font=\small},
    title={\bf \sf Impulse Response Function},
    title style={yshift=-1.5ex}
  ]
  \addplot+[draw=black] 
    coordinates {
      (0, 0.4)
      (1, 0.25)
      (2, 0.15)
      (3, 0.1)
      (4, 0.05)
      (5, 0.03)
      (6, 0.015)
      (7, 0.01)
      (8, 0.005)
      (9, 0.002)
      (10, 0.001)
    };
  \end{axis}
\end{tikzpicture}
\caption{An illustration of the retailer's impulse response function.}
 \label{fig:impulseresponse}
\end{center}
\end{figure}

%% file: Figures/Fig_Policies_Comparison.tex
\begin{figure}[h]
\begin{center}
\begin{tikzpicture}[font=\small]
 
\newcommand{\demandBars}{%
  \addplot[fill=black!80, draw=black]  coordinates{(-2,105.31452)(-1,137.91103)};
  \addplot[fill=colA!80, draw=colA!90] coordinates{(0,57.31463)};
  \addplot[fill=colB!80, draw=colB!90] coordinates{(1,50.41441)};
  \addplot[fill=colC!80, draw=colC!90] coordinates{(2,82.54504)};
  \addplot[fill=colD!90, draw=colD]    coordinates{(3,180.44810)};
  \addplot[fill=colE!80, draw=colE!90] coordinates{(4,111.32200)};
  \draw[densely dashed,black!55,thick](axis cs:-2.4,100)--(axis cs:4.4,100);
  \node[font=\scriptsize,black!55,right] at (axis cs:4.4,100){$\mu$};
}
 
\pgfplotsset{demandStyle/.style={
  ybar, every axis plot post/.append style={bar shift=0pt},
  /pgf/bar width=9pt, ymin=0, ymax=205,
  xtick={-2,-1,0,1,2,3,4},
  xticklabels={$t{-}2$,$t{-}1$,$t$,$t{+}1$,$t{+}2$,$t{+}3$,$t{+}4$},
  enlarge x limits=0.09,
  grid=major, grid style={dotted,gray!40},
  width=5.8cm, height=4.8cm,
  ticklabel style={font=\scriptsize}, label style={font=\small},
  ytick={0,50,100,150,200}
}}
 
\pgfplotsset{orderStyle/.style={
  ybar stacked, /pgf/bar width=9pt, ymin=0,
  xtick={0,1,2,3,4},
  xticklabels={$t$,$t{+}1$,$t{+}2$,$t{+}3$,$t{+}4$},
  enlarge x limits=0.15,
  grid=major, grid style={dotted,gray!40},
  width=5.8cm, height=4.8cm,
  ticklabel style={font=\scriptsize}, label style={font=\small}
}}
 
 
\begin{axis}[name=axD1, demandStyle,
  title={\bfseries\sffamily Demand $\{D_t\}$}, title style={yshift=-1.5ex}]
  \demandBars
\end{axis}
 
\begin{axis}[name=axT1, at={(axD1.east)}, anchor=west, xshift=0.9cm,
  ybar, /pgf/bar width=9pt, ymin=0, ymax=0.42,
  xtick={0,1,2,3,4,5}, enlarge x limits=0.12,
  ytick={0,0.1,0.2,0.3}, yticklabels={$0$,$0.1$,$0.2$,$0.3$},
  grid=major, grid style={dotted,gray!40},
  width=5.8cm, height=4.8cm,
  title={\bfseries\sffamily Impulse Response $\{\varphi_n\}$},
  title style={yshift=-1.5ex},
  ticklabel style={font=\scriptsize}]
  \addplot[fill=gray!35,draw=black] coordinates{
    (0,0.03125)(1,0.15625)(2,0.31250)(3,0.31250)(4,0.15625)(5,0.03125)};
  \node[font=\small\bfseries\sffamily, anchor=north]
    at (axis description cs:0.5,0.97) {Binomial};
\end{axis}
 
\begin{axis}[name=axO1, at={(axT1.east)}, anchor=west, xshift=0.9cm,
  orderStyle, ymax=135, ytick={0,25,50,75,100,125},
  title={\bfseries\sffamily Orders $\{O_t\}$}, title style={yshift=-1.5ex}]
  \addplot[fill=black!80,draw=black] coordinates{
    (0,119.05922)(1,99.23073)(2,63.21180)(3,24.83968)(4,4.30972)};
  \addplot[fill=colA!80,draw=colA!90] coordinates{
    (0,1.79108)(1,8.95541)(2,17.91082)(3,17.91082)(4,8.95541)};
  \addplot[fill=colB!80,draw=colB!90] coordinates{
    (0,0)(1,1.57545)(2,7.87725)(3,15.75450)(4,15.75450)};
  \addplot[fill=colC!80,draw=colC!90] coordinates{
    (0,0)(1,0)(2,2.57953)(3,12.89766)(4,25.79532)};
  \addplot[fill=colD!90,draw=colD] coordinates{
    (0,0)(1,0)(2,0)(3,5.63900)(4,28.19502)};
  \addplot[fill=colE!80,draw=colE!90] coordinates{
    (0,0)(1,0)(2,0)(3,0)(4,3.47881)};
  \draw[densely dashed,black!55,thick](axis cs:-0.5,100)--(axis cs:4.5,100);
  \node[font=\scriptsize,black!55,right] at (axis cs:4.5,100){$\mu$};
\end{axis}
 
 
\begin{axis}[name=axD2, demandStyle,
  at={(axD1.south west)}, anchor=north west, yshift=-0.7cm]
  \demandBars
\end{axis}
 
\begin{axis}[name=axT2, at={(axD2.east)}, anchor=west, xshift=0.9cm,
  ybar, /pgf/bar width=9pt, ymin=0, ymax=1.25,
  xtick={0,1,2,3,4,5}, enlarge x limits=0.12,
  ytick={0,0.25,0.5,0.75,1.0},
  grid=major, grid style={dotted,gray!40},
  width=5.8cm, height=4.8cm,
  ticklabel style={font=\scriptsize}]
  \addplot[fill=gray!35,draw=black] coordinates{
    (0,1.0)(1,0)(2,0)(3,0)(4,0)(5,0)};
  \node[font=\small\bfseries\sffamily, anchor=north]
    at (axis description cs:0.5,0.97) {Myopic};
\end{axis}
 
\begin{axis}[name=axO2, at={(axT2.east)}, anchor=west, xshift=0.9cm,
  orderStyle, ymax=205, ytick={0,50,100,150,200}]
  \addplot[fill=colA!80,draw=colA!90] coordinates{
    (0,57.31463)(1,0)(2,0)(3,0)(4,0)};
  \addplot[fill=colB!80,draw=colB!90] coordinates{
    (0,0)(1,50.41441)(2,0)(3,0)(4,0)};
  \addplot[fill=colC!80,draw=colC!90] coordinates{
    (0,0)(1,0)(2,82.54504)(3,0)(4,0)};
  \addplot[fill=colD!90,draw=colD] coordinates{
    (0,0)(1,0)(2,0)(3,180.44810)(4,0)};
  \addplot[fill=colE!80,draw=colE!90] coordinates{
    (0,0)(1,0)(2,0)(3,0)(4,111.32200)};
  \draw[densely dashed,black!55,thick](axis cs:-0.5,100)--(axis cs:4.5,100);
  \node[font=\scriptsize,black!55,right] at (axis cs:4.5,100){$\mu$};
\end{axis}
 
 
\begin{axis}[name=axD3, demandStyle,
  at={(axD2.south west)}, anchor=north west, yshift=-0.7cm]
  \demandBars
\end{axis}
 
\begin{axis}[name=axT3, at={(axD3.east)}, anchor=west, xshift=0.9cm,
  ybar, /pgf/bar width=9pt, ymin=0, ymax=0.27,
  xtick={0,1,2,3,4,5}, enlarge x limits=0.12,
  ytick={0,0.05,0.10,0.15},
  scaled y ticks=false,
  yticklabel style={/pgf/number format/fixed, /pgf/number format/precision=2},
  grid=major, grid style={dotted,gray!40},
  width=5.8cm, height=4.8cm,
  ticklabel style={font=\scriptsize}]
  \addplot[fill=gray!35,draw=black] coordinates{
    (0,0.16667)(1,0.16667)(2,0.16667)(3,0.16667)(4,0.16667)(5,0.16667)};
  \node[font=\small\bfseries\sffamily, anchor=north]
    at (axis description cs:0.5,0.97) {Simple MA};
\end{axis}
 
\begin{axis}[name=axO3, at={(axT3.east)}, anchor=west, xshift=0.9cm,
  orderStyle, ymax=130, ytick={0,25,50,75,100,125}]
  \addplot[fill=black!80,draw=black] coordinates{
    (0,104.32469)(1,86.32924)(2,60.05334)(3,40.53759)(4,22.98517)};
  \addplot[fill=colA!80,draw=colA!90] coordinates{
    (0,9.55244)(1,9.55244)(2,9.55244)(3,9.55244)(4,9.55244)};
  \addplot[fill=colB!80,draw=colB!90] coordinates{
    (0,0)(1,8.40240)(2,8.40240)(3,8.40240)(4,8.40240)};
  \addplot[fill=colC!80,draw=colC!90] coordinates{
    (0,0)(1,0)(2,13.75751)(3,13.75751)(4,13.75751)};
  \addplot[fill=colD!90,draw=colD] coordinates{
    (0,0)(1,0)(2,0)(3,30.07468)(4,30.07468)};
  \addplot[fill=colE!80,draw=colE!90] coordinates{
    (0,0)(1,0)(2,0)(3,0)(4,18.55367)};
  \draw[densely dashed,black!55,thick](axis cs:-0.5,100)--(axis cs:4.5,100);
  \node[font=\scriptsize,black!55,right] at (axis cs:4.5,100){$\mu$};
\end{axis}
 
 
\begin{axis}[name=axD4, demandStyle,
  at={(axD3.south west)}, anchor=north west, yshift=-0.7cm]
  \demandBars
\end{axis}
 
\begin{axis}[name=axT4, at={(axD4.east)}, anchor=west, xshift=0.9cm,
  ybar, /pgf/bar width=9pt, ymin=0, ymax=0.65,
  xtick={0,1,2,3,4,5}, enlarge x limits=0.12,
  ytick={0,0.1,0.2,0.3,0.4,0.5},
  grid=major, grid style={dotted,gray!40},
  width=5.8cm, height=4.8cm,
  ticklabel style={font=\scriptsize}]
  \addplot[fill=gray!35,draw=black] coordinates{
    (0,0.5)(1,0.25)(2,0.125)(3,0.0625)(4,0.03125)(5,0.03125)};
  \node[font=\small\bfseries\sffamily, anchor=north]
    at (axis description cs:0.5,0.97) {Exponential Smoothing};
\end{axis}
 
\begin{axis}[name=axO4, at={(axT4.east)}, anchor=west, xshift=0.9cm,
  orderStyle, ymax=148, ytick={0,25,50,75,100,125}]
  \addplot[fill=black!80,draw=black] coordinates{
    (0,63.26136)(1,32.40697)(2,15.56972)(3,7.60080)(4,4.30972)};
  \addplot[fill=colA!80,draw=colA!90] coordinates{
    (0,28.65731)(1,14.32866)(2,7.16433)(3,3.58216)(4,1.79108)};
  \addplot[fill=colB!80,draw=colB!90] coordinates{
    (0,0)(1,25.20721)(2,12.60360)(3,6.30180)(4,3.15090)};
  \addplot[fill=colC!80,draw=colC!90] coordinates{
    (0,0)(1,0)(2,41.27252)(3,20.63626)(4,10.31813)};
  \addplot[fill=colD!90,draw=colD] coordinates{
    (0,0)(1,0)(2,0)(3,90.22405)(4,45.11202)};
  \addplot[fill=colE!80,draw=colE!90] coordinates{
    (0,0)(1,0)(2,0)(3,0)(4,55.66100)};
  \draw[densely dashed,black!55,thick](axis cs:-0.5,100)--(axis cs:4.5,100);
  \node[font=\scriptsize,black!55,right] at (axis cs:4.5,100){$\mu$};
\end{axis}
 
\end{tikzpicture}
\caption{Demand $\{D_t\}$, impulse response $\{\varphi_n\}$, and resulting orders $\{O_t\}$ under four order-smoothing policies. Each colored segment of an order bar corresponds to the demand period of the same color.}
\label{fig:transferfunction}
\end{center}
\end{figure}
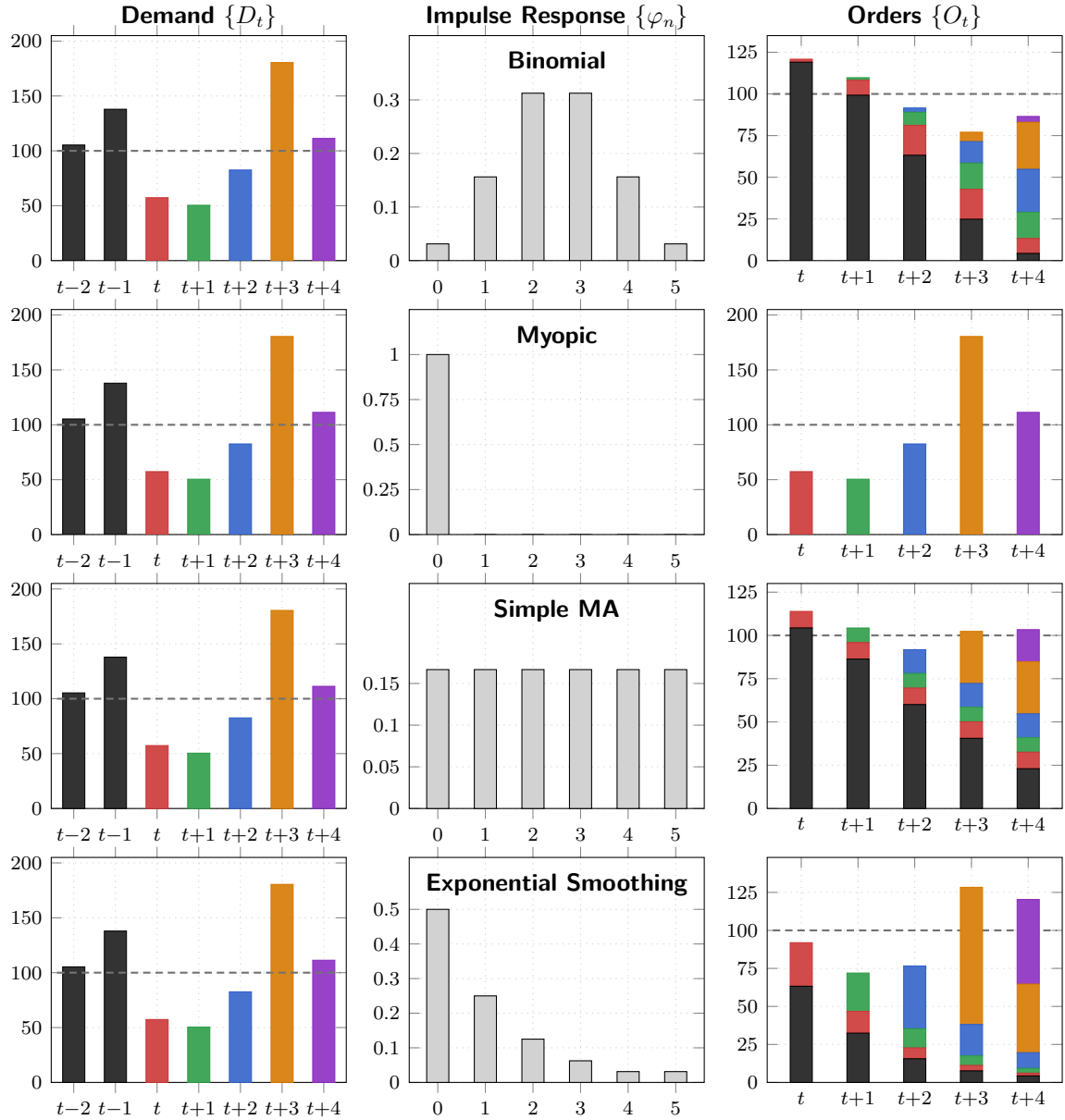

%% file: Management_Science/Model_MS.tex
\section{A Two-Tier Supply Chain Model}\label{sec:Model}

In this section, we provide a precise mathematical description of the supply chain system, as depicted in \cref{fig:twotierSC}, the associated optimization problem introduced in the previous section, and the formal notation used throughout the remainder of the paper.\smallskip

Our modeling builds on a broad stream of prior research on two-echelon and multiechelon supply chains, including \cite{Gavirneni1999}, \cite{LST2000}, \cite{Aviv2001}, \cite{ChenSong2001}, \cite{Li-Lee2009}, \cite{GHS}, and \cite{caldentey2024information}. In line with this literature, we adopt the following operating characteristics: (i) customer demand at the retailer is fully backlogged, so all unfilled demand is eventually fulfilled; (ii) the manufacturer uses a state-dependent order-up-to policy to replenish its inventory; and (iii) the manufacturer offers a 100\% service level to the retailer by relying on expedited production, at an additional cost, whenever on-hand inventory is insufficient.\smallskip

These assumptions reflect standard modeling choices in the multi-echelon inventory literature. Full backlogging is commonly used to capture settings in which demand is not permanently lost but rather delayed, allowing one to focus on the impact of operational policies on inventory and order variability rather than on demand substitution or lost sales. Base-stock policies arise naturally in systems with linear (or convex) costs and stochastic demand, and a large body of work has established their optimality or near-optimality in single- and multi-echelon settings. Finally, the possibility of expedited production (or emergency orders) at a premium cost is a standard device to enforce a high service level while preserving tractability, and is often interpreted as using overtime, premium freight, or external capacity to buffer against stockouts at the retailer. Taken together, these conditions yield a modeling framework that captures a broad and practically relevant class of supply-chain operating environments and ordering policies while preserving analytical tractability.\smallskip 

In what follows, we outline the key components of the model and highlight its central features. For expositional purposes, some derivations are relegated to \ref{App:CostCriterion}; for a more detailed treatment and discussion, we refer the reader to the above-cited literature and the references therein.\medskip

$\diamond$ {\bf Retailer's Inventory Replenishment Policy:} The retailer serves a random market demand process $\{D_t\}$, which we model as a weakly stationary Gaussian process with an MA($\infty$) representation:
\begin{equation}\label{eq:def_demand}
D_t \;=\; d \;+\; \sum_{n=0}^\infty \psi_n \,\eps_{t-n},
\end{equation}
where $d>0$ is the mean per-period market demand, $\psi=\{\psi_n\}_{n\ge 0}$ is an absolutely square-summable sequence of coefficients, and $\eps=\{\eps_t : t\in\mathbb{Z}\}$ is a Gaussian white noise sequence with variance $\var(\eps_t)=\sigma^2_{\eps}$.
In what follows, we set $\sigma_{\eps}=1$. This normalization is without loss of generality because any innovation variance can be absorbed into a rescaling of the MA coefficients $\psi_n$, so only their relative magnitudes matter. \smallskip

It is worth noting that the representation of the demand process in \eqref{eq:def_demand} is fairly general, as any weakly stationary ARMA process admits such a representation.  Moreover, under stationarity, the demand process ${D_t}$ is invertible with respect to the demand-shock sequence $\{\eps_t\}$ (see \cref{dfn:invertible} for a formal definition). This invertibility requirement entails no loss of generality, since \eqref{eq:def_demand} can be interpreted as the Wold representation of a purely non-deterministic Gaussian process ${D_t}$ (see \citealp[\S 5.7]{BrockwellDavis}). In other words, because demand is exogenous to the retailer, which only observes $\{D_t\}$, it can construct a shock sequence $\{\eps_t\}$ that spans the same linear past as $\{D_t\}$.\smallskip

The retailer fulfills the market demand from a stock of finished-goods inventory $I_t$, which is replenished periodically by placing orders $\{O_t\}$ to the manufacturer, with any excess demand being fully backlogged. Under the manufacturer's full-service guarantee, the retailer's net inventory $I_t$ at the end of period $t$ evolves according to the material flow dynamics:
\begin{equation}\label{eq:invR}
I_t = I_{t-1} + O_{t-1} - D_t.
\end{equation}

To replenish its inventory, the retailer employs a stationary ordering policy that gradually smooths demand over time. Specifically, the retailer’s order in period $t$ is a weighted response to current and past demand realizations, given by
\begin{equation}\label{eq:retailerorder}
O_t = \sum_{n=0}^\infty \varphi_n\, D_{t-n}.
\end{equation}

The coefficients $\varphi = \{\varphi_n\}_{n \geq 0}$ define the retailer's inventory policy and serve as the decision variables to be optimized. As such, in what follows we refer to $\varphi$ as the retailer’s (inventory or replenishment) policy, which acts as a linear {\em transfer function} mapping market demand into orders. \smallskip

A few remarks about the retailer's demand in \eqref{eq:def_demand} and its inventory policy in \eqref{eq:retailerorder} are in order. First, the MA($\infty$) representation in \eqref{eq:def_demand} entails no loss of generality: by the Wold decomposition (\citealp{Wold1938}), any purely non-deterministic weakly stationary process admits such a representation. Second, many ordering policies studied in the inventory literature can be expressed as causal, weakly stationary linear filters of past demand shocks and therefore admit an MA($\infty$) representation of the form \eqref{eq:retailerorder} (see \citealp{BrockwellDavis}). In particular, if the demand process is ARMA and the retailer employs a linear forecasting-based replenishment rule (including the widely used one-step-ahead order-up-to policy; e.g., \citealp{Veinott1965}), then the resulting order process is itself ARMA and hence admits the representation \eqref{eq:retailerorder}. More generally, even when demand is not ARMA, the same forecasting-based rule remains a causal weakly stationary linear process and thus still admits an MA($\infty$) expansion. Other examples include the AR(1) policy considered by \cite{LST2000}, as well as simple moving-average and exponential-smoothing policies examined by \cite{BGP2004}, to name a few.\smallskip

As will become clear from the analysis that follows, it is convenient to reinterpret the retailer’s ordering policy from the time domain to the $z$-transform domain. To this end, we associate to the retailer's policy $\varphi = \{\varphi_n\}_{n \geq 0}$ the $z$-transform
\begin{equation}\label{eq:ztrans}
\varphi(z) = \sum_{n=0}^\infty \varphi_n\,z^n.
\end{equation}
Similarly, we define the $z$-transform of the demand process in \eqref{eq:def_demand} by $\psi(z)=\sum_{n=0}^\infty \psi_n\,z^n$. For future reference, we use the notation $\psi\equiv\mathds{1}$, equivalently $\psi(z)=1$ for all $z$, for the special case of i.i.d.\ demand.\smallskip

In what follows, we impose regularity conditions on the market demand process to exclude degenerate or ill-conditioned cases. \smallskip

\begin{assumption}\label{assm1} The $z$-transform of the market demand $\psi(z)$ is uniformly bounded and uniformly bounded away from zero on $\mathbb{T}:=\{z \in \mathbb{C}\colon |z|=1\}$, i.e.,
\begin{equation}\label{eq:boundedpsi}
0< \psi_{\inf}:=\inf_{z\in \mathbb{T}}\lvert \psi(z)\rvert \le \sup_{z\in \mathbb{T}}\lvert \psi(z)\rvert =:\psi_{\sup}< \infty.
\end{equation}
\end{assumption}
These conditions are mild requirements and simply ensure that demand uncertainty behaves in a regular and stable way. In particular, they rule out extreme cases in which demand fluctuations disappear altogether or become excessively large in certain patterns over time. \smallskip

To ensure the long-term stability of the retailer’s inventory process in \eqref{eq:invR}, we restrict its ordering policy to the class of admissible policies, denoted by $\mathcal{A}$. A policy $\varphi = \varphi(z)$ is admissible if it satisfies two conditions: (i) it fulfills market demand on average, that is, $\e[O_t]=d$ or equivalently $\varphi(1) = 1$, ensuring that every unit of demand is eventually ordered, and (ii) it induces a finite steady-state inventory variance, $\sigma^2_{\ti I}(\varphi) < \infty$. The relationship between $\sigma^2_{\ti I}(\varphi)$ and $\varphi$ is established in the following result.\smallskip

\begin{lemma}\label{lem:sigma_I}
Let $\psi(z)$ and $\varphi(z)$ denote the $z$-transforms of the market demand process and the retailer's ordering policy, respectively. Let $I_t$ be the retailer's net-inventory process satisfying the dynamics in \eqref{eq:invR}. Suppose that $\varphi(1)=1$. Then the stationary variance of $I_t$ exists and is given by
\begin{equation}\label{eq_SigmaI_z}
\sigma^2_{\I}(\varphi;\psi):=\lim_{t\to\infty}\var[I_t]
=\frac{1}{2\pi}\int_{-\pi}^{\pi}|\psi(e^{-i\,\lambda})|^2
\left|\frac{e^{-i\,\lambda}\,\varphi(e^{-i\,\lambda})-1}{1-e^{-i\,\lambda}}\right|^2 \D\lambda.
\end{equation}
\end{lemma}
{\sc Proof:} The proofs of this and other results are given in \ref{App:Proofs}.\smallskip

In summary, the class of admissible policies for the retailer is given by
\begin{equation}\label{eq:feasibletpsi}
\Ad(\psi)
:=
\Bigg\{\varphi\in\ell^2 \;\colon\; \varphi(1)=1
\quad\text{and}\quad
\sigma^2_{\I}(\varphi;\psi) < \infty
\Bigg\}.
\end{equation}

An important subclass of admissible policies that will play a central role in the sequel is the class of finite-order moving-average policies. For each integer \(q\in\mathbb{N}_0:=\{0,1,2,\ldots\}\), define
\begin{equation}\label{eq:Aq}
\Ad_q(\psi)
:=
\Bigg\{
\varphi\in\Ad(\psi)
\;\colon\;
\varphi(z)=\sum_{n=0}^{q}\varphi_n z^n
\Bigg\}.
\end{equation}
Policies in $\Ad_q(\psi)$ admit a finite MA$(q)$ representation: the retailer's order in any period depends only on the current demand realization and on at most $q$ lagged demand realizations. Equivalently, a unit of realized demand is fully transmitted through the retailer's orders within $q+1$ periods. Thus, the parameter \(q\) controls the length of the retailer's delayed response. The special case \(q=0\) corresponds to an immediate, myopic response, while larger values of \(q\) allow the retailer to smooth orders over longer horizons. This finite-dimensional subclass is useful both analytically and operationally, since it restricts attention to policies with a transparent and implementable delayed-response structure.
\smallskip

\medskip

$\diamond$ {\bf Manufacturer's Inventory Replenishment Policy:} Turning to the manufacturer’s inventory problem, the manufacturer fulfills the retailer’s orders from on-hand inventory and replenishes that inventory via regular production orders with a one-period production lead time. If a stockout occurs---that is, if the retailer’s order in a given period exceeds the manufacturer’s available on-hand inventory---the manufacturer places an expedited production order to cover the shortfall and thereby maintain its guaranteed $100\%$ service level. These expedited units incur a higher per-unit cost than regular production. Let $c$ denote the manufacturer’s regular per-unit production cost, let $\Delta c$ denote its per-unit incremental cost of expediting, and let $h^{\ti M}$ denote its per-unit per-period holding cost.\smallskip

To control its net inventory level, the manufacturer uses a state-dependent base-stock policy: the base-stock level $S_t$ chosen in period-$t$ minimizes the manufacturer's expected period-($t+1$) inventory holding and production expediting costs, $\e[h^{\ti M}(S_t-O_{t+1})^+ + \Delta c\,(S_t-O_{t+1})^-\mid {\cal F}^{\ti M}_t]$, where  ${\cal F}^{\ti M}_t=\sigma(O_\tau\colon \tau \le t)$ is the manufacturer's information available at period $t$.\footnote{\label{footnote-base-stock}It is well known that a state-dependent base-stock policy is optimal in many practical settings when demand satisfies certain regularity conditions \citep{Karlin60,Karlin60b,Veinott1965,Veinott1969,JohnsonThompson}, or under a ``costless returns'' assumption (as in \citealp{HauLee1997,DongLee2003,MiyaokaHausman2004,Aviv2007,Li-Lee2009}). Even when these conditions are not satisfied, base-stock policies remain close to optimal under mild conditions on the coefficient of variation of demand; see, e.g., \cite{CGHORLetter}.}  For a weakly stationary policy $\varphi \in \Ad$, conditional on ${\cal F}^{\ti M}_t$, the retailer's order $O_{t+1}$ is normally distributed with mean $m_t(\varphi;\psi)=\e[O_{t+1}\mid{\cal F}^{\ti M}_t]$ and variance $\sigma_{\ti M}^2(\varphi;\psi)=\var(O_{t+1}\mid{\cal F}^{\ti M}_t)$.
Accordingly, the manufacturer’s base-stock policy is state-dependent with respect to ${\cal F}^{\ti M}_t$ and takes the familiar form
\begin{equation}\label{eq:basestock} 
S_t(\varphi;\psi)
= m_t(\varphi;\psi) + \sigma_{\ti M}(\varphi;\psi)\,\zeta^{\ti M},
\quad \text{where } 
\zeta^{\ti M} := \Phi^{-1}\!\left(\frac{\Delta c}{h^{\ti M}+\Delta c}\right).
\end{equation}

Here, $m_t(\varphi;\psi)$ and $\sigma^2_{\ti M}(\varphi;\psi)$ denote the manufacturer’s one-step-ahead mean demand forecast and mean squared forecast error (MSFE), respectively. In words, \eqref{eq:basestock} implies that in each period the manufacturer replenishes its inventory up to a base-stock level $S_t(\varphi)$ that covers the one-step-ahead mean forecast demand $m_t(\varphi;\psi)$, augmented by a safety stock $\sigma_{\ti M}(\varphi;\psi)\, \zeta^{\ti M}$ proportional to the square root of the MSFE.
\smallskip

To characterize $m_t(\varphi;\psi)$ and $\sigma^2_{\ti M}(\varphi;\psi)$, we first describe how the manufacturer forms one-step-ahead forecasts of the retailer's order process $\{O_t\}$. This characterization relies on a factorization of the retailer's ordering rule $\varphi(z)$ into an outer component and an inner component, known as the {\sf Smirnov--Beurling factorization}; see \cite{rudin1970real}. We defer the technical details to the appendix and state here only the implication needed for our analysis. The statement of the result uses the {\sf Backshift operator}, defined by ${\cal B}D_t=D_{t-1}$, so that the retailer's orders in \eqref{eq:retailerorder} can be written as $O_t=\varphi({\cal B})\,D_t=d+\varphi({\cal B})\,\psi({\cal B})\,\eps_t$.\smallskip

\begin{lemma}[Spectral Factorization]\label{lem:spectral}
Let $\psi(z)$ and $\varphi(z)$ denote the $z$-transforms of the market demand process and the retailer's ordering policy, respectively. Write $\varphi(z)={\cal Q}(z)\,{\cal I}(z)$ for the Smirnov--Beurling factorization of $\varphi(z)$, where ${\cal Q}(z)$ is outer and ${\cal I}(z)$ is inner. Then the manufacturer's one-step-ahead conditional mean forecast of $O_{t+1}$ and the associated mean squared forecast error (MSFE) satisfy
\begin{align}
m_t(\varphi;\psi)
&=
\e[O_{t+1}\mid{\cal F}^{\ti M}_t]
=
d+
\frac{\psi({\cal B})\,{\cal Q}({\cal B})-\psi(0)\,{\cal Q}(0)}{{\cal B}}\,
{\cal I}({\cal B})\,\eps_t, \\
\sigma^2_{\ti M}(\varphi;\psi)
&=
\var(O_{t+1}\mid{\cal F}^{\ti M}_t)
=
|\psi(0)|^2\,|{\cal Q}(0)|^2 =
|\psi(0)|^2
\exp\!\left(
\frac{1}{2\pi}\int_{-\pi}^{\pi}
\log\bigl|\varphi(e^{-i\lambda})\bigr|^2\,\mathrm{d}\lambda
\right). \label{eq_SigmaS_z}
\end{align}
\end{lemma}

The rightmost expression in \eqref{eq_SigmaS_z} is Kolmogorov's formula for the one-step-ahead prediction error; see \citet[Section~5.8]{BrockwellDavis}.\smallskip

From an operational standpoint, the Smirnov--Beurling factorization in \cref{lem:spectral} conceptually untangles the retailer's ordering policy into two distinct information streams. The outer component represents the demand information that is immediately accessible and actionable for the upstream manufacturer's forecasts. Conversely, the inner component captures any information that is hidden by the smoothing process itself. Thus, minimizing the manufacturer's forecast error fundamentally relies on maximizing the transparency of the outer component. The retailer's inventory volatility, however, depends on the full ordering policy, and hence on both the outer and inner components. In particular, the inner component remains operationally relevant because it changes the timing of the retailer's orders relative to demand, which affects how the retailer's inventory evolves over time.\smallskip

The result in \cref{lem:spectral} simplifies further in \cref{cor:invertible} below when the retailer's ordering policy is invertible in the following sense.\smallskip

\begin{definition}[Invertibility]\label{dfn:invertible}
Let $X=\{X_t\}$ be a weakly stationary process and let $Y_t=\theta({\cal B})\,X_t$ for some transfer function $\theta(z)$. We say that $Y$ is \emph{invertible with respect to $X$} if, for every $t$, $X_t \in \overline{\rm span}\{Y_t,Y_{t-1},Y_{t-2},\dots\}$, where the closure is taken in $L^2$. Equivalently, $\theta(z)$ is outer.
\end{definition}\smallskip

The following corollary follows from \cref{lem:spectral}.

\begin{corollary}[Invertible Policies]\label{cor:invertible}
If the retailer's ordering policy is invertible, i.e., if $\varphi(z)$ is outer, then
\begin{equation}\label{eq:invertible_orders}
\sigma^2_{\ti M}(\varphi;\psi)=|\psi(0)|^2\,|\varphi(0)|^2.
\end{equation}
\end{corollary}
In particular, the Binomial Smoothing policy that we propose is invertible. Operationally, invertibility means that smoothing orders does not mask demand information: the manufacturer can recover the underlying demand history from the observed order stream. Thus, the forecastability gains generated by Binomial Smoothing come from delaying and smoothing the retailer's orders, rather than from making demand information unrecoverable.\smallskip

$\diamond$ {\bf Retailer's Inventory Problem and Supply-Chain Cost Criterion:} 
As anticipated in \cref{sec:overview}, we formulate the supply-chain design problem from the perspective of a retailer that acts as a Stackelberg leader, choosing the replenishment policy $\varphi$ to maximize its expected long-run average payoff while guaranteeing the manufacturer's participation:
\[
\tag{Retailer's Inventory Problem}
\max_{\varphi}\ \Pi^{\ti R}(\varphi;\psi)
\qquad
\text{subject to}
\qquad
\Pi^{\ti M}(\varphi;\psi)\geq \bar\pi^{\ti M},
\]
where $\Pi^{\ti R}(\varphi;\psi)$ and $\Pi^{\ti M}(\varphi;\psi)$ denote the retailer's and manufacturer's expected long-run average payoffs, respectively, and $\bar\pi^{\ti M}$ is the manufacturer's participation requirement. This formulation captures the main economic tension in the model: by changing its replenishment policy, the retailer affects not only its own inventory performance, but also the predictability and cost of the order stream faced by the manufacturer.\smallskip

A convenient reformulation of the Retailer's Inventory Problem is obtained by expressing it as an equivalent cost-minimization problem. Under our standard inventory-cost structure---with linear holding and backlogging costs at the retailer, linear holding and expediting costs at the manufacturer, full backlogging of unmet demand, and a constant replenishment lead time---the only policy-dependent components of the long-run payoffs are the retailer's inventory cost and the manufacturer's inventory and expediting cost. Let ${\cal C}^{\ti R}(\varphi;\psi)$ and ${\cal C}^{\ti M}(\varphi;\psi)$ denote these expected long-run average per-period costs. The equivalent cost criterion takes the form
\[
{\cal C}(\varphi;\psi;\kappa)
=
\kappa\,{\cal C}^{\ti R}(\varphi;\psi)
+
{\cal C}^{\ti M}(\varphi;\psi),
\]
where $\kappa\geq 0$ captures the relative weight placed on the retailer's inventory cost. The derivation of this equivalent formulation, including its connection to the retailer's Stackelberg problem with the manufacturer's participation constraint, is provided in \ref{App:CostCriterion}.\smallskip

Moreover, the same cost structure implies that the retailer's cost is proportional to $\sigma_{\ti I}(\varphi;\psi)$, while the manufacturer's cost is proportional to $\sigma_{\ti M}(\varphi;\psi)$. Therefore, after normalizing the cost coefficients and with a slight abuse of notation, we write the supply-chain cost criterion as
\begin{equation}\label{eq:cost_criterion}
{\cal C}(\varphi;\psi;\kappa)
=
\kappa\,\sigma_{\ti I}(\varphi;\psi)
+
\sigma_{\ti M}(\varphi;\psi).
\end{equation}

This criterion provides a parsimonious representation of the central trade-off in the system: policies that make the manufacturer's order stream more predictable typically do so by increasing the retailer's inventory burden, and vice versa. Beyond yielding an equivalent formulation of the Stackelberg retailer's problem, the criterion also admits a centralized interpretation. For an appropriate value of $\kappa$, it corresponds to a planner minimizing aggregate long-run inventory-related costs in the supply chain, an objective that has been central to the literature on coordinated replenishment and information sharing; see, for example, \cite{Aviv2001,Aviv2007,Li-Lee2009,MiyaokaHausman2004,caldentey2024information}. In the limiting cases $\kappa\downarrow 0$ and $\kappa\uparrow\infty$, the criterion places all weight on the manufacturer's and retailer's costs, respectively. Thus, varying $\kappa$ traces the full spectrum of trade-offs between retailer inventory volatility, manufacturer forecastability, and overall system performance.\medskip

$\diamond$ {\bf The Supply-Chain Optimization Problem:} Combining the preceding discussion, the supply-chain optimization problem of interest is
\begin{equation}\label{eq:SCCost}
{\cal C}^*(\psi;\kappa):=\inf_{\varphi \in \Ad(\psi)} {\cal C}(\varphi;\psi;\kappa)
=\inf_{\varphi \in \Ad(\psi)}\big\{\kappa\,\sigma_{\ti I}(\varphi;\psi)+\sigma_{\ti M}(\varphi;\psi)\big\}.
\end{equation}

Problem \eqref{eq:SCCost} is, in general, difficult to solve. Since the decision variable is a function $\varphi$ ranging over the admissible set $\Ad(\psi)$, \eqref{eq:SCCost} is not a finite-dimensional convex program, but rather a constrained infinite-dimensional optimization problem over analytic functions. The demand transfer function $\psi$ plays a central role in this complexity. First, it enters the definition of the admissible set $\Ad(\psi)$ and therefore affects feasibility. Second, it shapes the objective itself through both $\sigma_{\ti I}(\varphi;\psi)$ and $\sigma_{\ti M}(\varphi;\psi)$. In particular, the inventory term reflects how the demand dynamics encoded by $\psi$ interact with the retailer's ordering rule $\varphi$, while the forecastability term depends on the temporal structure of the induced order process. When $\psi$ is simple, this trade-off can sometimes be analyzed explicitly; for general demand with richer temporal dependence, $\psi$ can vary substantially across frequencies, and the balance between the two terms typically does not yield a clean closed-form optimizer.\smallskip

Even in the benchmark case of i.i.d. market demand (e.g., $\psi(z)\equiv \mathds{1}$), \citet{caldentey2024information} construct a sequence of $\varepsilon$-optimal policies whose $z$-transforms converge to a singular inner function (see \cref{prop:solnosharing} below). Because such a limit cannot be uniformly approximated by polynomials or bounded-degree rational functions, any practically implementable finite-order ARMA policy (such as truncating the power-series expansion or adopting any fixed-order linear rule) cannot be $\varepsilon$-optimal for arbitrarily small $\varepsilon$; achieving $\varepsilon$-optimality necessarily requires the policy degree to grow unboundedly as $\varepsilon\downarrow 0$. \smallskip

Motivated by these limitations, we tackle the optimization problem in \eqref{eq:SCCost} by constructing a class of approximation policies that both admits a provable, uniform performance guarantee and reduces the infinite-dimensional optimization in \eqref{eq:SCCost} to a tractable, one-dimensional problem. As a first step, in \cref{sec:preliminaries} we assemble a set of preliminary results that will be used to analyze and approximate the optimal solution. We then introduce in \cref{sec:approximations} the notion of an $\alpha$-approximation class as a performance benchmark for \eqref{eq:SCCost}, providing a uniform multiplicative guarantee relative to an optimal admissible policy for all $\kappa$, and thus a principled criterion for evaluating and comparing implementable policy classes. Guided by these insights, in \cref{sec:ordersmooting} we propose our binomial smoothing policies and show that they achieve a uniformly bounded relative optimality gap under the mild regularity conditions in \cref{assm1} on the demand model.

%% file: Management_Science/Preliminaries_MS.tex
\section{Preliminaries}\label{sec:preliminaries}

In this section, we collect a set of preliminary results that help build intuition about key structural features of the problem and will play an important role in analyzing and approximating the solution to the supply-chain cost minimization problem \eqref{eq:SCCost}. We first review the benchmark i.i.d.\ demand case following \citet{caldentey2024information}, then study the two extreme regimes $\kappa \uparrow \infty$ and $\kappa\downarrow 0$ and obtain partial characterizations of optimal (or asymptotically optimal) replenishment policies in each limit. Finally, we derive two complementary lower bounds on ${\cal C}^*(\kappa;\psi)$ that serve as performance benchmarks; one of these bounds also yields a sufficient condition under which the information-relaxation bound is tight and leads to an explicit optimal policy.

\subsection {IID Demand}

Under i.i.d.\ demand, e.g., $\psi(z)\equiv \mathds{1}$, the optimization problem \eqref{eq:SCCost} simplifies to
\begin{equation} \label{eq:SCCost_IID}
{\cal C}^*_{\ti{IID}}(\kappa):=\inf_{\varphi \in \Ad} \; \kappa\,\sigma_{\ti I}(\varphi;\mathds{1})
+\sigma_{\ti M}(\varphi;\mathds{1}).
\end{equation}

\citet{caldentey2024information} show that the cost minimization problem \eqref{eq:SCCost_IID} exhibits qualitatively different solutions depending on the value of $\kappa$. When $\kappa\ge \sqrt{5}$, the retailer’s optimal policy is invertible and admits a simple MA(1) representation in terms of the demand shocks $\{\epsilon_t\}$. In contrast, when $\kappa<\sqrt{5}$ the solution to \eqref{eq:SCCost_IID} is substantially more intricate: the authors provide an $\varepsilon$-optimal sequence of policies whose associated transfer functions include a factor that converges to a singular inner function. 
For completeness, we present their main results in the following theorem. \smallskip 

\begin{theorem}[I.I.D. Demand - \citealp{caldentey2024information}]\label{prop:solnosharing}
Consider the problem of minimizing the system-wide inventory cost ${\cal C}_{\ti{IID}}(\varphi;\kappa)=\kappa\,\sigma_{\ti I}(\varphi;\mathds{1})+\sigma_{\ti M}(\varphi;\mathds{1})$ over the class of admissible policies $\Ad(\mathds{1})$ in \eqref{eq:feasibletpsi} for the special case in which market demand is i.i.d., $\psi(z)\equiv \mathds{1}$. Then, the optimal system-wide cost is 
\begin{equation}\label{eq:Ciid}
{\cal C}^*_{\ti{IID}}(\kappa)
= \begin{cases}
1+\sqrt{\kappa^2-1}, & \mbox{if }\kappa\ge \sqrt{5},\\[2pt]
\displaystyle \frac{\kappa}{2}\sqrt{5+2\gamma_{\kappa}}+\frac{e^{-\gamma_{\kappa}}}{2}, & \mbox{if } \kappa<\sqrt{5},
\end{cases}
\end{equation}
where $\gamma_\kappa \geq 0$ is the unique solution of $\kappa^2=(5+2\,\gamma)\,\exp(-2\,\gamma)$.

\medskip

\begin{enumerate}[\rm (a)]
\item\label{prop:solnosharing-a} For $\kappa \geq \sqrt{5}$, the retailer's optimal policy satisfies $\varphi^*(z)=\varphi^*_0+(1-\varphi^*_0)\,z$, where
$\displaystyle \varphi^*_0 = 1 - \frac{1}{\sqrt{\kappa^2 - 1}}.$\smallskip

\item\label{prop:solnosharing-b} For $\kappa < \sqrt{5}$,  define, for each $\delta >0$, the outer policies
\[
\varphi_{\tim \delta}(z) = \frac{1 + z}{2} 
\exp\!\left(\gamma_\kappa\, \frac{z-1}{1 +z+\delta}\right).
\]
Then $\{\varphi_{\tim \delta}\colon \delta >0\}$ is an $\varepsilon$-optimal sequence, that is, $\displaystyle {\cal C}^*_{\ti{IID}}(\kappa) = \lim_{\delta \downarrow 0} {\cal C}_{\ti{IID}}(\varphi_{\tim \delta};\kappa)$.
\end{enumerate}
\end{theorem}\medskip
Figure~4 illustrates the retailer's delayed-response function for three values of $\kappa$ and provides intuition for \cref{prop:solnosharing}.
\input{Figures/Fig_IR_IID_Case}
 The left panel, which corresponds to a relatively large value of $\kappa$, shows a simple and rapidly decaying impulse response: the retailer places most of the weight on current demand, with only a modest adjustment in the next period. By contrast, the middle and right panels, which correspond to smaller values of $\kappa$, show impulse responses that become increasingly spread out, oscillatory, and irregular. Operationally, this means that as the system places greater emphasis on reducing manufacturer forecast risk relative to retailer inventory variability, the theoretically optimal policy requires the retailer to respond to demand shocks through a long and intricate sequence of delayed adjustments. Such policies are difficult to use in practice because they rely on non-rational, infinite-memory singular inner factors and therefore cannot be represented by finite-dimensional ARMA processes. Not only are these singular limiting policies computationally prohibitive, but their oscillatory nature makes them operationally fragile to model misspecification.  Accordingly, although \cref{prop:solnosharing} provides an exact theoretical solution to \eqref{eq:SCCost_IID}, the policy in part~(b), together with its associated cost, should be interpreted as a benchmark for assessing the performance gap of practically implementable policies, rather than as an operational prescription.  This sharp contrast motivates our departure from exact Hardy-space filters, pivoting instead to robust, finite-window approximations that remain close to optimal while being operationally implementable.\smallskip

It is also worth noting that, when the retailer's costs dominate those of the manufacturer (i.e., $\kappa\geq\sqrt{5}$), the optimal policy is invertible and has a simple MA(1) structure. This invertibility provides intuition for why the policy takes an MA(1) form: given the manufacturer’s one-period production lead time, delaying the demand signal by more than one period is not cost-effective. In this regime, the optimal ordering rule can be written as $O_t=\varphi^*_0\,D_t+(1-\varphi^*_0)\,D_{t-1}$,
so the retailer smooths demand using at most one lagged observation. A key contribution of \cref{prop:solnosharing} is the explicit characterization of the coefficient $\varphi^*_0$ that minimizes total supply chain cost. \smallskip

Let us conclude our discussion of the i.i.d. demand case with a structural property of optimal finite-memory policies. While \citet{caldentey2024information} characterize the i.i.d.\ benchmark over the full admissible class $\Ad(\mathds{1})$ defined in \eqref{eq:feasibletpsi}, for $\kappa<\sqrt{5}$ that benchmark is approached through limiting policies that are not contained in any fixed finite-degree class $\Ad_q(\mathds{1})$ defined in \eqref{eq:Aq}. Thus, understanding the structure of an optimizer within $\Ad_q(\mathds{1})$ requires a separate argument. The next proposition uses the symmetry of reciprocal polynomials to establish a sharp bound on the average delay induced by any optimal MA$(q)$ policy.\smallskip

\begin{proposition}\label{prop:group_delay_bound}
Consider the i.i.d.\ demand case, $\psi\equiv\mathds{1}$, and fix $q\in\mathbb{N}_0$. Let $\varphi\in\Ad_q(\mathds{1})$ be an optimal policy for the supply-chain cost minimization problem restricted to the class of admissible MA$(q)$ policies. If $\varphi(z)=\sum_{k=0}^q \varphi_k z^k$, then its average (group) delay is bounded above by the midpoint of the smoothing window:
\begin{equation}\label{eq:group_delay_bound}
\sum_{k=0}^q k\,\varphi_k \leq \frac{q}{2}.
\end{equation}
\end{proposition}

\smallskip

\cref{prop:group_delay_bound} has a useful practical implication: optimal smoothing should not delay the retailer's demand response too much. Although postponing orders can improve the manufacturer's forecastability, excessive back-loading raises the retailer's inventory exposure and is not cost-effective. Thus, the optimal use of a finite smoothing window is disciplined: it may spread the response to demand over time, but it does not place its average response beyond the midpoint of the window. From a managerial perspective, smoothing should be viewed as a calibrated delay, not as a prescription to postpone replenishment as much as possible.

\subsection{Asymptotic Regimes}

Let us return to the general demand case and study the cost minimization problem \eqref{eq:SCCost} in the two limiting regimes $\kappa\uparrow\infty$ and $\kappa\downarrow 0$.\smallskip

{\bfseries\boldmath -) Large-$\kappa$ regime ($\kappa \uparrow \infty$):}
 In this case, the retailer’s cost dominates the supply-chain cost, and an optimal policy is obtained by minimizing the retailer’s inventory volatility $\sigma_{\ti I}(\varphi;\psi)$ defined in \eqref{eq_SigmaI_z}. We note that this policy coincides with the policy a retailer would choose in a decentralized supply chain if it internalized only its own inventory costs and not the effect of its ordering policy on the manufacturer’s forecast uncertainty and replenishment costs. Accordingly, for future reference, we refer to this limiting optimal policy as the {\em myopic policy}.
\smallskip

\begin{proposition}[Myopic Policy]\label{prop:myopic}
The retailer's myopic policy $
\varphi^{\MP}:=\arg\min\bigl\{\sigma_{\ti I}(\varphi;\psi)\colon \varphi \in \Ad\bigr\}$
is given by
\[
\varphi^{\MP}(z)=\frac{\psi(z)-(1-z)\,\psi(0)}{z\,\psi(z)}.
\]
Moreover, under this policy $\sigma_{\ti I}\bigl(\varphi^{\MP};\psi\bigr)=|\psi(0)|$.
\end{proposition}
\smallskip

In the special case in which market demand is i.i.d.\ (e.g., $\psi(z)\equiv \mathds{1}$), the retailer's myopic policy
reduces to $\varphi^{\MP}(z)\equiv \mathds{1}$, implying $O_t=D_t$ for all $t$. Hence, the retailer introduces no
delay in its response to demand, passing demand immediately to the manufacturer.\smallskip

\begin{example}
To illustrate the delayed response to demand embedded in the myopic policy, let us consider two canonical
examples of demand processes: (i) an AR(1) demand with $z$-transform 
$\psi_{\ti{AR}}(z)=\frac{1-\theta}{1-\theta\,z}$ for some $\theta\in(-1,1)$, and
(ii) an MA(1) demand with $z$-transform $\psi_{\ti{MA}}(z)=\psi_0+(1-\psi_0)\,z$ for some $\psi_0>0.5$
to ensure invertibility (see \cref{dfn:invertible}). It follows from \cref{prop:myopic} that the corresponding
myopic policy in each case is given by
\[
\varphi_{\ti{AR}}^{\MP}(z)=(1+\theta)-\theta\,z
\qquad \mbox{and}\qquad
\varphi_{\ti{MA}}^{\MP}(z)=\frac{1}{\psi_0+(1-\psi_0)\,z}.
\]
In other words, the myopic policy associated with an AR(1) demand is an MA(1), and the myopic policy
associated with an MA(1) demand is an AR(1). The corresponding impulse responses are depicted in
\cref{fig:impulse_responses_comparison}.
\input{Figures/Fig_Ex1_AR_MA}

In these examples, the myopic policy satisfies $\varphi^{\MP}_0>1$, meaning that the demand observed in
the current period contributes more than one-for-one to the current order quantity. This occurs when
$\theta>0$ in the AR(1) demand case, and whenever $\psi_0<1$ in the MA(1) demand case.
\end{example}\smallskip

{\bfseries\boldmath -) Small-$\kappa$ regime ($\kappa \downarrow 0$):} In this regime, the manufacturer’s cost dominates the supply-chain objective. Consequently, the limiting optimal policy $\varphi^*$ is characterized by minimizing the manufacturer’s root MSFE $\sigma_{\M}(\varphi;\psi)$ over admissible policies $\varphi\in\Ad$. We first note that $\sigma_{\M}(\varphi;\psi)>0$ for every admissible policy $\varphi\in\Ad(\psi)$. \smallskip

\begin{lemma}\label{lem:positive_sigma_m}
For every policy $\varphi\in\Ad(\psi)$, the manufacturer's root MSFE is strictly positive, $\sigma_{\M}(\varphi;\psi)>0$. Moreover, fix $q\in\mathbb{N}_0$. If $\varphi\in\Ad_q(\psi)$, so that $\varphi(z)=\sum_{n=0}^{q}\varphi_n z^n$, then
\[
\sigma_{\M}(\varphi;\psi)
\geq
|\psi(0)|\max_{0\leq n\leq q}
\left\{\frac{|\varphi_n|}{\binom{q}{n}}\right\}
\geq
\frac{|\psi(0)|}{2^q}.
\]
\end{lemma}\smallskip

While $\sigma_{\M}(\varphi;\psi)$ is strictly positive for each fixed admissible $\varphi$, its infimum over $\Ad(\psi)$ is nevertheless zero. To see this, consider the family of simple moving-average policies
\begin{equation}\label{eq:class_MA}
\Ad^{\ti{MA}}
:=\Bigg\{\varphi_{\tim N}^{\ti{MA}}\in\Ad\;\colon\;
\varphi_{\tim N}^{\ti{MA}}(z)=\frac{1}{N+1}\sum_{n=0}^N z^n
=\frac{1-z^{N+1}}{(N+1)(1-z)},
\quad N\in\mathbb N_0
\Bigg\}.
\end{equation}
Each $\varphi_{\tim N}^{\ti{MA}}$ is invertible (outer), and therefore \cref{cor:invertible} yields
\[
\sigma_{\M}\!\bigl(\varphi_{\tim N}^{\ti{MA}};\psi\bigr)
=|\psi(0)|\,\bigl|\varphi_{\tim N}^{\ti{MA}}(0)\bigr|
=\frac{|\psi(0)|}{N+1}\xrightarrow[N\to\infty]{}0.
\]

However, the convergence $\sigma_{\ti M}(\varphi_k;\psi)\to 0$ does \emph{not} by itself imply that $\varphi_k$ becomes asymptotically optimal as $\kappa$ becomes arbitrarily small. The reason is that making the manufacturer's forecast error very small may require ordering policies that create large swings in the retailer's inventory position. In other words, improving forecastability on the manufacturer side can come at the expense of much greater inventory volatility on the retailer side. The next proposition formalizes this tension and shows that the limit $\kappa\downarrow 0$ is singular in precisely this sense.\smallskip

\begin{proposition}\label{prop:lim_kappa_0}
Let $\{\varphi_k\}_{k\in\mathbb N}\subset \Ad(\psi)$ be a sequence of admissible policies such that the manufacturer’s root MSFE along this sequence satisfies $\sigma_{\ti M}(\varphi_k;\psi)\downarrow 0$ as $k\to\infty$. Then the retailer’s inventory volatility satisfies
\[
\sigma_{\ti I}(\varphi_k;\psi)\to\infty
\qquad\text{as }k\to\infty.
\]
\end{proposition}

Proposition~\ref{prop:lim_kappa_0} states that any sequence of admissible policies that drives the manufacturer’s root MSFE to zero must necessarily incur unbounded inventory volatility for the retailer. In particular, although the \emph{infimum} of $\sigma_{\M}$ over $\Ad(\psi)$ equals zero, it cannot be approached while keeping $\sigma_{\I}$ bounded. Equivalently, the limiting problem obtained by setting $\kappa=0$ has no optimizer in $\Ad$: minimizing sequences ``escape'' by making the retailer’s inventory increasingly volatile.  As a concrete illustration, for the moving-average class \eqref{eq:class_MA} one can compute directly from \eqref{eq_SigmaI_z} that for an i.i.d. demand $\psi\equiv \mathds{1}$
\[
\sigma_{\I}\!\bigl(\varphi_{\tim N}^{\ti{MA}};\mathds{1}\bigr)
= \sqrt{\frac{(N+2)\,(2\,N+3)}{6\,(N+1)}}\xrightarrow[N\to\infty]{}\infty.
\]

\cref{prop:lim_kappa_0} reveals that the asymptotic regime $\kappa\downarrow 0$ requires a careful characterization of the trade-off between $\sigma_{\ti M}(\varphi;\psi)$ and $\sigma_{\ti I}(\varphi;\psi)$. In particular, we will see that the class of simple MA policies $\Ad^{\ti{MA}}$, while satisfying $\lim_{N\to\infty}\sigma_{\ti M}\!\bigl(\varphi_{N}^{\ti{MA}};\psi\bigr)=0$, can nonetheless exhibit cost performance that is arbitrarily worse than that of an optimal policy as $\kappa\downarrow 0$. We investigate this trade-off further in \cref{sec:ordersmooting}, where we introduce our proposed class of smoothing policies and show, in fact, that (see \cref{prop:relerror})
\[
\lim_{\kappa \downarrow 0}\ \inf_{N \in \mathbb{N}_0}\ \frac{\C{C}(\varphi_{N}^{\ti{MA}};\psi;\kappa)}{\C{C}^*(\psi;\kappa)}=\infty.
\]
The proof of this result relies on a lower bound for $\C{C}^*(\psi;\kappa)$, which we derive next.

\smallskip

\subsection{Lower Bounds}\label{sec:lowerbound}
We conclude this section by deriving two complementary lower bounds on the optimal supply-chain cost ${\cal C}^*(\kappa;\psi)$: one that leverages the i.i.d.\ demand benchmark of \citet{caldentey2024information}, and another that applies an information relaxation that enlarges the manufacturer’s information set.

\subsubsection{Bounds Based on i.i.d. Demand}\label{sec:lowerboudiid}

Our first bound relies on \cref{prop:solnosharing} to derive a lower bound on ${\cal C}^*(\kappa;\psi)$ for a general market-demand process $\psi(z)$, expressed in terms of ${\cal C}^*_{\ti{IID}}(\kappa)$, the optimal cost in the special case of i.i.d.\ market demand (e.g., $\psi\equiv \mathds{1}$). 
\smallskip

\begin{proposition}[Bounds from the i.i.d. benchmark]\label{prop:boundsiid}
Let $\psi(z)$ be the market-demand transfer function and define
$\psi_{\inf}=\inf_{\zeta \in \mathbb{T}}|\psi(\zeta)|$. Then, 
\[
{\cal C}^*(\kappa;\psi) \geq |\psi(0)|\,{\cal C}^*_{\ti{IID}}(\kappa_{\inf}),\qquad \mbox{where ~$\kappa_{\inf}:=\dfrac{\kappa\,\psi_{\inf}}{|\psi(0)|}$.}
\]
\end{proposition}

By construction, the bound is trivially tight when $\psi\equiv \mathds{1}$. Moreover, $\kappa_{\inf}$ is well defined since $\psi(z)$ is outer and therefore has no zeros inside the unit circle, implying in particular that $\psi(0)\neq 0$.

\subsubsection{An Information-Relaxation Lower Bound}\label{sec:lowerbound_inf}

We obtain our second class of lower bounds on the supply chain cost ${\cal C}(\varphi;\psi;\kappa)=\kappa\,\sigma_{\ti I}(\varphi;\psi)+\sigma_{\ti M}(\varphi;\psi)$ by applying an information relaxation that enlarges the manufacturer’s information set ${\cal F}^{\ti M}_t=\sigma(O_\tau\colon \tau \leq t)$ with additional information about the market demand process beyond what is contained in the retailer’s order history. This relaxation can only improve the manufacturer’s forecasting problem and therefore weakly reduces the associated forecast error, i.e., it lowers $\sigma_{\ti M}(\varphi;\psi)$. Importantly, the retailer’s ordering policy and inventory dynamics are unchanged, so the inventory variability term $\sigma_{\ti I}(\varphi;\psi)$ is unaffected. Consequently, evaluating ${\cal C}(\varphi;\psi;\kappa)$ under the relaxed information set yields a cost no larger than the original one, and thus provides a valid lower bound on the achievable supply chain cost.\smallskip

We consider the extreme case of {\em full information}, in which we replace ${\cal F}^{\ti M}_t$ by 
${\cal F}^{\ti F}_t:=\sigma\!\bigl(D_\tau, O_\tau:\tau\le t\bigr)$, that is, in addition to observing the retailer’s orders, the manufacturer also observes the history of market demand. Under this information relaxation, the manufacturer has the same information as the retailer and can therefore recover the demand shocks $\{\eps_t\}$. Thus, for a given policy $\varphi(z)$ and using the representation $O_t=d+\tpsi({\cal B})\,\eps_t$, where $\tpsi(z)=\varphi(z)\,\psi(z)$, the manufacturer’s one-step-ahead mean forecast $m_t(\varphi;\psi)=\e[O_{t+1}|{\cal F}^{\ti F}_t]$ and MSFE $\sigma^2_{\ti M}(\varphi;\psi)=\var[O_{t+1}|{\cal F}^{\ti F}_t]$ are given by
$$m_t(\varphi;\psi)=\e\left[d+\sum_{n=0}^\infty \tpsi_n \eps_{t+1-n} \big| {\cal F}^{\ti F}_t\right]= d+\sum_{n=1}^\infty \tpsi_n \eps_{t+1-n}\qquad \mbox{and}\qquad \sigma^2_{\ti M}(\varphi;\psi)=|\psi(0)|^2\,|\varphi(0)|^2.$$

Because ${\cal F}^{\ti M}_t \subseteq {\cal F}^{\ti F}_t$, enlarging the manufacturer's information set can only weakly reduce the one-step-ahead MSFE. Since $\sigma_{\ti I}(\varphi;\psi)$ remains unchanged, this immediately yields a lower bound on the supply chain cost:
$${\cal C}(\varphi;\psi;\kappa) \geq {\cal C}_{\ti F}(\varphi;\psi;\kappa):=\kappa\,\sigma_{\ti I}(\varphi;\psi)
+|\psi(0)|\,|\varphi(0)|.$$
The following result characterizes the infimum ${\cal C}^*_{\ti F}(\psi;\kappa)$ over the class of admissible policies.\medskip

\begin{proposition}[Full Information Lower Bound]\label{prop:lowerbound}
Let $\displaystyle \tpsi_0^*(\kappa)\in \argmin_{\tpsi_0}
\Big\{\kappa\,\sqrt{\psi(0)^2\!+\!\bigl(\rho-\psi(0)-\psi'(0)\bigr)^2}\!+\!|\tpsi_0|\Big\}$.
Then,
\[
{\cal C}^*_{\ti F}(\psi;\kappa):=\inf_{\varphi \in \Ad} {\cal C}_{\ti F}(\varphi;\psi;\kappa)
=\kappa\,\sqrt{\psi(0)^2+\bigl(\tpsi_0^*(\kappa)-\psi(0)-\psi'(0)\bigr)^2}+|\tpsi_0^*(\kappa)|.
\]
Furthermore, the infimum is achieved by the policy
\[
\varphi^{\ti F}(z)=
\frac{\psi(z)-(1-z^2)\,\psi(0)+(1-z)\,z\,\bigl(\tpsi_0^*(\kappa)-\psi'(0)\bigr)}
{z\,\psi(z)}.
\]
\end{proposition}\medskip

It follows from the preceding discussion that the policy $\varphi^{\ti F}(z)$ in \cref{prop:lowerbound} would be optimal for the original optimization problem \eqref{eq:SCCost} if it were invertible. In that case, the manufacturer could recover the market-demand history $\{D_t\}$ from the retailer’s orders, implying that ${\cal F}^{\ti M}_t={\cal F}^{\ti F}_t$. The following corollary is immediate.\medskip

\begin{corollary}
If the retailer's policy $\varphi^{\ti F}(z)$ is invertible (outer), then it is an optimal solution to \eqref{eq:SCCost} and ${\cal C}^*(\psi;\kappa)={\cal C}^*_{\ti F}(\psi;\kappa)$.
\end{corollary}\smallskip

Unfortunately, verifying whether $\varphi^{\ti F}(z)$ is outer is not straightforward. While outerness can be checked by locating the zeros of $\varphi^{\ti F}(z)$ and confirming that they all lie outside the unit disk, this is generally a hard task in practice. Alternatively, one may verify the {\em Jensen's condition}, namely,
\[
\log\bigl|\varphi^{\ti F}(0)\bigr|
=\frac{1}{2\pi}\int_{-\pi}^{\pi}\log\bigl|\varphi^{\ti F}(e^{-i\lambda})\bigr|\,\D\lambda,
\]
which (under the usual integrability conditions) characterizes outerness; see \cite{MartinezAvendanoRosenthal}. In applied settings, however, evaluating the boundary integral may itself be technically delicate: the integrand $\log\bigl|\varphi^{\ti F}(e^{-i\lambda})\bigr|$ can require regularization when $\varphi^{\ti F}$ has zeros on or near the unit circle.\smallskip

{}

Another implication of \cref{prop:lowerbound} is that the myopic policy in \cref{prop:myopic} is asymptotically optimal as $\kappa \uparrow \infty$. To see this, note that the full-information policy $\varphi^{\ti F}(z)$ can be rewritten in terms of the myopic order-up-to policy as
\[
\varphi^{\ti F}(z)=\varphi^{\ti{MP}}(z)+\frac{1-z}{\psi(z)}\big(\tpsi_0^*(\kappa)-\psi(0)-\psi'(0)\big).
\]

\begin{corollary}\label{cor:myopic_asym}
As $\kappa \uparrow \infty$, we have $\tpsi_0^*(\kappa)\to \psi(0)+\psi'(0)$, and hence $\varphi^{\ti F}(z)\to \varphi^{\ti{MP}}(z)$.
\end{corollary}\smallskip

This confirms the intuition that the myopic order-up-to policy should be nearly optimal  when the retailer’s inventory-management costs dominate overall supply costs.


%% file: Figures/Fig_IR_IID_Case.tex

\begin{figure}[h!]
\begin{center}
\begin{tikzpicture}
\begin{groupplot}[
    group style={group size=3 by 1, horizontal sep=0.3cm},
    ybar,
    /pgf/bar width=5pt,
    width=0.374\textwidth,
    height=5.8cm,
    ymin=-0.25, ymax=0.95,
    enlarge x limits=0.02,
    enlarge y limits={upper, value=0.05},
    grid=major,
    xlabel={Delay $n$},
    ticklabel style={font=\small},
    label style={font=\small},
    title style={font=\small, yshift=-1.2ex}
]

\nextgroupplot[
    xmin=-0.5, xmax=10.5,
    xtick={0,2,4,6,8,10},
    ylabel={$\varphi_n$},
    ylabel style={rotate=-90},
    title={\bfseries \sffamily $\kappa=5$}
]
\addplot+[draw=black]
coordinates {
    (0, 0.795876)
    (1, 0.204124)
    (2, 0.000000)
    (3, 0.000000)
    (4, 0.000000)
    (5, 0.000000)
    (6, 0.000000)
    (7, 0.000000)
    (8, 0.000000)
    (9, 0.000000)
    (10, 0.000000)
};

\nextgroupplot[
    xmin=-0.5, xmax=10.5,
    xtick={0,2,4,6,8,10},
    ylabel={},
    yticklabels={},
    title={\bfseries \sffamily $\kappa=1$}
]
\addplot+[draw=black]
coordinates {
    (0,  0.189859)
    (1,  0.557532)
    (2,  0.356046)
    (3, -0.126164)
    (4,  0.007612)
    (5,  0.042749)
    (6, -0.054123)
    (7,  0.045646)
    (8, -0.029296)
    (9,  0.012094)
    (10, 0.002288)
};

\nextgroupplot[
    width=0.404\textwidth,
    xmin=-0.5, xmax=20.5,
    xtick={0,5,10,15,20},
    ylabel={},
    yticklabels={},
    title={\bfseries \sffamily $\kappa=0.01$}
]
\addplot+[draw=black, /pgf/bar width=4.37pt]
coordinates {
    (0,  0.001212)
    (1,  0.015807)
    (2,  0.087896)
    (3,  0.265002)
    (4,  0.444770)
    (5,  0.342817)
    (6, -0.048334)
    (7, -0.203077)
    (8,  0.054423)
    (9,  0.117495)
    (10, -0.096292)
    (11, -0.009700)
    (12, 0.069330)
    (13, -0.044917)
    (14, -0.009401)
    (15, 0.039670)
    (16, -0.023398)
    (17, -0.006880)
    (18, 0.022839)
    (19, -0.012490)
    (20, -0.004388)
};

\end{groupplot}
\end{tikzpicture} \vspace{-0.5cm}

\caption{Illustration of the retailer's impulse response for three values of $\kappa$. 
For $\kappa=5\ge \sqrt{5}$, the optimal policy is
$\varphi^*(z)=\varphi_0^*+(1-\varphi_0^*)z$ with
$\varphi_0^*=1-1/\sqrt{\kappa^2-1}=1-1/\sqrt{24}$.
For $\kappa=1$ and $\kappa=0.01$, the bars show the first coefficients of the outer policy
$\varphi_{\delta}(z)$ with $\delta=0.0001$.}
\label{fig:iid-impulse-three-cases}
\end{center}
\end{figure}
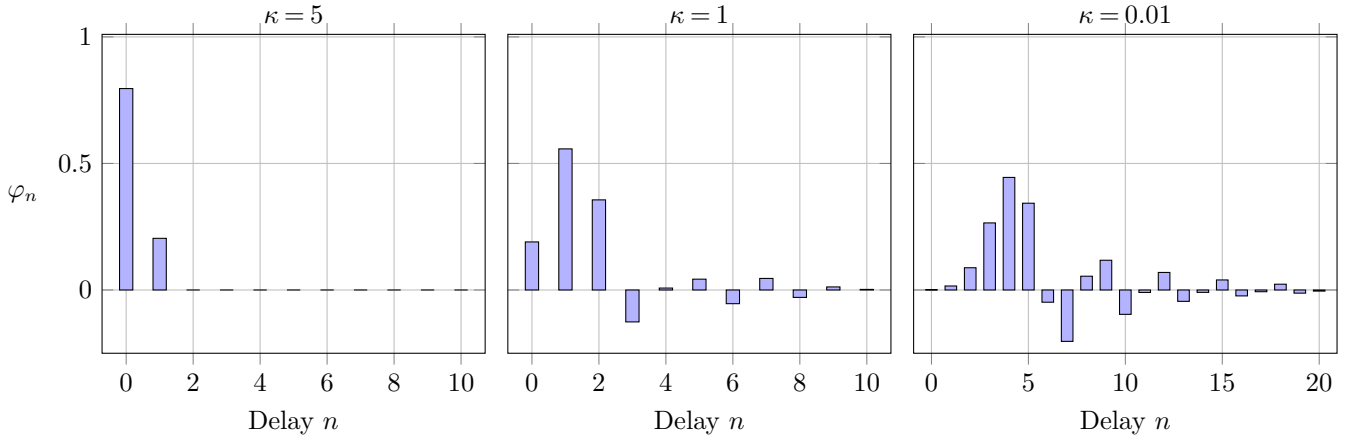

%% file: Figures/Fig_Ex1_AR_MA.tex
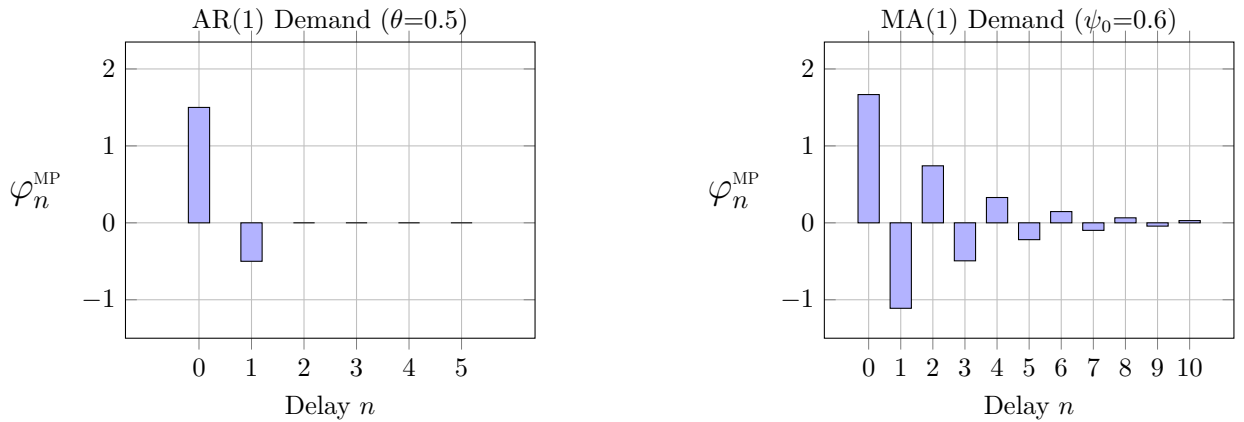
\begin{figure}[h]
\begin{center}
\begin{minipage}{0.48\textwidth}
\centering
\begin{tikzpicture}
  \begin{axis}[
    ybar,
    /pgf/bar width=8pt,
    ymin=-1.5,
    ymax=2,
    xmin=-0.5,
    xmax=5.5,
    xlabel={Delay $n$},
    ylabel={\large $\varphi^{\MP}_n$},
    ylabel style={rotate=-90},
    xtick={0,1,2,3,4,5},
    enlarge x limits=0.15,
    enlarge y limits={upper, value=0.1},
    grid=major,
    width=7cm,
    height=5.5cm,
    ticklabel style={font=\small},
    label style={font=\small},
    yticklabel style={/pgf/number format/fixed},
    title={AR(1) Demand ($\theta{=}0.5$)},
    title style={yshift=-1.5ex, font=\small},
  ]
  \addplot+[draw=black] coordinates {
    (0, 1.5)
    (1, -0.5)
    (2, 0)
    (3, 0)
    (4, 0)
    (5, 0)
  };
  \end{axis}
\end{tikzpicture}
\end{minipage}%
\hfill
\begin{minipage}{0.48\textwidth}
\centering
\begin{tikzpicture}
  \begin{axis}[
    ybar,
    /pgf/bar width=8pt,
    ymin=-1.5,
    ymax=2,
    xmin=-0.5,
    xmax=10.5,
    xlabel={Delay $n$},
    ylabel={\large $\varphi^{\MP}_n$},
    ylabel style={rotate=-90},
    xtick={0,1,2,3,4,5,6,7,8,9,10},
    enlarge x limits=0.08,
    enlarge y limits={upper, value=0.1},
    grid=major,
    width=7cm,
    height=5.5cm,
    ticklabel style={font=\small},
    label style={font=\small},
    yticklabel style={/pgf/number format/fixed},
    title={MA(1) Demand ($\psi_0{=}0.6$)},
    title style={yshift=-1.5ex, font=\small},
  ]
  \addplot+[draw=black] coordinates {
    (0, 1.667)
    (1, -1.111)
    (2, 0.741)
    (3, -0.494)
    (4, 0.329)
    (5, -0.219)
    (6, 0.146)
    (7, -0.098)
    (8, 0.065)
    (9, -0.043)
    (10, 0.029)
  };
  \end{axis}
\end{tikzpicture}
\end{minipage}

\caption{Impulse responses: (left) $\varphi(z) = (1{+}\theta) - \theta z$ with $\theta = 0.5$; (right) $\varphi(z) = 1/(\psi_0 + (1{-}\psi_0)z)$ with $\psi_0 = 0.6$.}
\label{fig:impulse_responses_comparison}
\end{center}
\end{figure}

%% file: Management_Science/ApproximationPolicies_MS.tex
\section{$\alpha$-Approximation  Policies}\label{sec:approximations}

In this section, we introduce the notion of an $\alpha$-approximation class as a performance benchmark for the supply-chain cost in \eqref{eq:SCCost}. Because exact optimal policies are typically difficult to characterize in a tractable and implementable form, this notion provides a uniform  multiplicative guarantee relative to an optimal admissible policy, valid for all $\kappa$, and thereby offers a principled criterion for evaluating and comparing implementable policy classes in the sections that follow.
\smallskip

\begin{definition}[$\alpha$-Approximation Class of Policies]\label{def:alpha_approx}
A subclass $\Ad'\subseteq\Ad(\psi)$ of admissible inventory policies is an \emph{$\alpha$-approximation class of policies} (uniformly in $\kappa$) if there exists a constant $\alpha \geq 1$, independent of $\kappa$, such that
\begin{equation}\label{eq:alpha_approx_def}
\inf_{\varphi \in \Ad'} {\cal C}(\varphi;\kappa) \leq \alpha \,{\cal C}^*(\kappa),
\qquad\text{for all }\kappa \in (0,\infty).
\end{equation}
\end{definition}\smallskip

This definition captures the idea of a \emph{uniformly reliable} class of simple policies. Since the full optimization over $\Ad(\psi)$ is infinite-dimensional and may yield policies that are difficult to characterize or implement, it is natural to restrict attention to a tractable subclass $\Ad'$. Condition \eqref{eq:alpha_approx_def} ensures that this restriction does not incur an uncontrolled loss: for every $\kappa$, the best policy in $\Ad'$ achieves a cost within a fixed multiplicative factor $\alpha$ of the true optimum. Importantly, the bound is \emph{uniform in $\kappa$}, so the same class remains near-optimal across all relative weightings of downstream inventory costs and upstream forecast-error costs, providing a robustness benchmark for implementable policy design.\smallskip

Even though $\alpha$-approximation is, by construction, a coarse notion of near-optimality, it is still strong enough to impose substantive structural restrictions on a class of inventory policies $\Ad'$. In particular, requiring a guarantee that holds uniformly over all $\kappa$ rules out policy families that are empirically appealing or easy to implement but whose performance can deteriorate arbitrarily on some instances of the problem. As we will show, this includes several widely used “smoothing” heuristics proposed in the literature, most notably simple moving-average and exponential-smoothing rules, which deliberately dampen order variability by filtering demand signals rather than explicitly optimizing the underlying cost tradeoff. While such policies may behave reasonably in certain parameter regimes, their fixed filtering structure can create instances in which the resulting delay generates costs that are arbitrarily larger than the true optimum; consequently, no finite $\alpha$ can certify them as uniformly reliable in the sense of \eqref{eq:alpha_approx_def}.\smallskip

Under the regularity requirement in \eqref{eq:boundedpsi}, our first result shows that $\alpha$-approximation for a general market-demand process $\psi$ is essentially equivalent, up to a condition-number factor, to $\alpha$-approximation in the i.i.d. demand benchmark. \smallskip

\begin{proposition}\label{prop:iid_reduction}
Suppose the market demand $\psi(z)$ satisfies \eqref{eq:boundedpsi}. Then there exists an $\alpha \geq 1$ such that the class of inventory policies $\Ad'$ is an $\alpha$-approximation class for the market demand $\psi$ if and only if there exists an $\tilde{\alpha}\geq 1$ such that $\Ad'$ is an $\tilde{\alpha}$-approximation class for i.i.d. market demand.
\end{proposition}\smallskip

Equipped with \cref{prop:iid_reduction}, we may therefore focus on characterizing policy classes that are $\alpha$-approximation classes in the benchmark case of i.i.d. demand.  This reduction lets us work in a simplified setting while retaining generality up to a constant factor, and it allows us to leverage \cref{prop:solnosharing}, which provides a precise characterization of the optimal cost under i.i.d.\ demand.

%% file: Management_Science/Binomial_MS.tex
\section{Binomial Smoothing Policies}\label{sec:ordersmooting}

In this section, we propose a novel class of inventory policies aimed at approximately minimizing the supply-chain cost in \eqref{eq:SCCost}. Our goal is to identify an $\alpha$-approximation class (cf.\ \cref{def:alpha_approx}) that delivers a uniform performance guarantee relative to an optimal admissible policy, uniformly over all $\kappa$. At the same time, we seek policies that are analytically tractable, operationally interpretable, and straightforward to implement using standard forecasting and ordering primitives, thereby overcoming the practical limitations of the $\eps$-optimal policies characterized in \cref{prop:solnosharing}. \smallskip

The construction of our proposed order-smoothing policies is guided by the {\em desiderata} that the class attain near-optimal performance in the two limiting regimes, $\kappa \downarrow 0$ and $\kappa \uparrow \infty$; accordingly, we select minimal-structure policies that satisfy both requirements, thereby underpinning the desired uniform (in $\kappa$) $\alpha$-approximation guarantee.\smallskip

First, by \cref{cor:myopic_asym}, it follows that when $\kappa \uparrow \infty$, an optimal policy converges to the myopic policy $\varphi^{\MP}(z)$. Accordingly, we seek a family of policies that includes the myopic policy as a member. \smallskip

In the opposite regime, when $\kappa \downarrow 0$, we would like our class of policies to satisfy the following design principles:

\begin{enumerate}[(a)]
\item {\sc Minimal MSFE:} To achieve near-optimality for small $\kappa$ (i.e., when the manufacturer's cost predominates in the overall supply-chain cost relative to the retailer's), we must prioritize policies that, {\em ceteris paribus}, minimize the manufacturer's MSFE, thereby reducing its associated cost.

\item {\sc Structural Simplicity:} To keep the proposed policy class transparent and easy to interpret (for example, in terms of delayed demand response), we restrict attention to policies $\varphi$ in the class ${\cal A}_q \subset {\cal A}(\psi)$ of admissible policies that admit a finite MA$(q)$ representation, $\varphi(z)=\sum_{n=0}^q \varphi_n z^n$, for some integer $q \in \mathbb{N}_0=\{0,1,2,\dots\}$. The order $q$ serves as a tuning parameter that controls the degree of delay and smoothing.\footnote{
Note that by \cref{prop:solnosharing}, when $\kappa \geq \sqrt{5}$, the optimal policy is an MA(1), so restricting to the class of MA($q$) policies entails no loss of optimality over this range of $\kappa$ values.}
\end{enumerate}\vspace{0.2cm}

As we see next, imposing these two requirements fully characterizes a specific class of smoothing policies. To this end, we introduce the following definition.\smallskip 

\begin{definition}{\rm (Binomial Smoothing)}
The class of {\rm Binomial Smoothing Policies (BN)} is given by
$$
\Ad^{\BN}
=\Bigg\{\varphi^{\BN}_q \in \Ad_q \;\colon\; 
\varphi^{\BN}_q(z)=\left(\frac{1+z}{2}\right)^q
=\frac{1}{2^q}\sum_{n=0}^q \binom{q}{n}\, z^n
\quad \text{for } q\in\mathbb{N}_0
\Bigg\}.
$$
\end{definition}\smallskip

It is worth noting that the class $\Ad^{\BN}$ of Binomial Smoothing policies is invertible in the sense of \cref{dfn:invertible}.\smallskip 

\begin{proposition}\label{prop:BN_min_MSFE}
Let $\Ad_q(\psi)\subseteq \Ad(\psi)$ denote the class of admissible policies $\varphi$ that admit an MA$(q)$ representation. Consider the problem of minimizing the manufacturer's root MSFE over the class $\Ad_q(\psi)$, 
$$
\inf_{\varphi \in \Ad_q(\psi)}\; \sigma_{\ti M}(\varphi;\psi).
$$
Then an optimal solution is given by the Binomial Smoothing policy $\varphi^{\BN}_q(z)$. 
\end{proposition}\smallskip

Proposition~\ref{prop:BN_min_MSFE} identifies Binomial Smoothing as the natural benchmark within the class of finite-history smoothing policies. Its connection to \cref{lem:positive_sigma_m} provides useful intuition for why the binomial weights arise. The lemma shows that, within $\Ad_q(\psi)$, the manufacturer's root MSFE is controlled from below by the largest coefficient after normalization by the corresponding binomial coefficient. Hence, reducing forecast error requires spreading the coefficients as evenly as possible relative to the binomial profile. Since admissibility requires $\varphi(1)=\sum_{n=0}^q \varphi_n=1$, this balance is achieved by choosing $\varphi_n=2^{-q}\binom{q}{n}$, which is precisely the Binomial Smoothing policy. Thus, \cref{prop:BN_min_MSFE} can be viewed as the sharp version of the lower-bound intuition in \cref{lem:positive_sigma_m}: among all MA$(q)$ policies, Binomial Smoothing allocates the unit demand response across the $q+1$ periods in the most forecastable way, while retaining a transparent finite delayed-response structure.\smallskip

As an immediate consequence of \cref{prop:BN_min_MSFE}, the optimality of Binomial Smoothing within $\Ad_q(\psi)$ admits an equivalent interpretation in terms of a forecast-adjusted notion of the bullwhip effect.\smallskip

\begin{corollary}\label{cor:BN_forecast_adjusted_bullwhip}
For a policy $\varphi\in\Ad_q(\psi)$, define its forecast-adjusted bullwhip effect by
\begin{equation}\label{eq:forecast_adjusted_bullwhip}
\mathscr{B}_{\ti{FA}}(\varphi;\psi)
:=
\frac{\var(O_{t+1}\mid {\cal F}^{\ti M}_t)}
{\var(D_{t+1}\mid {\cal F}^{\ti D}_t)}
=
\frac{\sigma_{\M}^2(\varphi;\psi)}{|\psi(0)|^2},
\end{equation}
where ${\cal F}^{\ti D}_t:=\sigma(D_\tau:\tau\leq t)$ denotes the demand history. Then, among all admissible MA$(q)$ policies, the Binomial Smoothing policy minimizes $\mathscr{B}_{\ti{FA}}(\varphi;\psi)$. In particular, $\displaystyle \min_{\varphi\in\Ad_q(\psi)}
\mathscr{B}_{\ti{FA}}(\varphi;\psi)
=
\mathscr{B}_{\ti{FA}}(\varphi_q^{\BN};\psi)
=
2^{-2q}.$
\end{corollary}

\smallskip
Within a fixed finite-memory class, Binomial Smoothing minimizes the unpredictable component of the order stream relative to the unpredictable component of demand. Thus, under a forecast-adjusted bullwhip measure, the binomial profile is the most effective MA$(q)$ smoothing rule for reducing the manufacturer's forecasting burden.\smallskip

In the remainder of this section, we study the performance of this policy class in greater detail. Motivated by \cref{prop:iid_reduction}, we begin with the i.i.d.\ demand case, which provides the cleanest setting for understanding the structure and performance of Binomial Smoothing. We then turn to general weakly stationary demand processes and show how the analysis extends beyond the i.i.d.\ benchmark.

\subsection{IID Demand}\label{sec:smooth_iid}

When demand is i.i.d., that is, when $\psi\equiv \mathds{1}$, the myopic policy belongs to $\Ad^{\BN}$ as the special case $q=0$; specifically, $\varphi^{\MP}(z)\equiv \mathds{1}=\varphi^{\BN}_0(z)\in\Ad^{\BN}$. Therefore, under i.i.d.\ demand, the class $\Ad^{\BN}$ satisfies our asymptotic design requirements: as $\kappa \uparrow \infty$, it includes the myopic policy, whereas as $\kappa \downarrow 0$, its higher-order members can be selected to minimize the manufacturer's MSFE while generating a delayed and smoothed order response. In this sense, $\Ad^{\BN}$ provides a parsimonious family that bridges these two limiting regimes and yields the desired uniform (in $\kappa$) $\alpha$-approximation guarantee.\smallskip

The next theorem characterizes the supply-chain cost induced by Binomial Smoothing policies and establishes the corresponding uniform $\alpha$-approximation guarantee.\smallskip

\begin{theorem}\label{thm:BN_cost_alpha} 
Under a Binomial Smoothing policy $\varphi^{\BN}_q(z)$, the associated volatilities of the retailer's inventory and the manufacturer's forecast error satisfy
\[
\sigma^2_{\ti I}(\varphi^{\BN}_q;\mathds{1})=\frac{q+2}{2}-\frac{q}{2^{2q+1}}\binom{2q}{q},
\qquad
\sigma^2_{\ti M}(\varphi^{\BN}_q)=2^{-2q}.
\]
Thus, for a given $\kappa$, an optimal Binomial Smoothing policy solves
\begin{equation}\label{eq:costBN}
{\cal C}_{\iid}^{\ti{BN}}(\kappa)
=\min_{q\in\mathbb{N}_0}
\left\{
\kappa\,\sqrt{\frac{q+2}{2}-\frac{q}{2^{2q+1}}\binom{2q}{q}}
+\frac{1}{2^q}
\right\}.
\end{equation}

Moreover, the class $\Ad^{\BN}$ of Binomial Smoothing policies is an $\alpha$-approximation class for \eqref{eq:SCCost_IID}, uniformly over $\kappa>0$. In the i.i.d.\ demand case, one may take $\alpha=\alpha^{\BN}:=1/\sqrt{\log(2)}$. 
\end{theorem}
\smallskip

In terms of delayed demand response, the retailer’s orders under a Binomial Smoothing policy $\varphi^{\BN}_q(z)$ are given by
\[
O_t=\frac{1}{2^q}\sum_{n=0}^q \binom{q}{n} D_{t-n}.
\]
Thus, the BN policy $\varphi^{\BN}_q$ spreads the replenishment of market demand over the next $q$ periods, with weights proportional to the binomial coefficients $\binom{q}{n}$, as depicted in
\cref{fig:impulseresponse_binomial}.
\input{Figures/Fig_Impulse_Binomial}

Interestingly, the Binomial Smoothing policy exhibits a distinctive delayed-demand response: its impulse-response coefficients $\{\varphi^{\BN}_{q,n}\}_{n\ge 0}$ are small for recent lags, then increase and peak at $n=\lfloor q/2 \rfloor$, and subsequently decrease for $n>\lceil q/2 \rceil$. This contrasts with the benchmark policies considered in the next section, which react more quickly to market demand (see \cref{fig:impulseresponsebench} for comparison).\smallskip

By construction, the BN policy minimizes the manufacturer’s MSFE within the class of admissible MA($q$) policies. We interpret this optimality not as a delay in information sharing---indeed, because the ordering policy is invertible, the manufacturer can infer demand from orders without delay---but rather as the selection of a smoothing profile that optimally trades off information and responsiveness. Specifically, increasing $q$ makes the order stream progressively less reactive to short-run swings in demand, which reduces the manufacturer’s root MSFE (with $\sigma_{\ti M}(\varphi^{\BN}_q;\mathds{1})\downarrow 0$ as $q\to\infty$), while simultaneously increasing the effective lag in the retailer’s replenishment response and thus raising inventory costs for the retailer. In this sense, the ``optimal smoothing'' induced by $\varphi^{\BN}_q$ balances the manufacturer’s ability to predict the retailer’s orders against the operational cost of delaying inventory replenishment.\smallskip

The following corollary follows immediately from the binomial form of $\varphi_q^{\BN}$ in \cref{thm:BN_cost_alpha} and complements \cref{prop:group_delay_bound}. Recall that \cref{prop:group_delay_bound} shows that, under i.i.d.\ demand, an optimal MA$(q)$ smoothing rule does not place its average demand response beyond the midpoint $q/2$ of its smoothing window.\smallskip

\begin{corollary}\label{cor:BN_group_delay}
For each $q\in\mathbb{N}_0$, the Binomial Smoothing policy $\varphi_q^{\BN}(z)$ has group delay exactly equal to the midpoint of its smoothing window:
\begin{equation}\label{eq:BN_group_delay}
\sum_{n=0}^q n\,\varphi_{q,n}^{\BN}
=
\frac{1}{2^q}\sum_{n=0}^q n \binom{q}{n}
=
\frac{q}{2}.
\end{equation}
\end{corollary}
\smallskip

\cref{cor:BN_group_delay} shows that Binomial Smoothing attains the boundary identified in \cref{prop:group_delay_bound}. Thus, in terms of average delay, the binomial profile uses the finite smoothing window as fully as an optimal MA$(q)$ policy can, while avoiding excessive delay in the retailer's demand response.

%% file: Figures/Fig_Impulse_Binomial.tex
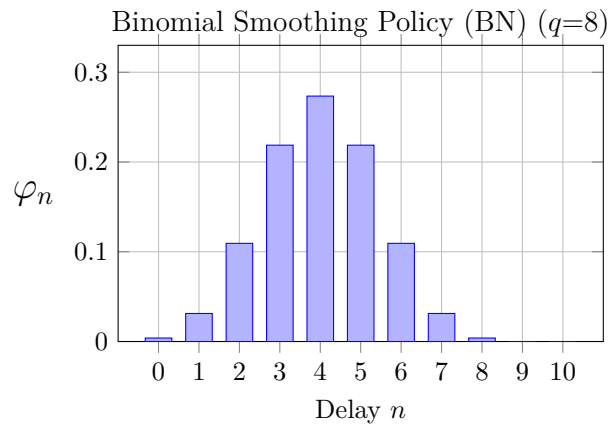
\begin{figure}[h]
\begin{center}
\begin{tikzpicture}
  \begin{axis}[
    ybar,
    /pgf/bar width=10pt, 
    ymin=0,
    ymax=0.3,
    xmax=10,
    xlabel={Delay $n$},
    ylabel={\large $\varphi_n$},
    ylabel style={rotate=-90},
    xtick={0,...,10},
    enlarge x limits=0.1,
    enlarge y limits={upper, value=0.1},
    grid=major,
    width=8cm,
    height=5.5cm,
    ticklabel style={font=\small},
    label style={font=\small},
    yticklabel style={/pgf/number format/fixed}, 
    every axis plot/.append style={fill=gray!40},
    title={Binomial Smoothing Policy (BN) ($q{=}8$)},
    title style={yshift=-1.5ex},
  ]
  \addplot+[] coordinates {
    (0, 0.00391)
    (1, 0.03125)
    (2, 0.10938)
    (3, 0.21875)
    (4, 0.27344)
    (5, 0.21875)
    (6, 0.10938)
    (7, 0.03125)
    (8, 0.00391)
    (9, 0)
    (10, 0)
  };
  \end{axis}
\end{tikzpicture}\vspace{-0.1cm}
\caption{Impulse response of the Binomial Policy with $q = 8$.}
\label{fig:impulseresponse_binomial}
\end{center}
\end{figure}

%% file: Management_Science/Benchmark-Analysis_MS.tex
\subsubsection{Benchmark Analysis}\label{sec:benchmark}

To assess the performance of our proposed Binomial Smoothing (BN) class, we compare it against two standard order-smoothing benchmarks: the Simple Moving Average (MA) and Exponential Smoothing (ES) policies studied in \cite{BGP2004}. These policies provide natural benchmarks because they smooth the retailer's order stream by spreading the response to a demand shock over time. The detailed expressions for their corresponding values of \(\sigma_{\ti M}\), \(\sigma_{\ti I}\), and the optimized benchmark costs are reported in \ref{app:benchmark_details}.

\medskip

{\bf Simple Moving Average (MA):}
The class of simple moving-average policies is
\[
\Ad^{\ti{MA}}
:=
\Bigg\{
\varphi_{\tim N}^{\ti{MA}} \in \Ad
\;\colon\;
\varphi_{\tim N}^{\ti{MA}}(z)
=
\frac{1}{N+1}\sum_{n=0}^{N} z^n
=
\frac{1-z^{N+1}}{(N+1)(1-z)},
\quad
N\in\mathbb{N}_0
\Bigg\}.
\]
The parameter \(N\) determines the length of the delayed demand response. The response is spread uniformly over the \(N+1\) periods \(0,1,\ldots,N\). Under i.i.d. demand, the myopic order-up-to policy is the special case \(N=0\), for which \(\varphi_{0}^{\ti{MA}}(z)=\mathds{1}\).

\medskip

{\bf Exponential Smoothing (ES):}
The class of exponential-smoothing policies is
\[
\Ad^{\ti{ES}}
:=
\Bigg\{
\varphi_{\tim \theta}^{\ti{ES}} \in \Ad
\;\colon\;
\varphi_{\tim \theta}^{\ti{ES}}(z)
=
(1-\theta)\sum_{n=0}^{\infty}(\theta z)^n
=
\frac{1-\theta}{1-\theta z},
\quad
\theta\in[0,1)
\Bigg\}.
\]
The parameter \(\theta\) controls the persistence of the delayed response. Larger values of \(\theta\) place more weight on later periods and therefore induce greater smoothing. The myopic policy is again included as a special case: when \(\theta=0\), \(\varphi_{0}^{\ti{ES}}(z)=\mathds{1}\).

\medskip

\cref{fig:impulseresponsebench} illustrates the impulse response functions of the benchmark MP, MA, and ES policies. In this example, the MA policy uses a smoothing window of \(N=4\), while the ES policy uses \(\theta=0.7\). The figure highlights the different ways these benchmark policies delay the retailer's response to a demand shock: MA spreads the response uniformly over a finite window, while ES uses geometrically decaying weights.

\input{Figures/Fig_Impulse-Response_Benchmarks}

The choice of \(N\) in the MA class and \(\theta\) in the ES class directly controls the extent of the delayed demand response embedded in the retailer's orders. This makes both benchmark classes simple and operationally transparent. However, the same one-parameter structure also limits their ability to balance downstream inventory risk and upstream forecastability uniformly across cost regimes. In particular, even when their smoothing parameters are optimized for each value of \(\kappa\), these benchmark classes do not provide a uniform approximation guarantee.

To formalize this limitation, for any policy class \(\widehat{\Ad}\), define its relative performance by
\begin{equation}\label{eq:relerror}
{\cal E}_{\iid}^{\widehat{\Ad}}(\kappa)
:=
\inf_{\varphi\in\widehat{\Ad}}
\frac{
{\cal C}_{\iid}(\varphi;\kappa)
}{
{\cal C}_{\iid}^{*}(\kappa)
}.
\end{equation}
For a single policy or benchmark \(j\), we write \({\cal E}_{\iid}^{j}(\kappa)\) for the corresponding relative performance ratio. For the MA and ES classes, the infimum in \eqref{eq:relerror} is taken over \(N\in\mathbb{N}_0\) and \(\theta\in[0,1)\), respectively. The optimized cost formulas used to compute these ratios are given in \ref{app:benchmark_details}.\smallskip

\begin{proposition}\label{prop:relerror}
Let \({\cal E}_{\iid}^{\ti{MA}}(\kappa)\) and \({\cal E}_{\iid}^{\ti{ES}}(\kappa)\) denote the relative performance ratios for the simple moving-average and exponential-smoothing policy families, \(\Ad^{\ti {MA}}\) and \(\Ad^{\ti {ES}}\), respectively. Then
\[
\lim_{\kappa \downarrow 0} {\cal E}_{\iid}^{\ti{MA}}(\kappa)
=
\lim_{\kappa \downarrow 0} {\cal E}_{\iid}^{\ti{ES}}(\kappa)
=
\infty.
\]
Consequently, neither the simple moving-average family nor the exponential-smoothing family constitutes an \(\alpha\)-approximation class of inventory policies for any finite \(\alpha\).
\end{proposition}

\smallskip

\cref{prop:relerror} shows that both MA and ES can perform arbitrarily poorly as \(\kappa\downarrow 0\), i.e., when the total supply-chain cost becomes increasingly dominated by the manufacturer's forecasting component. This highlights a key limitation of generic smoothing rules: reducing order variability is not sufficient to guarantee good performance unless the smoothing pattern also preserves useful demand information for the manufacturer.

To assess the magnitude of this deterioration away from the asymptotic regime, \cref{table:suboptbenchmarks_noinfonew} reports the relative performance ratios \({\cal E}_{\iid}^{j}(\kappa)\) over a range of \(\kappa\) values for the myopic policy (MP), the optimized MA and ES benchmarks, and the proposed BN class.

\begin{table}[!h]
\centering
\setlength{\tabcolsep}{10pt}
\begin{tabular}{c c c c c}
\toprule
\(\kappa\) & MP & MA & ES & BN  \\
\midrule
\rowcolor{gray!20}
0.01 & 46.214 & 2.811 & 3.198 & 1.078 \\
0.1  & 5.842  & 1.569 & 1.748 & 1.051 \\
\rowcolor{gray!20}
0.5  & 1.839  & 1.146 & 1.226 & 1.027 \\
1    & 1.327  & 1.049 & 1.093 & 1.012 \\
\rowcolor{gray!20}
5    & 1.017  & 1.017 & 1.001 & 1.017 \\
10   & 1.005  & 1.005 & 1.000 & 1.005 \\
\bottomrule
\end{tabular}
\vspace{0.2cm}
\caption{Relative performance \({\cal E}_{\iid}^{j}(\kappa)\), as defined in \eqref{eq:relerror}, of benchmark policies as a function of \(\kappa\).}
\label{table:suboptbenchmarks_noinfonew}
\end{table}

As shown in \cref{table:suboptbenchmarks_noinfonew}, BN dominates the benchmark policies for \(\kappa\leq 1\). This is precisely the regime in which the manufacturer's forecasting component is most important. For example, when \(\kappa=0.01\), the best MA and ES policies incur costs \(2.811\) and \(3.198\) times the optimum, respectively, while BN remains within roughly \(7.8\%\) of optimality. The myopic policy performs especially poorly in this regime, with relative performance \(46.214\).

For larger values of \(\kappa\), all policies become nearly optimal because the retailer's inventory component receives greater weight. In this region, ES can slightly outperform the integer-parameter BN class because ES is parameterized by a continuous smoothing parameter \(\theta\), whereas the baseline BN class is indexed by an integer \(q\in\mathbb{N}_0\). As discussed in \cref{sec:extensions}, a simple refinement of BN removes this granularity disadvantage and yields a modified BN policy that uniformly dominates these benchmarks.

Overall, the comparison highlights the importance of controlling the information content embedded in replenishment orders. While MA and ES smooth orders through simple one-parameter rules, \cref{prop:relerror} shows that these rules do not provide uniform approximation guarantees. By contrast, BN uses a calibrated delay that preserves demand information in the order stream and enables the manufacturer to extract it through optimal forecasting. In this sense, BN provides an operational mechanism for supply-chain coordination without explicit collaboration.

{}

%% file: Figures/Fig_Impulse-Response_Benchmarks.tex
{}

\begin{figure}[h]
\begin{center}
\begin{tikzpicture}
  \begin{groupplot}[
    group style={
      group size=3 by 1,
      horizontal sep=1.2cm 
    },
    ybar,
    /pgf/bar width=4pt,
    ymin=0,
    ymax=1, 
    xmax=10,
    xlabel={Delay $n$},
    ylabel={\large $\varphi_n$},
    ylabel style={rotate=-90},
    xtick={0,...,10},
    enlarge x limits=0.1,
    enlarge y limits={upper, value=0.1},
    grid=major,
    width=5.5cm,
    height=5cm,
    ticklabel style={font=\small},
    label style={font=\small},
    yticklabel style={/pgf/number format/fixed}, 
    every axis plot/.append style={fill=gray!40},
  ]

  \nextgroupplot[
    title={Myopic Policy},
    title style={yshift=-1.5ex}
  ]
  \addplot+[/pgf/bar width=8pt] coordinates {
    (0, 1)
  };

  \nextgroupplot[
    title={Moving Average ($N{=}4$)},
    title style={yshift=-1.5ex},
    ylabel={}
  ]
  \addplot+[/pgf/bar width=8pt] coordinates {
    (0, 0.2)
    (1, 0.2)
    (2, 0.2)
    (3, 0.2)
    (4, 0.2)
  };

  \nextgroupplot[
    title={Exp. Smoothing ($\theta{=}0.7$)},
    title style={yshift=-1.5ex},
    ylabel={}
  ]
  \addplot+[/pgf/bar width=8pt] coordinates {
    (0, 0.3)
    (1, 0.21)
    (2, 0.147)
    (3, 0.1029)
    (4, 0.07203)
    (5, 0.05042)
    (6, 0.03529)
    (7, 0.0247)
    (8, 0.01729)
    (9, 0.0121)
    (10, 0.00847)
  };

  \end{groupplot}
\end{tikzpicture}\vspace{-0.1cm}
\caption{Impulse response functions for the three benchmark policies.}
\label{fig:impulseresponsebench}
\end{center}
\end{figure}
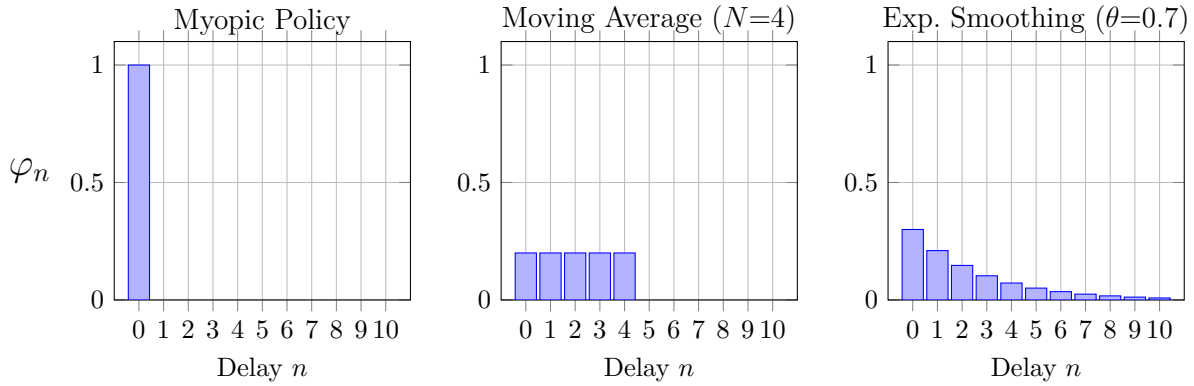

%% file: Management_Science/Non-IID-Demand.tex
\subsection{Binomial Smoothing for Weakly Stationary Market Demand}\label{sec:noniid}

Let us return to the case of a general weakly stationary market demand, as defined in \eqref{eq:def_demand}, and investigate the optimization problem in \eqref{eq:SCCost}. In contrast to the i.i.d. demand case of the previous section, the existing literature offers limited guidance on optimal policy structure beyond heuristics such as those in \cref{sec:benchmark}. To partially fill this gap, we extend our proposed class of Binomial Smoothing policies and numerically assess their performance across a range of canonical demand models, yielding promising results.\smallskip

To this end, we first show that the performance guarantee in \cref{thm:BN_cost_alpha}, originally derived for i.i.d. demand, extends to general weakly stationary demand. In particular, combining \cref{prop:iid_reduction} and \cref{thm:BN_cost_alpha} yields the following corollary.\smallskip

\begin{corollary}\label{cor:BN_alpha_ws}
Suppose the demand $z$-transform $\psi$ satisfies \eqref{eq:boundedpsi}. Then the class $\Ad^{\BN}$ of Binomial Smoothing policies is an $\alpha$-approximation class for \eqref{eq:SCCost}, uniformly over $\kappa>0$, with
\[
\alpha=\alpha^{\BN}\,\frac{\psi_{\sup}}{\psi_{\inf}}
=\frac{1}{\sqrt{\log(2)}}\,\frac{\psi_{\sup}}{\psi_{\inf}}.
\]
\end{corollary}\smallskip

While \cref{cor:BN_alpha_ws} provides a robust, uniform performance guarantee, in the small-$\kappa$ regime one can make a sharper statement. When $\kappa\downarrow 0$, the manufacturer’s cost dominates the supply-chain objective, so the primary design goal is to reduce the manufacturer’s forecast error. In particular, within any fixed MA$(q)$ class, \cref{prop:BN_min_MSFE} implies that the Binomial policy $\varphi_q^{\BN}$ is the \emph{exact} minimizer of the manufacturer’s root MSFE, and therefore it is the natural candidate for asymptotic optimality as $\kappa\downarrow 0$. The following proposition formalizes this intuition by showing that, among admissible MA$(q)$ policies, no alternative policy can improve upon $\varphi_q^{\BN}$ by more than an $O(\kappa)$ fraction.\smallskip

\begin{proposition}\label{prop:BN_smallkappa_opt}
Fix $q \in \mathbb{N}_0$. For any $\varphi \in \Ad_q$, the class of admissible policies that admit a finite MA$(q)$ representation, we have
\[
\frac{{\cal C}(\varphi;\psi;\kappa)}{{\cal C}(\varphi_q^{\BN};\psi;\kappa)} \;\ge\;
1-O\!\left(q\,2^q\right)\,\kappa.
\]
\end{proposition}\smallskip

Therefore, for each fixed smoothing order $q$, the Binomial Smoothing policy is \emph{asymptotically optimal} within $\Ad_q$ as $\kappa\downarrow 0$, in the sense that
\[
\lim_{\kappa\downarrow 0}\inf_{\varphi \in \Ad_q} \frac{{\cal C}(\varphi;\psi;\kappa)}{{\cal C}(\varphi_q^{\BN};\psi;\kappa)}= 1.
\]
Intuitively, at the boundary $\kappa=0$ the objective reduces to the manufacturer’s term alone, which, by \cref{prop:BN_min_MSFE}, $\varphi_q^{\BN}$ minimizes over $\Ad_q$; when $\kappa$ is small, any potential gains from deviating from $\varphi_q^{\BN}$ can only come through the volatility of the retailer's inventory term, and hence are necessarily of order $\kappa$.\smallskip

On the other hand, the large-$\kappa$ regime is more subtle. As $\kappa\uparrow\infty$, an optimal policy converges to the myopic ordering rule $\varphi^{\MP}(z)$ (cf.\ \cref{cor:myopic_asym}); however, under weakly stationary demand this myopic rule is demand-dependent (i.e., it depends on $\psi$) and, in general, need not belong to the Binomial class $\Ad^{\BN}$. Thus, unlike in the i.i.d.\ case---where the Binomial family contains the myopic policy via the choice $q=0$---the same demand-agnostic class may fail to capture the limiting optimal structure when demand is correlated.\smallskip

This motivates a simple and natural refinement: augment the Binomial class with the myopic policy $\varphi^{\MP}(z)$, which is asymptotically optimal as $\kappa\uparrow\infty$. While heuristic in nature, doing so is sensible for three reasons. First, the enlargement is \emph{minimal}---it adds only a single additional policy---thereby preserving the analytical tractability and operational simplicity of the Binomial family. Second, it is \emph{structurally consistent} with the i.i.d.\ benchmark: in that case, the myopic policy already lies in $\Ad^{\BN}$, and the augmentation leaves the class unchanged. Third, it directly resolves the sole source of mismatch under correlated demand, namely that the limiting policy depends on $\psi$ whereas $\Ad^{\BN}$ is demand-agnostic; including $\varphi^{\MP}$ therefore ensures that the candidate class contains the correct large-$\kappa$ limit without otherwise altering the Binomial structure. \smallskip

Accordingly, to capture the optimality of $\varphi^{\MP}(z)$ as $\kappa \uparrow \infty$, we define, for an arbitrary class of policies $\widehat{\Ad}$, its expanded class as the {\it join} of $\widehat{\Ad}$ and $\varphi^{\MP}$:
\[
\widehat{\Ad}_{\MP}:=\join\bigl(\widehat{\Ad},\varphi^{\MP}\bigr)
=\Bigl\{(1-x)\,\varphi+x\,\varphi^{\MP}\colon \varphi \in \widehat{\Ad},\;x \in [0,1]\Bigr\}.
\]

In particular, the expanded class of Binomial Smoothing policies is given by
\begin{equation}\label{eq:join_BN}
\Ad_{\MP}^{\BN}
=\Bigl\{(1-x)\,\varphi^{\BN}_q+x\,\varphi^{\MP}\colon \varphi^{\BN}_q \in \Ad^{\BN},\;q\in \mathbb{N}_0,\;x \in [0,1]\Bigr\}.
\end{equation}

The class $\Ad_{\MP}^{\BN}$ is parametrized by the pair $(q,x)$, with $q\in\mathbb{N}_0$ indexing the underlying Binomial policy $\varphi^{\BN}_q$ and $x\in[0,1]$ governing the interpolation with the myopic policy. Consequently, minimizing the supply-chain cost $\C{C}(\varphi;\kappa)$ over $\Ad_{\MP}^{\BN}$ reduces to a low-dimensional search over $(q,x)$ and is therefore computationally tractable. \smallskip

The remainder of this section is devoted to numerically evaluating the performance of an optimal policy within $\Ad_{\MP}^{\BN}$ and comparing it against the two  benchmarks: simple moving-average $\Ad^{\ti{MA}}_{\MP}$ and exponential-smoothing policies $\Ad^{\ti{ES}}_{\MP}$. Unlike the i.i.d.\ case, for general weakly stationary demand we do not have access to the optimal cost $\C{C}^*(\psi;\kappa)$; accordingly, we assess performance relative to the lower bounds derived in \cref{sec:lowerbound}. In particular, combining the bounds in Propositions \ref{prop:boundsiid} and \ref{prop:lowerbound}, we define
\[
\underline{\C{C}}(\psi;\kappa)
:=\max\left\{{\cal C}^*_{\ti F}(\psi;\kappa)\;;\;|\psi(0)|\,{\cal C}^*_{\ti{IID}}\Big(\dfrac{\kappa\,\psi_{\inf}}{|\psi(0)|}\Big)\right\}.
\]
Similar to \eqref{eq:relerror}, we define the relative performance of a class of policies $\widehat{\Ad}_{\MP}$ with respect to the above lower bound on the optimal cost:
\[\mbox{Relative Performance}: \hspace{0.3cm} \qquad \widehat{\cal E}(\psi,\kappa):=\inf_{\varphi \in \widehat{\Ad}_{\MP}}\;{{\cal C}(\varphi;\psi;\kappa) \over \underline{\C{C}}(\psi;\kappa)}. \]

Our first numerical experiment considers a weakly stationary AR(1) market-demand process with persistence parameter $\theta\in(-1,1)$. In our notation, this corresponds to the transfer function
\[
\psi(z)=\frac{1-\theta}{1-\theta z}.
\]
The AR(1) specification is a standard parsimonious model for capturing demand autocorrelation and has been widely adopted in the inventory and supply-chain literature, including work that studies how serial correlation propagates into ordering variability and the bullwhip effect. For example, \citet{lee1997distortion,chen2000quantifying,zhang2004impact} model market demand as an AR(1) process and analyze the impact of inventory policies and/or forecasting methods on the bullwhip effect (see also \citet{hosoda2006variance, Agrawal2009Impact}, among others).\smallskip

For each choice of the AR(1) parameter $\theta\in\{-0.8,-0.4,0.4,0.8\}$ and each cost weight $\kappa$ reported in \cref{tab:ar1_performance_alpha_grid}, we compute the relative performance $\widehat{\cal E}(\kappa)$ for the three classes of inventory policies $\{\Ad^{\ti{MA}}_{\MP},\Ad^{\ti{ES}}_{\MP},\Ad^{\BN}_{\MP}\}$. 

\begin{table}[h]
\centering
\setlength{\tabcolsep}{5pt}
\renewcommand{\arraystretch}{1.05}
\small
\begin{tabular}{@{}c c c c p{0.4em} c c c p{0.4em} c c c p{0.4em} c c c@{}}
\toprule
& \multicolumn{3}{c}{$\theta=-0.8$} && \multicolumn{3}{c}{$\theta=-0.4$} && \multicolumn{3}{c}{$\theta=0.4$} && \multicolumn{3}{c}{$\theta=0.8$} \\
\cmidrule(lr){2-4}\cmidrule(lr){6-8}\cmidrule(lr){10-12}\cmidrule(lr){14-16}
$\kappa$ & MA & ES & BN && MA & ES & BN && MA & ES & BN && MA & ES & BN \\
\cmidrule(r){1-1}\cmidrule(lr){2-4}\cmidrule(lr){6-8}\cmidrule(lr){10-12}\cmidrule(l){14-16}
\rowcolor{gray!10}
0.01 & 5.191 & 3.922 & 1.281 && 4.267 & 3.537 & 1.128 && 5.333 & 5.672 & 2.058 && 7.762 & 9.147 & 3.443 \\
0.1  & 2.023 & 2.232 & 1.323 && 1.754 & 1.952 & 1.129 && 2.309 & 2.587 & 1.637 && 3.026 & 3.585 & 2.442 \\
\rowcolor{gray!10}
0.5  & 1.541 & 1.584 & 1.301 && 1.318 & 1.412 & 1.129 && 1.365 & 1.473 & 1.284 && 1.645 & 1.748 & 1.645 \\
1    & 1.367 & 1.390 & 1.249 && 1.179 & 1.240 & 1.093 && 1.105 & 1.157 & 1.105 && 1.263 & 1.263 & 1.263 \\
\rowcolor{gray!10}
5    & 1.119 & 1.123 & 1.121 && 1.006 & 1.013 & 1.002 && 1.003 & 1.003 & 1.003 && 1.010 & 1.010 & 1.010 \\
10   & 1.059 & 1.061 & 1.060 && 1.001 & 1.003 & 1.001 && 1.001 & 1.001 & 1.001 && 1.003 & 1.003 & 1.003 \\
\bottomrule
\end{tabular}\vspace{0.2cm}
\caption{Relative performance measure for an AR(1) demand process with persistence parameter $\theta$.}
\label{tab:ar1_performance_alpha_grid}
\end{table}

Similar to the results in \cref{table:suboptbenchmarks_noinfonew} for the i.i.d. demand case, the results show that the Binomial class $\Ad^{\BN}_{\MP}$ is uniformly competitive and typically dominates the two smoothing benchmarks, with the largest performance gaps arising in the manufacturer-dominated regime (small $\kappa$), where moving-average and exponential-smoothing rules can be several times larger than the bound while Binomial policies remain much closer to $1$. As $\kappa$ increases, all three classes rapidly approach the lower bound (ratios near $1$), consistent with the fact that in the large-$\kappa$ regime the inventory term dominates and all three classes deliver near-optimal performance by selecting a policy that is close to the myopic policy. Moreover, the relative advantage of $\Ad^{\BN}_{\MP}$ is most pronounced when demand is strongly positively correlated (e.g., $\theta=0.8$), where the na\"ive smoothing policies MA and ES generate substantial cost amplification at small $\kappa$.
\smallskip

Our second computational experiment considers the case of an MA(1) demand process with transfer function $\psi(z)=\psi_0+(1-\psi_0)\,z$, where we restrict to $\psi_0\geq 0.5$ to ensure invertibility of the demand process. We report results for $\psi_0\in\{0.5,0.75,1.25,1.5\}$ and the same grid of cost weights $\kappa$ as in the AR(1) experiment.

\begin{table}[h]
\centering
\setlength{\tabcolsep}{5pt}
\renewcommand{\arraystretch}{1.05}
\small
\begin{tabular}{@{}c c c c p{0.4em} c c c p{0.4em} c c c p{0.4em} c c c@{}}
\toprule
& \multicolumn{3}{c}{$\psi_0=0.5$} && \multicolumn{3}{c}{$\psi_0=0.75$} && \multicolumn{3}{c}{$\psi_0=1.25$} && \multicolumn{3}{c}{$\psi_0=1.5$} \\
\cmidrule(lr){2-4}\cmidrule(lr){6-8}\cmidrule(lr){10-12}\cmidrule(lr){14-16}
$\kappa$ & MA & ES & BN && MA & ES & BN && MA & ES & BN && MA & ES & BN \\
\cmidrule(r){1-1}\cmidrule(lr){2-4}\cmidrule(lr){6-8}\cmidrule(lr){10-12}\cmidrule(l){14-16}
\rowcolor{gray!10}
0.010 & 4.538 & 4.951 & 1.945 && 5.020 & 5.072 & 1.804 && 3.956 & 3.416 & 1.114 && 4.487 & 3.626 & 1.162 \\
0.100 & 2.073 & 2.294 & 1.575 && 2.122 & 2.363 & 1.490 && 1.691 & 1.882 & 1.108 && 1.810 & 2.010 & 1.175 \\
\rowcolor{gray!10}
0.500 & 1.284 & 1.342 & 1.236 && 1.314 & 1.400 & 1.219 && 1.261 & 1.348 & 1.105 && 1.301 & 1.387 & 1.128 \\
1.000 & 1.060 & 1.065 & 1.060 && 1.087 & 1.128 & 1.085 && 1.136 & 1.181 & 1.074 && 1.138 & 1.185 & 1.067 \\
\rowcolor{gray!10}
5.000 & 1.000 & 1.000 & 1.000 && 1.000 & 1.000 & 1.000 && 1.002 & 1.003 & 1.001 && 1.003 & 1.004 & 1.001 \\
100.000 & 1.000 & 1.000 & 1.000 && 1.000 & 1.000 & 1.000 && 1.000 & 1.000 & 1.000 && 1.000 & 1.000 & 1.000 \\
\bottomrule
\end{tabular}\vspace{0.2cm}
\caption{Relative performance measure for an MA(1) demand process parameterized by $\psi_0$.}
\label{tab:ma1_performance_psi0_grid}
\end{table}

Similar to the results for the AR(1) process, \cref{tab:ma1_performance_psi0_grid} shows that the Binomial class $\Ad^{\BN}_{\MP}$ remains uniformly competitive and typically dominates the two smoothing benchmarks across all values of $\psi_0$. The largest performance gaps again arise in the manufacturer-dominated regime (small $\kappa$), where moving-average and exponential-smoothing rules can be several times larger than the bound, while Binomial policies remain much closer to $1$. As $\kappa$ increases, all three classes rapidly approach the lower bound (ratios near $1$), consistent with the fact that in the large-$\kappa$ regime the inventory term dominates and each class selects policies close to the myopic policy. Finally, the relative advantage of $\Ad^{\BN}_{\MP}$ is most pronounced when the demand process departs most from i.i.d.\ behavior (e.g., for $\psi_0$ far from $1$), where the na{\"i}ve smoothing benchmarks exhibit the greatest cost amplification at small $\kappa$.\smallskip

%% file: Management_Science/Extensions_MS.tex
\section{Extension: Policy Refinement}\label{sec:extensions}

A limitation of the Binomial Smoothing class is its granularity as a single-parameter family indexed by $q\in\mathbb{N}_0$, in the sense that the manufacturer’s root MSFE $\sigma_{\ti M}$ is restricted to a discrete set of values. To overcome this limitation, we introduce a variant of the BN policy that takes a convex combination of two BN policies, thereby allowing $\sigma_{\ti M}$ to vary continuously. For expositional simplicity, we present this refinement in the context of the i.i.d.\ demand model of \cref{sec:ordersmooting}. Extending this approach to general weakly stationary demand requires no additional ideas.\smallskip

For i.i.d.\ demand with $\psi(z)\equiv \mathds{1}$, \cref{thm:BN_cost_alpha} implies that $\sigma_{\ti M}\in\{2^{-q}: q\in\mathbb{N}_0\}$. By taking a convex combination of two BN policies, we can instead allow $\sigma_{\ti M}$ to take any value in the interval $(0,1]$. Specifically, for any target root MSFE $\sigma_{\ti M}=\eta$ with $\eta\in(0,1]$, we construct a policy $\varphi_\eta\in\Ad$ as a convex combination of two BN policies such that $\sigma_{\ti M}(\varphi_\eta;\mathds{1})=\eta$.\smallskip

\begin{definition}{\rm (Modified Binomial Smoothing Policy)}  
For a given $\eta \in (0,1]$, let $q_\eta := \lfloor -\log_2(\eta) \rfloor$ and define the scalar $\alpha_\eta := \eta \, 2^{q_\eta + 1} - 1$. We define the {\rm Modified Binomial Smoothing Policy (MB)}, denoted by $\varphi_\eta$, whose $z$-transform $\varphi_{\eta}(z)$ is given by $
\varphi_{\eta}(z) = \alpha_\eta\, \varphi_{q_\eta}(z) + (1 - \alpha_\eta)\, \varphi_{q_\eta + 1}(z),$
where $\varphi_{q_\eta}(z)$ and $\varphi_{q_\eta + 1}(z)$ are the $z$-transforms of the Binomial policies of order $q_\eta$ and $q_\eta + 1$, respectively. That is,
\begin{equation}\label{eq:bin_approx_z}
\varphi_{\eta}(z) = \left(\frac{1 + z}{2}\right)^{q_\eta} \left[ \alpha_\eta + (1 - \alpha_\eta) \, \frac{1 + z}{2} \right].
\end{equation}
\end{definition}

The following result shows that the policy $\varphi_\eta$ in fact achieves the target $\sigma_{\ti M}(\varphi_\eta;\mathds{1})=\eta$.\smallskip

\begin{lemma}\label{lem:MB} The policy $\varphi_\eta$ whose $z$-transform is given by \eqref{eq:bin_approx_z} is invertible and satisfies
\[\sigma^2_{\ti M}(\varphi_\eta;\mathds{1})=\eta^2=\frac{(1+\alpha_\eta)^2}{2^{2q_\eta+2}}\qquad \mbox{and}\qquad \sigma_{\ti I}^2(\varphi_\eta;\mathds{1})=\frac{q_\eta+3-\alpha_\eta}{2}-\frac{2q_\eta+1-\alpha_\eta^2}{2^{2q_\eta+2}}\binom{2q_\eta}{q_\eta}.
\]

\end{lemma}

The relative performance \eqref{eq:relerror} of the Modified Binomial policy (MB), together with that of the original Binomial (BN) policy and benchmark policies, is evaluated numerically in \cref{table:suboptbenchmarks_noinfonew_ref}. 

\begin{table}[!h]
\centering
\setlength{\tabcolsep}{10pt} 
\begin{tabular}{c c c c c c}
\toprule
$\kappa$ & MP & MA & ES & BN & MB \\
\midrule
\rowcolor{gray!20}
0.01   & 46.214  &  2.811  &  3.198 &  1.078   & 1.078 \\
0.1  & 5.842 & 1.569 & 1.748 & 1.051 & 1.051 \\
\rowcolor{gray!20}
0.5  & 1.839 & 1.146 & 1.226 & 1.027 & 1.027 \\
1    & 1.327 & 1.049 & 1.093 & 1.012 & 1.012 \\
\rowcolor{gray!20}
5    & 1.017 & 1.017 & 1.001 & 1.017 & 1.000 \\
10   & 1.005 & 1.005 & 1.000 & 1.005 & 1.000 \\
\bottomrule
\end{tabular}\vspace{0.2cm}
\caption{Relative performance ${\cal E}_{\iid}^j(\kappa)$, as defined in \eqref{eq:relerror}, of benchmark policies as a function of $\kappa$.}
\label{table:suboptbenchmarks_noinfonew_ref}
\end{table}

In \cref{table:suboptbenchmarks_noinfonew_ref}, the Modified Binomial policy (MB) matches the performance of the original Binomial policy (BN) for small and moderate values of $\kappa$, reflecting that the discrete BN grid already selects (nearly) optimal smoothing levels in those regimes. The benefit of MB becomes apparent as $\kappa$ grows: for $\kappa\in\{5,10\}$, MB improves upon BN by attaining essentially optimal relative performance, whereas BN remains slightly suboptimal due to its coarse, integer-indexed choice of $q$. Thus, MB preserves the strong performance guarantees of BN while eliminating the residual discretization error by allowing a continuous calibration of the manufacturer’s root MSFE.

%% file: Management_Science/Conclusions.tex
\section{Concluding Remarks}\label{sec:conclusion}

This paper centers on a simple but important observation: in a decentralized supply chain, a retailer's ordering rule shapes not only inventory outcomes, but also what the manufacturer can learn from the order stream. Order smoothing therefore plays a dual role. On the one hand, it stabilizes operations by reducing short-run order variability. On the other hand, it changes the timing and clarity with which demand information is revealed upstream, especially when explicit information sharing is not feasible.\smallskip

Motivated by the fact that fully optimal policies can be difficult to compute, interpret, and implement, we propose a family of smoothing rules that is deliberately simple and operationally transparent. The Binomial Smoothing class introduces a controlled delay in replenishment and spreads the response to demand over a finite horizon. Despite its simplicity, it delivers robust performance across a wide range of cost trade-offs, and it avoids the severe worst-case behavior that can arise under commonly used textbook heuristics such as simple moving average and exponential smoothing. We further refine the binomial family through a Modified Binomial construction that interpolates between adjacent binomial policies, allowing the degree of smoothing to be tuned continuously and eliminating the residual discretization effects that arise from an integer-indexed policy parameter.\smallskip

From a practical standpoint, our results suggest a clear implementation path for binomial smoothing. A retailer can estimate the demand transfer function, calibrate $\kappa$ to reflect the trade-off between its inventory exposure and the manufacturer's operating costs, and then choose the smoothing horizon $q$ that minimizes the resulting weighted cost. Thus, our results recast order smoothing from an ad hoc variance-reduction heuristic into a strategic, data-driven replenishment decision that can be incorporated into existing delayed-replenishment protocols.\smallskip

A second practical implication concerns how smoothing should be measured. Binomial Smoothing is not merely effective as a variance-reduction device: within each finite smoothing horizon, it minimizes the forecast-adjusted bullwhip measure, namely the unpredictable component of orders relative to the unpredictable component of demand. This distinction matters because an order stream can have low variance yet remain difficult for the manufacturer to forecast. Thus, smoothing policies should be evaluated by the uncertainty they transmit upstream, not only by their effect on order variability.\smallskip

Several directions for future research follow naturally. First, it would be valuable to extend the framework to richer supply-chain settings, including multiple retailers, multiple products, and multi-echelon networks, where strategic interaction and network structure may amplify or dampen the benefits of smoothing. Second, incorporating additional operational frictions---such as capacity constraints, fixed ordering costs, lost sales, or stochastic lead times---would clarify how robust the main insights are in more realistic environments. Third, demand is often nonstationary in practice; developing adaptive versions of the proposed policies that can learn and adjust to changing demand conditions would improve their practical relevance. Finally, an important complement is to study how smoothing-based policies interact with contracts and explicit information-sharing arrangements, and to identify when implicit information transmission through orders can substitute for, or enhance, formal coordination mechanisms. \smallskip

Taken together, our results yield an actionable managerial takeaway: carefully structured smoothing can improve upstream predictability while remaining practical for downstream implementation. This provides a useful tool for improving supply-chain performance even when data sharing is limited or unavailable. More broadly, order smoothing should be viewed not merely as a variance-reduction device, but as a way to control how demand information is conveyed through the replenishment process.

%% file: Management_Science/Appendix_MS.tex
\renewcommand{\thefootnote}{\fnsymbol{footnote}}
\setcounter{footnote}{1}
\section{Proofs}\label{App:Proofs}
\setcounter{equation}{0}
\renewcommand{\theequation}{A\arabic{equation}}
\vspace{0.3cm}

{\sc Proof of \cref{lem:sigma_I}:} Let  ${\cal B}$ denote the Backshift operator, i.e., ${\cal B}D_t=D_{t-1}$.  From the retailer's inventory dynamics in \eqref{eq:invR}, we have $(1-{\cal B})I_t={\cal B}O_t-D_t$. Moreover, \eqref{eq:def_demand} and \eqref{eq:retailerorder} imply that $D_t=d+\psi({\cal B})\,\eps_t$ and $O_t=\varphi({\cal B})\,d+\varphi({\cal B})\,\psi({\cal B})\,\eps_t$. Since $\varphi(1)=1$ and ${\cal B}^k d=d$ for all $k\in\mathbb{N}$, it follows that $\varphi({\cal B})\,d=d$. Substituting these expressions into the inventory dynamics yields $(1-{\cal B})I_t=\psi({\cal B})\bigl[{\cal B}\,\varphi({\cal B})-1\bigr]\eps_t$, so the transfer function of the inventory process is $\psi_{\I}(z)=\psi(z)\,(z\,\varphi(z)-1)/(1-z)$. Therefore, by Parseval's identity, the stationary variance of the retailer's inventory under policy $\varphi$ is
\begin{equation}\nonumber
\sigma^2_{\I}(\varphi;\psi)
=\frac{1}{2\pi}\int_{-\pi}^{\pi}|\psi(e^{-i\lambda})|^2
\left|\frac{e^{-i\lambda}\,\varphi(e^{-i\lambda})-1}{1-e^{-i\lambda}}\right|^2 \D\lambda,
\end{equation}
where we use the standard frequency-domain representation for linear stationary processes; see Section 4 of \cite{BrockwellDavis}. \qed
\vspace{0.5cm}

{\sc Proof of \cref{lem:spectral}:} From \eqref{eq:def_demand} and \eqref{eq:retailerorder}, the retailer's orders satisfy $O_t=\varphi({\cal B})D_t=\varphi({\cal B})\,(d+\psi({\cal B})\eps_t)=d+\varphi({\cal B})\psi({\cal B})\eps_t$, where the second equality uses the fact that $\varphi(1)=1$ and hence $\varphi({\cal B})d=d$. Furthermore, from the inner-outer factorization $\varphi(z)={\cal Q}(z)\,{\cal I}(z)$ we get
$$O_t=d+{\cal Q}({\cal B})\,{\cal I}({\cal B})\,\psi({\cal B})\,\eps_t=d+\psi({\cal B}){\cal Q}({\cal B})\tilde{\eps}_t,$$
where $\tilde{\eps}_t={\cal I}({\cal B})\,\eps_t$ is a Gaussian white noise sequence since ${\cal I}(z)$ is an inner function. Furthermore, since both $\psi$ and $\C{Q}$ are outer and, by \cref{assm1} $\psi \in \mathbb{H}^\infty$,  their product is outer  and $O_t=d+\psi({\cal B}){\cal Q}({\cal B})\tilde{\eps}_t$ corresponds to the manufacturer's  Wold representation of the retailer's orders. As a result, $O_t$ is invertible with respect to $\tilde{\eps}_t$. It follows
\begin{align*}O_{t+1}&=d+\psi({\cal B}){\cal Q}({\cal B})\tilde{\eps}_{t+1} \\
&=\psi(0){\cal Q}(0)\,\tilde{\eps}_{t+1}+d+\Big(\psi({\cal B}){\cal Q}({\cal B})-\psi(0){\cal Q}(0)\Big)\,\tilde{\eps}_{t+1}\\
&=\psi(0){\cal Q}(0)\,\tilde{\eps}_{t+1}+d+\left(\psi({\cal B}){\cal Q}({\cal B})-\psi(0){\cal Q}(0) \over {\cal B}\right)\,\tilde{\eps}_{t}.
\end{align*}
It is not hard to see that $\displaystyle {\psi(z){\cal Q}(z)-\psi(0){\cal Q}(0) \over z}$ is in $\mathbb{H}^2$. From the invertibility of $O_t$ with respect to $\tilde{\eps}_t$, we conclude that
\begin{align*}
m_t(\varphi;\psi)
&=
\e[O_{t+1}\mid{\cal F}^{\ti M}_t]
=
d+
\frac{\psi({\cal B})\,{\cal Q}({\cal B})-\psi(0)\,{\cal Q}(0)}{{\cal B}}\,
\tilde{\eps}_{t}, \\
\sigma^2_{\ti M}(\varphi;\psi)
&=
\var(O_{t+1}\mid{\cal F}^{\ti M}_t)
=
|\psi(0)|^2\,|{\cal Q}(0)|^2 \,\var[{\tilde{\eps}_{t+1}}]=|\psi(0)|^2\,|{\cal Q}(0)|^2,
\end{align*}
where the last equality uses $\var[{\tilde{\eps}_{t+1}}]=\var[{\eps_{t+1}}]$ and our normalization $\var[{\eps_{t+1}}]=1$. 

Finally to conclude the proof we invoke Kolmogorov's formula (see \citealp[Section~5.8]{BrockwellDavis}) to get
$$|{\cal Q}(0)|^2=\exp\!\left(
\frac{1}{2\pi}\int_{-\pi}^{\pi}
\log\bigl|\varphi(e^{-i\lambda})\bigr|^2\,\mathrm{d}\lambda
\right).$$ \qed
\vspace{0.5cm}

{\sc Proof of \cref{prop:group_delay_bound}:} Let $\varphi(z) \in \Ad_q(\mathds{1})$ given by
$$ \varphi(z)=\sum_{n=0}^q \varphi_n\, z^n,$$
 such that $\varphi(1)=1$. Define the {\em reciprocal} order 
 $$\widehat{\varphi}(z):=z^q\,\varphi(z^{-1})=\sum_{n=0}^q \varphi_{q-n}\,z^n,$$
 which trivially satisfies $\widehat{\varphi}(z) \in \Ad_q(\mathds{1})$ and $\widehat{\varphi}(1)=1$. Since, the MA coefficients $\{\varphi_n\}$ are real, we have $|\varphi(z)|=|\widehat{\varphi}(z)|$ on the unit circle $z \in \mathbb{T}$. As a result,
 $$\sigma^2_{\M}(\varphi)=\exp\!\left(
\frac{1}{2\pi}\int_{-\pi}^{\pi}
\log\bigl|\varphi(e^{-i\lambda})\bigr|^2\,\mathrm{d}\lambda
\right)=\exp\!\left(
\frac{1}{2\pi}\int_{-\pi}^{\pi}
\log\bigl| \widehat{\varphi}(e^{-i\lambda})\bigr|^2\,\mathrm{d}\lambda
\right)=\sigma^2_{\M}(\widehat{\varphi}).$$

Also, using the fact that $\varphi(1)=1$ we have that
$${z\,\varphi(z)-1 \over 1-z}={1 \over 1-z}\,\sum_{n=0}^q \varphi_n\,(z^{n+1}-1)=-\sum_{n=0}^q \varphi_n\,\sum_{k=0}^n z^k=\sum_{k=0}^q \Big(\sum_{n=k}^q \varphi_n\Big)\,z^k.$$
Thus, from Parseval's theorem we get
\[
\sigma^2_{\I}(\varphi)=\frac{1}{2\pi}\int_{-\pi}^{\pi}
\left|\frac{e^{-i\,\lambda}\,\varphi(e^{-i\,\lambda})-1}{1-e^{-i\,\lambda}}\right|^2 \D\lambda= \sum_{k=0}^q \Big(\sum_{n=k}^q \varphi_n\Big)^2.\]
A similar derivation produces
\[
\sigma^2_{\I}(\widehat{\varphi})=\frac{1}{2\pi}\int_{-\pi}^{\pi}
\left|\frac{e^{-i\,\lambda}\,\widehat{\varphi}(e^{-i\,\lambda})-1}{1-e^{-i\,\lambda}}\right|^2 \D\lambda= \sum_{k=0}^q \Big(\sum_{n=0}^k \varphi_n\Big)^2.\]
Let us compare  $\sigma^2_{\I}(\varphi)$ and $\sigma^2_{\I}(\widehat{\varphi})$.
\begin{align*}
\sigma^2_{\I}(\widehat{\varphi})-\sigma^2_{\I}(\varphi)&= \sum_{k=0}^q\Bigg[\Big(\sum_{n=0}^k \varphi_n\Big)^2-\Big(\sum_{n=k}^q \varphi_n\Big)^2\Bigg] =\sum_{k=0}^q (1+\varphi_k)\,\Big[\sum_{n=0}^k \varphi_n-\sum_{n=k}^q \varphi_n]\\
&=\sum_{k=0}^q \Big[\sum_{n=0}^k \varphi_n-\sum_{n=k}^q \varphi_n\Big]=\sum_{k=0}^q (q+1-k)\,\varphi_k -\sum_{k=0}^q (k+1)\,\varphi_k\\
&=q-2\,\sum_{k=0}^q k\,\varphi_k.
\end{align*}
Thus, since $\sigma^2_{\M}(\varphi)= \sigma^2_{\M}(\widehat{\varphi})$, we conclude that if $\varphi$ is optimal we must have 
$$\sigma^2_{\I}(\varphi) \leq \sigma^2_{\I}(\widehat{\varphi}) \qquad \Longleftrightarrow \qquad \sum_{k=0}^q k\,\varphi_k \leq \frac{q}{2}.$$
\qed

\vspace{0.5cm}
{\sc Proof of \cref{prop:myopic}:} We prove the result in terms of the transfer $\tpsi(z)=\varphi(z)\,\psi(z)$. Under this representation, the stationary volatility of the retailer's inventory process satisfies
\begin{align*}\sigma^2_{\ti I}(\varphi;\psi)&=\frac{1}{2\pi}\int_{-\pi}^{\pi}|\psi(e^{-i\lambda})|^2
\left|\frac{e^{-i\lambda}\,\varphi(e^{-i\lambda})-1}{1-e^{-i\lambda}}\right|^2 \D\lambda=\frac{1}{2\pi}\int_{-\pi}^{\pi}
\left|\frac{e^{-i\lambda}\,\tpsi(e^{-i\lambda})-\psi(e^{-i\lambda})}{1-e^{-i\lambda}}\right|^2 \D\lambda\\&= {1 \over 2\pi i} \oint_{|z|=1} \left|\frac{z\,\tpsi(z)-\psi(z)}{1-z}\right|^2\, \frac{\D z}{z}.\end{align*}

Under the admissibility requirement $\varphi \in \Ad$, $z=1$ is a root of $z\,\tpsi(z)-\psi(z)$ and so we write 
$$z\,\tpsi(z)-\psi(z)=(1-z)\,\sum_{n=0}^{\infty}\theta_n z^n,$$
for some sequence $\{\theta_n\}$ in $\ell^2$. Equating the coefficient of $z^n$ on both sides, we get that 
$$\theta_0=-\psi_0 \quad \mbox{and}\quad \theta_{n}-\theta_{n-1}=\tpsi_{n-1}-\psi_n \;\Longrightarrow\quad \theta_n=\sum_{k=0}^{n-1} (\psi_{k}-\tpsi_{k-1})-\psi_n.$$
From Parseval's identity, we have that
$$\sigma^2_{\ti I}(\varphi;\psi)={1 \over 2\pi i} \oint_{|z|=1} \left|\sum_{n=0}^{\infty}\theta_n z^n\right|^2\,\frac{\D z}{z}=\sum_{n=0}^\infty |\theta_n|^2=\psi_0^2+\sum_{n=1}^\infty |\theta_n|^2.$$
We can trivially minimize $\sigma_{\ti I}^2(\varphi;\psi)$ by choosing $\theta_n=0$ for all $n\geq 1$. This is achieved by setting $\tpsi_0=\psi_0+\psi_1$ and $\tpsi_n=\psi_{n+1}$ for $n\geq 1$. Consequently, the retailer's myopic policy is
\begin{align*}
\varphi^{\MP}(z)
&=\frac{\tpsi(z)}{\psi(z)}
=\frac{1}{\psi(z)}\left(\psi_0+\psi_1+\sum_{n=1}^\infty \psi_{n+1}\,z^n\right)
=\frac{1}{\psi(z)}\left(\psi_0+\psi_1+\frac{1}{z}\sum_{n=2}^\infty \psi_{n}\,z^n\right)\\
&=\frac{1}{\psi(z)}\left(\psi_0+\psi_1+\frac{1}{z}\big[\psi(z)-\psi_0-\psi_1 z\big]\right)
=\frac{\psi(z)-(1-z)\,\psi(0)}{z\,\psi(z)}\,,
\end{align*}
and the corresponding volatility of the retailer's inventory process satisfies $\sigma_{\ti I}^2(\varphi^{\MP};\psi)=|\psi(0)|^2$. \qed
\vspace{0.5cm}

{\sc Proof of \cref{lem:positive_sigma_m}:} Any $\varphi\in\Ad(\psi)$ satisfies the admissibility requirement $\varphi(1)=1$, so $\varphi$ is not identically zero and thus admits the canonical Nevanlinna inner--outer factorization (see \citealp{rudin1970real,Nikolski}):
\[
\varphi(z)=\mathcal O(z)\,\mathcal I(z),
\]
where $\mathcal O$ is outer and $\mathcal I$ is inner (hence unimodular on the unit circle
$\mathbb T:=\partial\mathbb D=\{z\in\mathbb C:\ |z|=1\}$). Since $|\mathcal I(e^{-i\lambda})|=1$ for a.e.\ $\lambda$, it follows from \eqref{eq_SigmaS_z} that
\begin{align*}
\sigma_{\M}(\varphi)
&=|\psi(0)|\,
\exp\!\left(\frac{1}{2\pi}\int_{-\pi}^{\pi}\log\bigl|\varphi(e^{-i\lambda})\bigr|\,\D\lambda\right)\\
&=|\psi(0)|\,
\exp\!\left(\frac{1}{2\pi}\int_{-\pi}^{\pi}\log\bigl|\mathcal O(e^{-i\lambda})\bigr|\,\D\lambda\right)
=|\psi(0)|\,|\mathcal O(0)|>0,
\end{align*}
where we used the fact that $\psi$ and $\mathcal O$ are outer and therefore have no zeros in $\mathbb D$, implying $|\psi(0)|>0$ and $|\mathcal O(0)|>0$. 

Consider now the special case of a policy $\varphi \in \Ad_q(\psi)$, the set of admissible policies that admit a finite MA($q$) representation as in \eqref{eq:Aq}. Let $\varphi(z)=\sum_{n=0}^q \varphi_n\,z^n$, and let $\{z_1,\dots,z_q\}$ denote its roots, counted with multiplicity. Then,
\begin{align*}
    \sigma_{\M}(\varphi)&=|\psi(0)|\,
\exp\!\left(\frac{1}{2\pi}\int_{-\pi}^{\pi}\log\bigl|\varphi(e^{-i\lambda})\bigr|\,\D\lambda\right)\\
&=|\psi(0)|\,|\varphi(0)|\,\prod_{n:|z_n|<1}\frac{1}{|z_n|}\\
&=|\psi(0)|\,|\varphi(0)|\,\prod_{n=1}^q \frac{\max\{1,|z_n|\}}{|z_n|}\\
&=|\psi(0)|\,|\varphi_q|\,\prod_{n=1}^q \max\{1,|z_n|\}\\
&=|\psi(0)|\,M(\varphi),
\end{align*}
where the first equality uses Jensen's formula (see \citealp{rudin1970real}), the third equality uses Vieta's formula for the product of the roots of $\varphi(z)$, and $M(\varphi)$ denotes the Mahler measure of the polynomial $\varphi(z)$ (\citealp{Mahler1960}). Since the coefficients $\{\varphi_n\}$ satisfy (see \citealp{Mahler1962Inequalities})
$$|\varphi_n| \leq \binom{q}{n}\,M(\varphi),$$
we get that
$$ \sigma_{\M}(\varphi) \geq |\psi(0)| \frac{|\varphi_n|}{\binom{q}{n}}, \quad n=0,\dots,q.$$
We conclude that
$$\sigma_{\M}(\varphi) \geq |\psi(0)|\,\max_{0 \leq n \leq q} \left\{\frac{|\varphi_n|}{\binom{q}{n}}\right\}.$$
Finally, it remains to show that the maximum on the right-hand side is bounded below by $1/2^q$. Suppose, by contradiction, that this is not the case. Then $|\varphi_n|< \binom{q}{n}/2^q$ for all $n=0,\dots,q$. Summing over $n$, we get $\sum_{n=0}^q |\varphi_n|<1$. But $\varphi$ is admissible and hence $\varphi(1)=1$, that is, $\sum_{n=0}^q \varphi_n=1$, which contradicts the previous inequality.
\qed
\vspace{0.5cm}

{\sc Proof of \cref{prop:lim_kappa_0}:}  Let us define
\[
G_k:=\exp\!\left(\frac{1}{2\pi}\int_{-\pi}^{\pi}\log\bigl|\varphi_k(e^{-i\lambda})\bigr|\,d\lambda\right),
\qquad
I_k:=\frac{1}{2\pi}\int_{-\pi}^{\pi}
\left|\frac{e^{-i\lambda}\,\varphi_k(e^{-i\lambda})-1}{1-e^{-i\lambda}}\right|^2\,d\lambda,
\]
and note that $\sigma_{\ti M}(\varphi_k)=|\psi(0)|\,G_k$ and
\[
\sigma^2_{\I}(\varphi_k)=\frac{1}{2\pi}\int_{-\pi}^{\pi}|\psi(e^{-i\lambda})|^2
\left|\frac{e^{-i\lambda}\,\varphi_k(e^{-i\lambda})-1}{1-e^{-i\lambda}}\right|^2 \,d\lambda
\geq \psi^2_{\inf}\, I_k,
\qquad
\mbox{where }\psi_{\inf}:=\inf_{z\in \mathbb{T}}\lvert \psi(z)\rvert.
\]
We next prove that if $G_k\to 0$, then $I_k\to +\infty$, which in turn establishes the statement of the proposition.\smallskip

To this end, set $f_k(z):=z\,\varphi_k(z)\in \mathbb{H}^2$.
Then, by the admissibility of $\varphi_k$ we have $f_k(1)=1$ and on the unit circle $|f_k(e^{-i\lambda})|=|\varphi_k(e^{-i\lambda})|$.
Also define
\[
g_k(z):=\frac{f_k(z)-1}{z-1}=-\,\frac{z\,\varphi_k(z)-1}{1-z},
\]
which is analytic on $\mathbb D$ because $z-1\neq 0$ for $|z|<1$. Notice that $I_k$ is exactly the $L^2$ boundary norm of $g_k$:
\[
I_k=\frac1{2\pi}\int_{-\pi}^{\pi}\bigl|g_k(e^{-i\lambda})\bigr|^2\,d\lambda.
\]
Hence, whenever $I_k<\infty$ we have $g_k\in \mathbb{H}^2$ and $\|g_k\|_{\mathbb{H}^2}^2=I_k$.
\smallskip

We next show that if $I_k$ is bounded then $|f_k(r)|$ must be uniformly bounded away from $0$ at some fixed $r<1$.  Fix $M>0$ and suppose $I_k\le M$ for infinitely many $k$.
For such $k$, $\|g_k\|_{\mathbb{H}^2}\le \sqrt M$, and the standard point evaluation bound in $\mathbb{H}^2$ gives
\[
|g_k(r)|\le \frac{\|g_k\|_{\mathbb{H}^2}}{\sqrt{1-r^2}}\le \frac{\sqrt M}{\sqrt{1-r^2}}
\qquad(0<r<1).
\]
Since $f_k(z)=1+(z-1)g_k(z)$, we get for real $r\in(0,1)$,
\[
|f_k(r)-1|=(1-r)\,|g_k(r)|
\le (1-r)\frac{\sqrt M}{\sqrt{1-r^2}}
=\sqrt M\,\sqrt{\frac{1-r}{1+r}}.
\]
Choose
\[
r:=\frac{4M}{4M+1}\in(0,1).
\]
Then $\frac{1-r}{1+r}=\frac{1}{8M+1}\le \frac{1}{4M}$, so
\[
|f_k(r)-1|\le \sqrt M\cdot \frac{1}{2\sqrt M}=\frac12,
\]
hence
\[
|f_k(r)|\ge \frac12
\qquad\text{for all those }k\text{ with }I_k\le M.
\]

Next, let us combine the inner--outer factorization of $f_k$ and Harnack inequality for harmonic functions to obtain  a positive lower bound on $G_k$. Factor $f_k$ as $f_k=\C{I}_k\,\C{O}_k$, where $\C{I}_k$ is inner and $\C{O}_k$ is outer. Then for outer functions,
\[
|\C{O}_k(0)|
=\exp\!\left(\frac{1}{2\pi}\int_{-\pi}^{\pi}\log|f_k(e^{-i\lambda})|\,\D\lambda\right)
=\exp\!\left(\frac{1}{2\pi}\int_{-\pi}^{\pi}\log|\varphi_k(e^{-i\lambda})|\,\D\lambda\right)
=G_k.
\]
Also, since $|\C{I}_k(z)|\le 1$ for $|z|<1$, we have
\[
|\C{O}_k(r)|=\frac{|f_k(r)|}{|\C{I}_k(r)|}\ge |f_k(r)|\ge \frac12.
\]

Next, bound $\C{O}_k$ from above on a slightly larger disk.
Because $|z|=1$ on the boundary, $\|f_k\|_{\mathbb{H}^2}=\|\varphi_k\|_{\mathbb{H}^2}$, and since
$f_k=1+(z-1)g_k$ we have
\[
\|f_k\|_{\mathbb{H}^2}\le \|1\|_{\mathbb{H}^2}+\|(z-1)g_k\|_{\mathbb{H}^2}
\le 1+\|g_k\|_{\mathbb{H}^2}+\|zg_k\|_{\mathbb{H}^2}
\le 1+2\sqrt M=:N_M.
\]
Because multiplication by an inner function preserves the $\mathbb{H}^2$ norm,
\(\|\C{O}_k\|_{\mathbb{H}^2}=\|f_k\|_{\mathbb{H}^2}\le N_M\). Fix any $R$ with $r<R<1$, e.g.
\[
R:=\frac{1+r}{2}.
\]
Then the $\mathbb{H}^2$ point evaluation bound yields, for all $|z|\le R$,
\[
|\C{O}_k(z)|\le \frac{\|\C{O}_k\|_{\mathbb{H}^2}}{\sqrt{1-|z|^2}}
\le \frac{N_M}{\sqrt{1-R^2}}=:C_M.
\]
Define
\[
V_k(z):=\log\frac{C_M}{|\C{O}_k(z)|}.
\]
Since $\log|\C{O}_k(z)|$ is harmonic (outer functions are zero-free in $\mathbb D$),
$V_k$ is harmonic on $\mathbb D$ and satisfies $V_k\ge 0$ on $|z|\le R$.
By Harnack's inequality for nonnegative harmonic functions on the disk $|z|<R$,
\[
V_k(0)\le \frac{R+r}{R-r}\,V_k(r).
\]
But
\[
V_k(r)=\log\frac{C_M}{|\C{O}_k(r)|}\le \log\frac{C_M}{1/2}=\log(2C_M).
\]
Therefore,
\[
\log\frac{C_M}{|\C{O}_k(0)|}\le \frac{R+r}{R-r}\,\log(2C_M),
\]
so
\[
|\C{O}_k(0)|\ge C_M\,(2C_M)^{-\frac{R+r}{R-r}}
=:c(M)>0.
\]
Recalling $|\C{O}_k(0)|=G_k$, we conclude:
\[
I_k\le M \quad\Longrightarrow\quad G_k\ge c(M).
\]

Finally, let us combine the previous results to complete the proof. Per hypothesis, assume that $G_k\to 0$, and fix any $M>0$. Since $c(M)>0$, there exists $K$ such that for all $k\ge K$, $G_k<c(M)$.
For those $k$, we cannot have $I_k\le M$ (otherwise $G_k\ge c(M)$), hence
\[
I_k>M \qquad \forall k\ge K.
\]
Because $M>0$ was arbitrary, this implies $I_k\to +\infty$. \qed

\vspace{0.5cm}

{\sc Proof of \cref{prop:lowerbound}:} We prove the result by solving for $\tpsi(z)=\varphi(z)\,\psi(z)$ using a time-domain argument. Specifically, let $\tpsi(z)=\sum_{n=0}^\infty \tpsi_n\,z^n$, and take the sequence $\{\tpsi_n\}$ as the decision variables used to minimize the relaxed supply chain cost ${\cal C}_{\ti F}$. From the proof of \cref{prop:myopic}, we have that 
$$\sigma^2_{\ti I}(\tpsi)={1 \over 2\pi i} \oint_{|z|=1} \left|\sum_{n=0}^{\infty}\theta_n z^n\right|^2\,\frac{\D z}{z}=\sum_{n=0}^\infty |\theta_n|^2,$$
where 
$$\theta_0=-\psi_0 \quad \mbox{and}\quad \theta_n=\sum_{k=0}^{n-1} (\psi_{k}-\tpsi_{k-1})-\psi_n.$$
It follows that the relaxed cost ${\cal C}_{\ti F}$ is equal to
$${\cal C}_{\ti F}=\kappa\,\left(\psi_0^2+\sum_{n=1}^{\infty}\Big[\sum_{k=0}^{n-1}\bigl(\tpsi_k-\psi_k\bigr)-\psi_n\Big]^2\right)^{1 \over 2}+|\tpsi_0|.$$
In order to minimize ${\cal C}_{\ti F}$ over the sequence $\{\tpsi_n\}$ note that for a fixed $\tpsi_0$, the first summand is minimized by setting $\tpsi_1=\psi_2+\psi_1+\psi_0-\tpsi_0$ and $\tpsi_n=\psi_{n+1}$ for all $n \geq 2$. With this choice, all the terms in the summation over $n$ vanish except for the first one,  corresponding to $n=1$, which is equal to $\tpsi_0-\psi_0-\psi_1$. Thus, we can solve for the minimum value of ${\cal C}_{\ti F}$ by solving the one-dimensional problem.
$${\cal C}^*_{\ti F}=\min_{\tpsi_0} \; \kappa\,\sqrt{\psi_0^2+(\tpsi_0-\psi_0-\psi_1)^2}+\tpsi_0^2.$$
Letting $\tpsi^*_0$ denote the optimal value, the resulting optimal policy $\tpsi^{\ti F}(z)$ is given by
\begin{align*}
   \tpsi^{\ti F}(z)&=\tpsi^*_0+(\psi_2+\psi_1+\psi_0-\tpsi^*_0)\,z+\sum_{n=2}^\infty \psi_{n+1}\,z^n = \tpsi^*_0+(\psi_1+\psi_0-\tpsi^*_0)\,z+ \sum_{n=1}^\infty \psi_{n+1}\,z^n\\
   &= \tpsi^*_0+(\psi_1+\psi_0-\tpsi^*_0)\,z+ {1 \over z}\, \left(\sum_{n=0}^\infty \psi_n\,z^n-\psi_0-\psi_1\,z\right) \\
   & = (1-z)\, (\tpsi^*_0-\psi_1)+{\psi(z)-(1-z^2)\, \psi_0 \over z}.
\end{align*}
Finally, since $\psi_0=\psi(0)$ and $\psi_1=\psi'(0)$, we conclude that

\[
\varphi^{\ti F}(z)=\frac{\tpsi^{\ti F}(z)}{\psi(z)}=
\frac{\psi(z)-(1-z^2)\,\psi(0)+(1-z)\,z\,\bigl(\tpsi^*_0-\psi'(0)\bigr)}
{z\,\psi(z)}.~~~\Box
\]
\vspace{0.5cm}

{\sc Proof of \cref{prop:boundsiid}}: For a general demand $\psi(z)$ and admissible policy $\varphi(z) \in \Ad$, we have
\begin{align*}
{\cal C}(\varphi;\kappa)&=\kappa\, \left(\frac{1}{2\pi}\int_{-\pi}^{\pi}|\psi(e^{-i\lambda})|^2
\left|\frac{e^{-i\lambda}\,\varphi(e^{-i\lambda})-1}{1-e^{-i\lambda}}\right|^2 \D\lambda\right)^{1 \over 2}
+|\psi(0)|\,\exp\!\left(\frac{1}{2\pi}\int_{-\pi}^{\pi}
\log\bigl|\varphi(e^{-i\lambda})\bigr|\,\mathrm{d}\lambda\right)\\
&\geq |\psi(0)| \left[\frac{\kappa\,\psi_{\inf}}{|\psi(0)|}\, \left(\frac{1}{2\pi}\int_{-\pi}^{\pi}
\left|\frac{e^{-i\lambda}\,\varphi(e^{-i\lambda})-1}{1-e^{-i\lambda}}\right|^2 \D\lambda\right)^{1 \over 2}
+\exp\!\left(\frac{1}{2\pi}\int_{-\pi}^{\pi}
\log\bigl|\varphi(e^{-i\lambda})\bigr|\,\mathrm{d}\lambda\right)\right],\\
&= |\psi(0)|\,{\cal C}_{\ti{IID}}(\varphi,\kappa_{\inf}),\end{align*}
where 
$\psi_{\inf}:=\inf_{\zeta \in \mathbb{T}}   |\psi(\zeta)|$ and $\kappa_{\inf}:=\frac{\kappa\,\psi_{\inf}}{|\psi(0)|}$.
By taking infimum over $\varphi$, and applying the result in \cref{prop:solnosharing}, we obtain the lower bound
${\cal C}^*(\kappa)\;\geq\;|\psi(0)|\,{\cal C}^*_{\ti{IID}}(\kappa_{\inf})$. 
Letting $\psi_{\sup}:=\sup_{\zeta \in \mathbb{T}}   |\psi(\zeta)|$, and using the same argument, we can derive the upper bound: ${\cal C}^*(\kappa)\;\leq\;|\psi(0)|\,{\cal C}^*_{\ti{IID}}(\kappa_{\sup})$, where $\kappa_{\sup}:=\frac{\kappa\,\psi_{\sup}}{|\psi(0)|}$. \qed

\vspace{0.5cm}

{\sc Proof of \cref{prop:iid_reduction}}: Throughout the proof we use the following intermediate result: for $\kappa_1 \geq \kappa_2 > 0$,
\begin{align}\label{eq:auxbound}
    {\cal C}^*(\kappa_1)
    &= \inf_{\varphi \in \Ad} \Big\{\kappa_1\,\sigma_{\ti I}(\varphi)+\sigma_{\ti M}(\varphi)\Big\}
    \leq \inf_{\varphi \in \Ad} \Big\{\kappa_1\,\sigma_{\ti I}(\varphi)+{\kappa_1 \over \kappa_2}\,\sigma_{\ti M}(\varphi)\Big\} \notag\\
    &= {\kappa_1 \over \kappa_2}\,\inf_{\varphi \in \Ad} \Big\{\kappa_2\,\sigma_{\ti I}(\varphi)+\sigma_{\ti M}(\varphi)\Big\}
    = {\kappa_1 \over \kappa_2}\, {\cal C}^*(\kappa_2).
\end{align}

We begin with the ``only if'' direction of \cref{prop:iid_reduction}. Suppose $\Ad'$ is an $\alpha$-approximation class under market demand $\psi(z)$. Then, by \eqref{eq:alpha_approx_def},
\[
\inf_{\varphi \in \Ad'} {\cal C}(\varphi;\kappa) \leq \alpha \,{\cal C}^*(\kappa),
\qquad\text{for all }\kappa \in (0,\infty).
\]
Recall the upper and lower bounds $|\psi(0)|\, {\cal C}_{\ti{IID}}(\varphi;\kappa_{\inf}) \leq {\cal C}(\varphi;\kappa) \leq |\psi(0)|\, {\cal C}_{\ti{IID}}(\varphi;\kappa_{\sup})$, where $\kappa_{\inf}=\kappa\, \psi_{\inf}/|\psi(0)|$ and $\kappa_{\sup}=\kappa\, \psi_{\sup}/|\psi(0)|$. Therefore,
$|\psi(0)|\,{\cal C}_{\ti{IID}}(\varphi;\kappa)\leq {\cal C}\big(\varphi;\kappa\,\frac{|\psi(0)|}{\psi_{\inf}}\big) \leq |\psi(0)|\,{\cal C}_{\ti{IID}}\big(\varphi;\kappa\,\frac{\psi_{\sup}}{\psi_{\inf}}\big)$.
Taking the infimum over $\varphi \in \Ad'$ in the first inequality, and the infimum over $\varphi \in \Ad$ in the second inequality, yields
\begin{align*}
|\psi(0)|\,\inf_{\varphi \in \Ad'} {\cal C}_{\ti{IID}}(\varphi;\kappa)
&\leq \inf_{\varphi \in \Ad'} {\cal C}\Big(\varphi;\kappa\,\frac{|\psi(0)|}{\psi_{\inf}}\Big)
\leq \alpha\, {\cal C}^*\Big(\kappa\,\frac{|\psi(0)|}{\psi_{\inf}}\Big) \leq \alpha\,|\psi(0)|\, {\cal C}^*_{\ti{IID}}\Big(\kappa\,\frac{\psi_{\sup}}{\psi_{\inf}}\Big)
\\
&\leq \alpha\,\frac{\psi_{\sup}}{\psi_{\inf}}\,|\psi(0)|\,{\cal C}^*_{\ti{IID}}(\kappa),
\end{align*}
where the last inequality follows from \eqref{eq:auxbound}. Dividing by $|\psi(0)|$ and setting $\tilde{\alpha}=\alpha\,\frac{\psi_{\sup}}{\psi_{\inf}}$ shows that $\Ad'$ is a $\tilde{\alpha}$-approximation class when demand is i.i.d., completing the proof of the ``only if'' direction.

For the ``if'' direction, suppose that $\Ad'$ is an $\tilde{\alpha}$-approximation class when demand is i.i.d., that is,
$$\inf_{\varphi \in \Ad'} {\cal C}_{\ti{IID}}(\varphi;\kappa) \leq \tilde{\alpha}\,{\cal C}^*_{\ti{IID}}(\kappa),
\qquad\text{for all }\kappa \in (0,\infty).$$
This, together with the bounds $|\psi(0)|\, {\cal C}_{\ti{IID}}(\varphi;\kappa_{\inf}) \leq {\cal C}(\varphi;\kappa) \leq |\psi(0)|\, {\cal C}_{\ti{IID}}(\varphi;\kappa_{\sup})$ imply
\begin{align*}
\inf_{\varphi \in \Ad'} {\cal C}(\varphi;\kappa) & \leq  |\psi(0)|\, \inf_{\varphi \in \Ad'}\, {\cal C}_{\ti{IID}}(\varphi;\kappa_{\sup}) \leq |\psi(0)|\,\tilde{\alpha}\,{\cal C}^*_{\ti{IID}}(\kappa_{\sup}) \leq \tilde{\alpha}\,{\cal C}^*\Big(\kappa\,\frac{\psi_{\sup}}{\psi_{\inf}}\Big) \leq \tilde{\alpha}\,\frac{\psi_{\sup}}{\psi_{\inf}}\, {\cal C}^*(\kappa).
\end{align*}
Setting $\alpha=\tilde{\alpha}\,\frac{\psi_{\sup}}{\psi_{\inf}}$ shows that $\Ad'$ is an $\alpha$-approximation class for the market demand $\psi$. \qed
\vspace{0.5cm}

{\sc Proof of \cref{prop:BN_min_MSFE}:} From \cref{lem:spectral}, it follows that $\sigma_{\ti M}(\varphi;\psi)=|\psi(0)|\,\sigma_{\ti M}(\varphi;\mathds{1})$. Therefore, without loss of generality, we may restrict the proof to the special case of i.i.d.\ demand, that is, $\psi(z)\equiv \mathds{1}$. \smallskip

By the Fundamental Theorem of Algebra, any $\varphi \in \Ad_q$ can be expressed as
$$\varphi(z) = \prod_{k=1}^q \frac{z - z_k}{1 - z_k},$$
where $\{z_k\}_{k=1}^q$ are the (possibly complex) roots of $\varphi(z)$. Note that $\varphi(1) = 1$, which corresponds to the inventory stability condition required by the admissibility set $\Ad$.

Let $\varphi(z)=B(z)\,Q(z)$ denote the Nevanlinna inner--outer factorization of $\varphi(z)$, where $B(z)$ is a Blaschke product and $Q(z)$ is an outer function given by
$$B(z)=\prod_{k \in {\cal I}} \frac{z-z_k}{1-\bar{z}_k\,z}\qquad \mbox{and}\qquad Q(z)=\left(\prod_{k \in {\cal I}} \frac{1-\bar{z}_k\,z}{1-z_k}\right)\,\left(\prod_{k \in {\cal O}} \frac{z-z_k}{1-z_k}\right),$$
where ${\cal I}:=\{k \colon |z_k|<1\}$ and  ${\cal O}:=\{k \colon |z_k|\geq 1\}$. It follows that
\begin{align*}\exp\left(\frac{1}{2\pi}\int_{-\pi}^{\pi}
\log\bigl|\varphi(e^{-i\lambda})\bigr|\,\mathrm{d}\lambda\right)&= \exp\left(\frac{1}{2\pi}\int_{-\pi}^{\pi}
\log\bigl|B(e^{-i\lambda})\bigr|\,|Q(e^{-i\lambda})\bigr|\,\mathrm{d}\lambda\right) \\ &= \exp\left(\frac{1}{2\pi}\int_{-\pi}^{\pi}
\log|Q(e^{-i\lambda})\bigr|\,\mathrm{d}\lambda\right)\\ &=|Q(0)|,\end{align*}
where the second equality uses the fact that the Blaschke factor $B(z)$ is unimodular on the unit circle and the last equality uses Jensen’s formula for the outer function $Q(z)$.
Thus,  the problem of minimizing the manufacturer's root MSFE, $\sigma_{\ti M}(\varphi)$, over the class $\Ad_q$ reduces to
$$\inf_{\{z_k\}} |Q(0)|= \inf_{\{z_k\}} \;\left(\prod_{k \in {\cal I}} \frac{1}{|1-z_k|}\right)\,\left(\prod_{k \in {\cal O}} \frac{|z_k|}{|1-z_k|}\right).$$

It is not hard to see that 
$$\inf_{|a|<1}\, \frac{1}{|1-a|}=\frac{1}{2} \qquad \mbox{and} \qquad \inf_{|a|\geq 1} \frac{|a|}{|1-a|}=\frac{1}{2}.$$
As a result, we have 
$$
\inf_{\varphi \in \Ad_q}\; \exp\left(\frac{1}{2\pi}\int_{-\pi}^{\pi}
\log\bigl|\varphi(e^{-i\lambda})\bigr|\,\mathrm{d}\lambda\right)=\left(1 \over 2\right)^q.
$$
Finally, the Binomial Smoothing policy $\varphi^{\BN}_q(z)=({1+z \over 2})^q$ is invertible and satisfies $\sigma_{\ti M}(\varphi^{\BN}_q)=|\varphi^{\BN}_q(0)|=({1 \over 2})^q$. We conclude that $\varphi^{\BN}_q(z)$ minimizes the manufacturer's root MSFE over the class $\Ad_q$ of MA$(q)$ policies. \qed

\vspace{0.5cm}

{\sc Proof of \cref{thm:BN_cost_alpha}:} It is easy to see that, since $\varphi^{\BN}_q(z)$ is invertible (i.e., an outer function), by \cref{lem:spectral}
\[
\sigma_{\ti{M}}^2(\varphi^{\BN}_q)=(\varphi^{\BN}_q(0))^2={1 \over 2^{2q}}.
\]
Let us next evaluate the retailer's inventory volatility 
$$\sigma_{\ti I}^2(\varphi^{\BN}_q)=\frac{1}{2\pi}\int_{-\pi}^{\pi}
 \Big|\frac{z\varphi^{\BN}_q(z)-1}{1-z}\Big|^2\, d\theta, \qquad z=e^{-i\,\theta}$$
under the Binomial policy $\varphi^{\BN}_q(z)=\Big({1+z \over 2}\Big)^q$. Let us set
\[
\varrho(z)
:=\frac{z\varphi^{\BN}_q(z)-1}{1-z}
=\frac{1}{1-z}\frac{z(1+z)^q-2^q}{2^q}
\]
and
\[
\varrho(1/z)
=
\frac{1}{1-\frac{1}{z}}\Big(\frac{\frac{1}{z}\bigl(1+\frac{1}{z}\bigr)^q}{2^q}-1\Big)=\frac{-z}{1-z}\frac{(1+z)^q-2^qz^{q+1}}{2^qz^{q+1}}.
\]
Therefore,
\[
\varrho(z)\varrho(1/z)=-\frac{z}{(1-z)^2}\left(\frac{(1+z)^{2q}}{2^{2q}z^q}-\frac{z(1+z)^{q}}{2^{q}}
-\frac{(1+z)^q}{2^qz^{q+1}}+1\right).
\]
On the unit circle \(z=e^{-i\theta}\), we want to calculate 
\[
\sigma_{\ti I}^2=\frac{1}{2\pi}\int_{-\pi}^{\pi}
\varrho(z)\varrho(1/z)d\theta.
\]
Define
\(
\varphi(z) = \varrho(z)\varrho(1/z).
\)
We claim that, with proof coming later,
\[
\frac{1}{2\pi}\int_{-\pi}^{\pi} \varphi\bigl(z\bigr)d\theta
=
(\text{coefficient of }z^0\text{ in }\varphi(z)).
\]
Define $\varphi(z)=\varphi^{(1)}(z)+\varphi^{(2)}(z)+\varphi^{(3)}(z)+\varphi^{(4)}(z)$, where 
\[
\varphi^{(k)}(z)
=
-\frac{z}{(1-z)^2}
\Phi^{(k)}(z),
\]
where we define
\[
\Phi^{(1)}(z)
=
\frac{(1+z)^{2q}}{2^{2q}z^q},
\quad
\Phi^{(2)}(z)
=
-\frac{z(1+z)^q}{2^q},
\quad
\Phi^{(3)}(z)
=
-\frac{(1+z)^q}{2^qz^{q+1}},
\quad
\Phi^{(4)}(z)
=
1.
\]
Although \(\varrho(z)\) appears singular at \(z=1\), its numerator also vanishes there:
\(\bigl[z(1+z)^q/2^q - 1\bigr]\big|_{z=1} = 0,\)
so this singularity is removable.  In fact,
\(\lim_{z\to 1} \varrho(z) = -\frac{q+2}{2},\)
and thus \(z=1\) does not contribute any extra residue. It suffices to consider the residue at $z=0$. Since we only care about the singularity near $z=0$, we consider only the existence of any expansion near $z=0$. Note that 
\[
-\frac{z}{(1-z)^2}
=
-\sum_{m=1}^{\infty} mz^m.
\]
We separately consider four functions:
\[
\Phi^{(1)}(z)
=
\frac{(1+z)^{2q}}{2^{2q}z^q}=
\frac{1}{2^{2q}}
\sum_{r=0}^{2q} \binom{2q}{r} z^{r - q},\text{ and }\varphi^{(1)}(z) = -\sum_{m=1}^\infty
\frac{(1+z)^{2q}}{2^{2q}z^q}=
\frac{1}{2^{2q}}
\sum_{r=0}^{2q} \binom{2q}{r} mz^{m+r - q}.
\]
Its coefficient of $z^0$ is determined by (taking $m=q-r\geq 1$, i.e., $r\leq q-1$), 
\[
\tau^{(1)}_0
=
-\frac{1}{2^{2q}}
\sum_{r=0}^{q-1}
\binom{2q}{r}(q-r).
\]
Second, for $\varphi^{(2)}(z)$, 
\[
\Phi^{(2)}(z)
=
-\frac{z(1+z)^q}{2^q}
=
-\frac{1}{2^q}
\sum_{r=0}^q \binom{q}{r} z^{r+1}.
\]
All resulting exponents of $\varphi^{(2)}(z)$ exceed 0, so
\(
\tau^{(2)}_0 = 0.
\)

\noindent
Third, for $\varphi^{(3)}(z)$, 
\[
\Phi^{(3)}(z)
=
-\frac{(1+z)^q}{2^qz^{q+1}}
=
-\frac{1}{2^q}
\sum_{r=0}^q \binom{q}{r}z^{-r-1},\text{ and }\varphi^{(3)}(z) = \sum_{m=1}^\infty
\frac{1}{2^q}
\sum_{r=0}^q \binom{q}{r}mz^{m-r-1},
\]
Its coefficient of $z^0$ is determined by (taking $m=r+1\geq 1$).
Thus
\[
\tau^{(3)}_0
=
\frac{1}{2^q}
\sum_{r=0}^q \binom{q}{r}(r+1).
\]

Lastly, since $\Phi^{(4)}(z)=
1$, all resulting exponents of $\varphi^{(4)}(z)$ exceed 0, so
\(
\tau^{(4)}_0 = 0.
\)

Summing up $
\tau^{(1)}_0 + \tau^{(2)}_0 + \tau^{(3)}_0 + \tau^{(4)}_0$. We have
\[
\tau_0
=
-\frac{1}{2^{2q}}
\sum_{r=0}^{q-1}
\binom{2q}{r}(q-r)
+
\frac{1}{2^q}
\sum_{r=0}^q
\binom{q}{r}(r+1).
\]
By standard binomial coefficients formula, 
\begin{equation}\label{eq:binom_formula}
\sum_{r=0}^q \binom{q}{r}(r+1) = 2^{q-1}(q+2),\quad \text{
and }
\sum_{r=0}^{q-1} \binom{2q}{r}(q-r) = \frac{q}{2}\binom{2q}{q}.
\end{equation}
We then have
\[
\tau_0
=
\frac{q+2}{2}
-
\frac{q}{2^{2q+1}}\binom{2q}{q}.
\]
Therefore
\[
\frac{1}{2\pi}\int_{-\pi}^{\pi} \varrho(z)\varrho\!\bigl(1/z\bigr)d\theta
=
\tau_0
=
\frac{q+2}{2}-\frac{q}{2^{2q+1}}\binom{2q}{q}.
\]\medskip

Let us now turn to the second part of the proof that shows that the Binomial Smoothing class is an $\alpha$-approximation family. We prove this result by deriving lower bound for  ${\cal C}_{\ti{IID}}^*(\kappa)$ and upper bound for ${\cal C}^{\ti{BN}}(\kappa)$ and showing their ratio is bounded uniformly in $\kappa$.

Let us first derive the lower bound for ${\cal C}_{\ti{IID}}^*(\kappa)$. First, from \cref{prop:solnosharing}, we have that 
${\cal C}_{\ti{IID}}^*(\kappa) = 1 + \sqrt{\kappa^2 - 1}$ for $\kappa \geq \sqrt{5}$. On the other hand, for $\kappa<\sqrt{5}$, we have $\displaystyle {\cal C}^*_{\ti{IID}}(\kappa)={\kappa \over 2}\,\sqrt{5+2\,\gamma_\kappa}+{e^{-\gamma_\kappa} \over 2}$ where $\gamma_\kappa \geq 0$ is the unique solution of $\kappa^2=(5+2\,\gamma)\,\exp(-2\,\gamma)$. It follows that $2\,\gamma_{\kappa} \geq -\log(\kappa^2)$, and thus 
\[
{\cal C}_{\ti{IID}}^*(\kappa) = {\kappa \over 2}\,\sqrt{5+2\,\gamma_\kappa}+{e^{-\gamma_\kappa} \over 2} \geq {\kappa \over 2} \sqrt{5 - \log(\kappa^2)} \quad \mbox{for }\kappa <\sqrt{5}.
\]

Similarly, we derive separate upper bounds for ${\cal C}^{\ti{BN}}(\kappa)$ for the cases $\kappa \geq \sqrt{5}$ and $\kappa<\sqrt{5}$. For $\kappa \geq \sqrt{5}$, we upper bound  using the cost of the myopic policy $\varphi^{\MP}(z)$, which is a special member of the class of Binomial Smoothing policies  with $q=0$. Thus, we have ${\cal C}^{\ti{BN}}(\kappa) \leq {\cal C}(\varphi^{\MP},\kappa) = \kappa+1$, for all $\kappa$ and in particular
\[
\frac{{\cal C}^{\ti{BN}}(\kappa)}{{\cal C}_{\ti{IID}}^*(\kappa)}
\le \frac{\kappa+1}{1+\sqrt{\kappa^2-1}}
\le {\sqrt{5}+1 \over 3}\approx 1.079, \qquad \mbox{for }\kappa \ge \sqrt{5}.
\]

For the case $\kappa < \sqrt{5}$,  an upper bound for ${\cal C}^{\ti{BN}}(\kappa)$ is obtained as follows:
\begin{align*}{\cal C}^{\ti{BN}}(\kappa)&=\min_{q \in \mathbb{N}_0} \left\{\kappa\,\sqrt{\frac{q+2}{2}-\frac{q}{2^{2q+1}}\binom{2q}{q}}+ {1 \over 2^q}  \right\}  \leq  \min_{q \in \mathbb{N}_0} \left\{\kappa\,\sqrt{\frac{q+2}{2}-\frac{q}{2^{2q+1}}{4^q \over 2\,\sqrt{q}}}+ {1 \over 2^q}  \right\} \\ & \leq \min_{q \in \mathbb{N}_0} \left\{\kappa\,\sqrt{\frac{q+2}{2}-\frac{\sqrt{q}}{4}}+ {1 \over 2^q}\right\}\leq \min_{q \in \mathbb{N}} \left\{\kappa\,\sqrt{\frac{q+2}{2}-\frac{1}{4}}+ {1 \over 2^q}\right\}\qquad  \mbox{Define }q_\kappa:=\left\lceil \frac{\gamma_{\kappa}}{\log(2)} \right\rceil+1
 \\
& \leq \kappa\,\sqrt{\frac{q_\kappa}{2}+\frac{3}{4}}+ {1 \over 2^{q_\kappa}} 
\leq \kappa\,\sqrt{\frac{\gamma_\kappa}{2\,\log(2)}+\frac{7}{4}}+ \frac{e^{-\gamma_{\kappa}}}{2} \leq  \kappa\,\sqrt{\frac{\gamma_\kappa}{2\,\log(2)}+\frac{5}{4\,\log(2)}}+ \frac{e^{-\gamma_{\kappa}}}{2}
\\ & \leq \frac{1}{\sqrt{\log(2)}}\left(\frac{\kappa}{2}\,\sqrt{5+2\,\gamma_{\kappa}}+\frac{e^{-\gamma_{\kappa}}}{2} \right)= \frac{1}{\sqrt{\log(2)}}\, {\cal C}_{\ti{IID}}^*(\kappa).
\end{align*}
The first inequality follows from Stirling approximation for $q \geq 1$.

Combining the two upper bounds for the cases $\kappa \ge \sqrt{5}$ and $\kappa < \sqrt{5}$, we conclude that
\[
{\cal C}^{\ti{BN}}_{\ti{IID}}(\kappa) \le \frac{1}{\sqrt{\log(2)}}\, {\cal C}_{\ti{IID}}^*(\kappa)
\qquad \text{for all } \kappa>0,
\]
which shows that the class of Binomial Smoothing policies is an $\alpha$-approximation with $\alpha^{\BN} = 1/\sqrt{\log(2)}$.

Finally, from the proof of \cref{prop:iid_reduction}, we can extend the result to a general demand model by choosing $\alpha=\alpha^{\BN}\,\frac{\psi_{\sup}}{\psi_{\inf}}$.
 \qed

\vspace{0.5cm}

{\sc Proof of \cref{lem:MB}:}  Let us consider the class of policies with $z$-transform 
\[
\varphi(z)=\Big(\frac{1+z}{2}\Big)^q\Big(\alpha+(1-\alpha)\frac{1+z}{2}\Big),
\]
 for some $q \in \mathbb{N}_0$ and $\alpha \in (0,1]$. The invertibility of $\varphi(z)$ follows from noticing that the monomial $\alpha +(1-\alpha)\,{1+z \over 2}$ has a root at $\zeta=(\alpha+1)/(\alpha-1)$ with $|\zeta|\geq 1$ for $\alpha \in (0,1]$. As a result, all the roots of $\varphi(z)$ lie on or outside of  the unit circle which implies that $\varphi$  is invertible. From this we conclude that 
$$\sigma^2_{\ti M}(\varphi)=\varphi^2(0)\,\sigma^2_\eps={1 \over 2^{2q}}\,\left(\alpha +(1-\alpha)\,{1 \over 2}\right)^2\, \sigma^2_\eps ={1 \over 2^{2q+2}}\,\Big(1+\alpha\Big)^2\, \sigma^2_\eps.$$

For the special case of $\varphi_\eta(z)$ in \cref{lem:MB}, with $q=q_\eta$ and $\alpha=\alpha_\eta = \eta \, 2^{q_\eta + 1} - 1$, we get $\sigma^2_{\ti M}(\varphi)=\eta^2\,\sigma^2_\eps$.

Let us turn to the derivation of $\sigma^2_{\ti I}(\varphi)$. 
We will use the same notation as in the proof of \cref{thm:BN_cost_alpha} with 
\[
\varrho_q(z)
:=\frac{z\varphi_q(z)-1}{1-z}
=\frac{1}{1-z}\frac{z(1+z)^q-2^q}{2^q}.
\]

By definition, we can see that
\begin{align}\label{eq:alpha_combine}
\varrho(z)\varrho\bigl(1/z\bigr)
=&
\alpha^2\varrho_q(z)\varrho_q(1/z)
+\alpha(1-\alpha)\varrho_q(z)\varrho_{q+1}(1/z)\\
&
+\alpha(1-\alpha)\varrho_{q+1}(z)\varrho_q(1/z)
+(1-\alpha)^2\varrho_{q+1}(z)\varrho_{q+1}(1/z).\nonumber
\end{align}
Therefore, we can write
\[
\sigma^2_{\ti I}(\varphi)=\frac{1}{2\pi}\int_{-\pi}^{\pi} \varrho(z)\varrho\!\bigl(1/z\bigr)d\theta=\alpha^2\tau_{0,(q,q)}+\alpha(1-\alpha)\big(\tau_{0,(q,q+1)}+\tau_{0,(q+1,q)}\big)+(1-\alpha)^2\tau_{0,(q+1,q+1)},
\]
where $\tau_{0,(q+j_0,q+j_1)}$ denotes the normalized integral of $\varrho_{q+j_0}(z)\varrho_{q+j_1}(1/z)$. 
By an argument in the proof of \cref{thm:BN_cost_alpha}, we have
\begin{equation}\label{eq:qq1}
\tau_{0,(q,q)}
=
\frac{q+2}{2}-\frac{q}{2^{2q+1}}\binom{2q}{q},\quad \tau_{0,(q+1,q+1)}=
\frac{q+3}{2}-\frac{q+1}{2^{2q+3}}\binom{2q+2}{q+1}.
\end{equation}
It remains to calculate the other two terms. 
Following the similar arguments,
\[
\varrho_q(z)\varrho_{q+1}(1/z)
=-\frac{z}{(1-z)^2}\left[\frac{(1+z)^{2q+1}}{2^{2q+1}z^{q+1}}-\frac{z(1+z)^q}{2^q}-\frac{(1+z)^{q+1}}{2^{q+1}z^{q+2}}+1\right].
\]
Therefore, we are now working on $\Phi^{(4)}_{(q,q+1)}(z)=1
$, and 
\[
\Phi^{(1)}_{(q,q+1)}(z)
=
\frac{(1+z)^{2q+1}}{2^{2q}z^{q+1}},
\quad
\Phi^{(2)}_{(q,q+1)}(z)
=
-\frac{z(1+z)^q}{2^q},
\quad
\Phi^{(3)}_{(q,q+1)}(z)
=
-\frac{(1+z)^{q+1}}{2^qz^{q+2}}.
\]
A similar calculation will lead to
\[
\varphi^{(1)}_{(q,q+1)}(z)
=
-\sum_{m=1}^\infty
mz^m
\Bigl[
  \frac{1}{2^{2q+1}}
  \sum_{r=0}^{2q+1}
  \binom{2q+1}{r}z^{r-(q+1)}
\Bigr],\quad
\tau^{(1)}_{0,(q,q+1)}
=
-\frac{1}{2^{2q+1}}
\sum_{r=0}^{q}
\binom{2q+1}{r}(q+1-r).
\]
We also have
\[
\varphi^{(3)}_{(q,q+1)}(z)
=
\sum_{m=1}^\infty
mz^m
\Bigl[
  -\frac{1}{2^{q+1}}
  \sum_{r=0}^{q+1}
  \binom{q+1}{r}z^{r-(q+2)}
\Bigr],\quad
\tau^{(3)}_{0,(q,q+1)}
=
\frac{1}{2^{q+1}}
\sum_{r=0}^{q+1}
\binom{q+1}{r}(q+2-r).
\]
Then we use the same arguments by \eqref{eq:binom_formula} to get
\begin{equation}\label{eq:qq2}
\tau_{0,(q,q+1)}
=
\frac{2q+5}{4}-\frac{2q+1}{2^{2q+2}}\binom{2q}{q}.
\end{equation}
For the term $\tau_{0,(q+1,q)}$, we get the expression 
\[
\varrho_{q+1}(z)\varrho_{q}(1/z)
=-\frac{z}{(1-z)^2}\left[
\frac{(1+z)^{2q+1}}{2^{2q+1}z^q}
-\frac{z(1+z)^{q+1}}{2^{q+1}}
-\frac{(1+z)^q}{2^qz^{q+1}}
+1\right].
\]
Following similar arguments, we have $\tau^{(2)}_{0,(q+1,q)}=
\tau^{(4)}_{0,(q+1,q)}=0$, and 
\[
\varphi^{(1)}_{(q+1,q)}(z)
=
-\sum_{m=1}^\infty
mz^m
\Bigl[
  \frac{1}{2^{2q+1}}
  \sum_{r=0}^{2q+1}
  \binom{2q+1}{r}z^{r-q}
\Bigr],\quad
\tau^{(3)}_{0,(q+1,q)}
=
-\frac{1}{2^{2q+1}}
\sum_{r=0}^{q-1}
\binom{2q+1}{r}(q-r).
\]
Since $\Phi_{(q+1,q)}^{(3)}=\Phi^{(3)}$, we still have 
\[
\tau^{(3)}_{0,(q+1,q)}
=\tau^{(3)}_{0,(q,q)}=
\frac{1}{2^q}
\sum_{r=0}^q \binom{q}{r}(r+1).
\]
Therefore we use the same arguments by \eqref{eq:binom_formula} to get
\begin{equation}\label{eq:qq3}
\tau_{0,(q+1,q)}=\frac{2q+5}{4}-\frac{2q+1}{2^{2q+2}}\binom{2q}{q}.
\end{equation}
With the use of standard identities of binomial coefficients
\[
\binom{2q+2}{q+1}=\frac{2(2q+1)}{q+1}\binom{2q}{q},
\]
we sum up \eqref{eq:qq1}--\eqref{eq:qq3} with weights in \eqref{eq:alpha_combine} and get
\[
\tau_0=\frac{q+3-\alpha}{2}-\frac{(2q+1)-\alpha^2}{2^{2q+2}}\binom{2q}{q}.~~~\Box
\]

\vspace{0.5cm}

{}
 
\vspace{0.5cm}

{\sc Proof of \cref{prop:relerror}:}  Recall from \eqref{eq:relerror} that the relative performance ratio of a policy $j$ under no information sharing is given by  
$${\cal E}^j(\kappa) = \frac{{\cal C}^j(\kappa)}{{\cal C}^*(\kappa)}.$$ 

To show that ${\cal E}^{\ti{MA}}(\kappa)$ and ${\cal E}^{\ti{ES}}(\kappa)$ grow unboundedly as $\kappa \downarrow 0$, we derive lower bounds for each and show that these bounds diverge in the limit. To this end, let us first derive an upper bound for ${\cal C}^*(\kappa)$. From \cref{prop:solnosharing}, recall that ${\cal S}_\kappa$ solves $4\,x^2\,(5-\log(4\,x^2))=\kappa^2$ for $0 \leq x \leq 1/2$ and 
$${\cal C}^*(\kappa)={\kappa \over 2}\,\sqrt{5-\log(4\,{\cal S}^2_\kappa)}+{\cal S}_\kappa ={\kappa^2 \over 4\,{\cal S}_\kappa}+{\cal S}_{\kappa}.$$
It is not hard to see that ${\cal S}_\kappa$ satisfies the following properties: (i) it is monotonically  increasing in $\kappa$, (ii) ${\cal S}_\kappa \leq \kappa/2$, and (iii) $\lim_{\kappa \downarrow 0} {\cal S}^2_\kappa=0$. Also, the function
${\kappa^2 \over 4\,y}+y$ is minimized at $y=\kappa/2$. Thus, from property (ii), we 
have that ${\cal C}^*(\kappa) \leq {\kappa^2 \over 4\,y}+y$ for all $y \leq {\cal S}_\kappa$. We conclude that we can find an upper bound for ${\cal C}^*(\kappa)$ by finding a lower bound on ${\cal S}_\kappa$. \medskip

\begin{lemma}\label{lem:LB} Assume $\kappa \leq 1$. A lower bound for  ${\cal S}_\kappa$ is given by $\kappa^2/6$.
\end{lemma}
{\sc Proof of \cref{lem:LB}:} {\sf The equation $4\,x^2\,(5-\log(4\,x^2))=\kappa^2$ is equivalent to $5+\log(y)=\kappa^2\,y$ for $y=1/(4x^2)$. Thus, a lower bound on $x$ is obtained from an upper bound on $y$. Since $\log(y)=2\,\log(\sqrt{y}) \leq 2\,(\sqrt{y}-1)$, such a lower bound for $y$ is obtained solving $5+2\,(\sqrt{y}-1)=\kappa^2\,y$. We get $y=\kappa^{-4}\,(1+\sqrt{1+3\kappa^2})^2$. As a result, a lower bound for $x$ is equal to $\kappa^2/(2(1+\sqrt{1+3\kappa^2})) \geq \kappa^2/6$. $\Box$}
\vspace{0.2cm}

It follows that an upper bound for ${\cal C}^*(\kappa)$ is given by
$$ {\cal C}^*(\kappa) \leq {\kappa \over 2}\,\sqrt{5-\log(\kappa^4/9)}+{\kappa^2 \over 6}.$$

Let us now obtain lower bounds for ${\cal C}^{\ti{MA}}(\kappa)$ and  ${\cal C}^{\ti{ES}}(\kappa)$. 
From \eqref{eq:CostMA}, we have that 
\begin{align*}
{\cal C}^{\ti{MA}}(\kappa)&=\min_{N \in \mathbb{N}_0} \,\left\{\kappa\,\sqrt{\frac{(N+2)\,(2\,N+3)}{6\,(N+1)}} + \frac{1}{N+1} \right\} 
 \geq \min_{N \in \mathbb{N}_0} \,\left\{\kappa\,\sqrt{\frac{(N+1)\,(2\,(N+1))}{6\,(N+1)}} + \frac{1}{N+1} \right\} \\
 & \geq  \min_{x \in \mathbb{R}_+} \,\left\{\kappa\,\sqrt{\frac{x}{3}} + \frac{1}{x} \right\} =  \left(3\,\kappa \over 2\right)^{2 \over 3}.
\end{align*}
Combining  the upper bound for ${\cal C}^*(\kappa)$ and the lower bound for ${\cal C}^{\ti{MA}}(\kappa)$, we obtain
$$\lim_{\kappa \downarrow 0} {\cal E}^{\ti{MA}}(\kappa)  \geq \lim_{\kappa \downarrow 0} {\left(3\,\kappa \over 2\right)^{2 \over 3} \over {\kappa \over 2}\,\sqrt{5-\log(\kappa^4/9)}+{\kappa^2 \over 6}}=\infty.$$

Using a similar argument, let us lower bound the cost under an exponential smoothing policy. From \eqref{eq:CostES}, we have that
\begin{align*}
{\cal C}^{\ti{ES}}(\kappa)&=\min_{0 \leq \theta <1} \,\left\{{\kappa \over \sqrt{1-\theta^2}}+ 1-\theta \right\} = \left(\kappa^2 \over \theta_{\kappa}\right)^{1 \over 3}
+1-\theta_\kappa \geq \left(\kappa^2 \over \theta_{\kappa}\right)^{1 \over 3},
\end{align*}
where $\theta_\kappa$ is the optimal value of the smoothing parameter given $\kappa$, that is, the solution to the first-order optimality condition $\kappa^2\,\theta^2 = (1 - \theta^2)^3$ in $[0,1]$.  It is easy to see that $\lim_{\kappa \downarrow 0} \theta_{\kappa} =1$. It follows that  
$$\lim_{\kappa \downarrow 0} {\cal E}^{\ti{ES}}(\kappa)  \geq \lim_{\kappa \downarrow 0} {\left(\kappa^2 \over \theta_{\kappa}\right)^{1 \over 3} \over{\kappa \over 2}\,\sqrt{5-\log(\kappa^4/9)}+{\kappa^2 \over 6}}=\infty.$$
This completes the proof that the relative performance of MA and ES policies grows unboundedly as $\kappa \downarrow 0$. \qed
\vspace{0.5cm}

{\sc Proof of \cref{prop:BN_smallkappa_opt}:} From \cref{prop:myopic} and \cref{thm:BN_cost_alpha}, for any $\varphi \in \Ad_q$ we have 
\[{\cal C}(\varphi;\kappa) \geq |\psi(0)|\, (\kappa+2^{-q})\qquad \mbox{and}\qquad {\cal C}(\varphi_q^{\BN};\kappa) \leq \psi_{\sup}\,\left(\frac{q+2}{2}\right)\,\kappa+|\psi(0)|\,2^{-q}.\]
It follows that
\begin{align*}
\frac{{\cal C}(\varphi;\kappa)}{{\cal C}(\varphi_q^{\BN};\kappa)}  & \geq \frac{|\psi(0)|\, (\kappa+2^{-q})}{\psi_{\sup}\,\left(\frac{q+2}{2}\right)\,\kappa+|\psi(0)|\,2^{-q}}= 1- \Big(\frac{\psi_{\sup}\,\left(\frac{q+2}{2}\right)-|\psi(0)|}{\psi_{\sup}\,\left(\frac{q+2}{2}\right)\,\kappa+|\psi(0)|\,2^{-q}}\Big)\,\kappa \\ & \geq 1- \Big(\frac{\psi_{\sup}}{|\psi(0)|}\Big)\,\left(\frac{q+2}{2}\right)\,2^q\,\kappa=1-O\!\left(q\,2^q\right)\,\kappa,
\end{align*}
which establishes the stated bound.  \qed

%% file: Management_Science/Appendix_Cost_Criterion_MS.tex
\section{Supply-Chain Performance and Cost Criterion}\label{App:CostCriterion}
\setcounter{equation}{0}
\renewcommand{\theequation}{B\arabic{equation}}

In this appendix, we formulate the retailer's inventory problem and show that, under our inventory-cost structure, it is equivalent to the linear cost criterion \eqref{eq:cost_criterion}. We first state the retailer's Stackelberg problem in payoff terms, then derive the retailer's and manufacturer's long-run cost representations, and finally show how the payoff formulation reduces to the cost-minimization problem used in the main text.\smallskip

{\sf The Retailer's Inventory Problem:}
Suppose the retailer chooses a wholesale price $w\in\mathbb{R}_+$ together with a replenishment policy $\varphi\in\Ad$. Since admissible policies satisfy $\sum_{n\geq 0}\varphi_n=1$, average orders equal average demand; let $d:=\e[D_t]=\e[O_t]$. The retailer's expected long-run average payoff is
\begin{align*}
\Pi^{\ti R}(\varphi,w)
&= \liminf_{T \to \infty} \frac{1}{T} \,\e\left[\sum_{t=1}^T p\,D_t-w\,O_t- h^{\ti R} (I_t)^+ - b^{\ti R} (I_t)^-\right] = (p-w)\,d-{\cal C}^{\ti R}(\varphi),
\end{align*}
where ${\cal C}^{\ti R}(\varphi)$ is the retailer's long-run average inventory cost. Similarly, the manufacturer's expected long-run average payoff, evaluated under its best-response state-dependent base-stock policy, is
\begin{align*}
\Pi^{\ti M}(\varphi,w)
&= \liminf_{T \to \infty} \frac{1}{T} \,\e\left[\sum_{t=1}^T (w-c)\,O_t- h^{\ti M}(S_t-O_{t+1})^+ - \Delta c\,(S_t-O_{t+1})^-\right] = (w-c)\,d-{\cal C}^{\ti M}(\varphi),
\end{align*}
where ${\cal C}^{\ti M}(\varphi)$ is the manufacturer's long-run average inventory holding and expediting cost.\smallskip

We formulate the retailer's decision problem as a Stackelberg problem in which the retailer chooses $(w,\varphi)$ to maximize its own payoff while guaranteeing the manufacturer's participation:
\begin{equation}\label{eq:optim}
\tag{Retailer's Inventory Problem}
\sup_{(w,\varphi)\in\mathbb{R}_+\times\Ad}\;\Pi^{\ti R}(\varphi,w)
\qquad \mbox{subject to}\qquad
\Pi^{\ti M}(\varphi,w) \geq \pi^{\ti M}.
\end{equation}
This formulation captures the central economic tension of the model: by changing $\varphi$, the retailer affects not only its own inventory performance, but also the predictability and cost of the order stream faced by the manufacturer.\medskip

{\sf Cost Representations:}
We now express ${\cal C}^{\ti R}(\varphi)$ and ${\cal C}^{\ti M}(\varphi)$ in terms of the two performance measures used in the main text. The retailer's long-run average expected per-period inventory cost is
\[
{\cal C}^{\ti R}(\varphi)= \limsup_{T \to \infty} \frac{1}{T}\,\e\left[\sum_{t=1}^T h^{\ti R}(I_t)^+ + b^{\ti R}(I_t)^-\right],
\]
where $h^{\ti R}$ and $b^{\ti R}$ are the retailer's per-unit per-period holding and backorder costs, respectively. By combining \eqref{eq:def_demand}, \eqref{eq:invR}, and \eqref{eq:retailerorder}, the retailer's inventory process $\{I_t\}$ converges in distribution to a normal random variable $I_\infty$ with mean $\mu_{\ti I}(\varphi)=\lim_{t \to \infty} \e[I_t]$ and variance $\sigma_{\ti I}^2(\varphi)=\lim_{t\to\infty}\var[I_t]$. Hence,
\[
{\cal C}^{\ti R}(\varphi)= \sigma_{\ti I}(\varphi)\left[h^{\ti R}\,{\cal L}\!\left(-\frac{\mu_{\ti I}(\varphi)}{\sigma_{\ti I}(\varphi)}\right)+b^{\ti R}\,{\cal L}\!\left(\frac{\mu_{\ti I}(\varphi)}{\sigma_{\ti I}(\varphi)}\right)\right],
\]
where ${\cal L}(z)=\phi(z)-z(1-\Phi(z))$, with $\phi$ and $\Phi$ denoting the pdf and cdf of the standard normal distribution. The mean $\mu_{\ti I}(\varphi)$ depends on the retailer's inventory-position choice and can therefore be optimized. Evaluating the retailer's cost at this optimal value yields
\begin{equation}\label{eq:Cost_Retailer2}
{\cal C}^{\ti R}(\varphi) = K^{\ti R}\,\sigma_{\ti I}(\varphi),
\quad \text{where} \quad
K^{\ti R} := (h^{\ti R} + b^{\ti R})\,{\cal L}\left(\Phi^{-1}\left(\frac{b^{\ti R}}{h^{\ti R} + b^{\ti R}}\right)\right) + h^{\ti R}\,\Phi^{-1}\left(\frac{b^{\ti R}}{h^{\ti R} + b^{\ti R}}\right).
\end{equation}

Turning to the manufacturer, its long-run average expected per-period cost under its state-dependent base-stock policy $\{S_t\}$ is
\[
{\cal C}^{\ti M}(\varphi)= \limsup_{T \to \infty} \frac{1}{T}\,\e\left[\sum_{t=1}^T h^{\ti M}(S_t-O_{t+1})^+ + \Delta c\,(S_t-O_{t+1})^-\right],
\]
where $h^{\ti M}$ is the per-unit per-period holding cost and $\Delta c$ is the per-unit expediting cost. As discussed in \cref{sec:Model}, the manufacturer chooses $S_t$ to minimize
\[
\e\big[h^{\ti M}(S_t-O_{t+1})^+ + \Delta c\,(S_t-O_{t+1})^- \mid {\cal F}^{\ti M}_t\big],
\]
where ${\cal F}^{\ti M}_t=\sigma(O_\tau\colon \tau \le t)$ denotes the manufacturer's information at period $t$. This gives
\[
S_t(\varphi)
= m_t(\varphi) + \sigma_{\ti M}(\varphi)\,\zeta^{\ti M},
\quad \text{where } 
\zeta^{\ti M} := \Phi^{-1}\!\left(\frac{\Delta c}{h^{\ti M}+\Delta c}\right),
\]
and where $m_t(\varphi)$ and $\sigma^2_{\ti M}(\varphi)$ are the manufacturer's one-step-ahead mean forecast and mean squared forecast error (MSFE), respectively. Substituting the optimal base-stock level $S_t(\varphi)$ into the manufacturer's objective yields
\begin{equation}\label{eq:manufacturercost}
{\cal C}^{\ti M}(\varphi)
= K^{\ti M}\,\sigma_{\ti M}(\varphi),
\qquad
\text{where }
K^{\ti M} := h^{\ti M}\,\zeta^{\ti M} + (h^{\ti M}+\Delta c)\,{\cal L}(\zeta^{\ti M}).
\end{equation}
Thus, the retailer's cost is proportional to the volatility of its net inventory, while the manufacturer's cost is proportional to the one-step-ahead root mean squared forecast error of the retailer's orders.\medskip

 {\sf Equivalence to the Linear Cost Criterion:}
We now show that the retailer's inventory problem is equivalent to minimizing the cost criterion used in the main text. For any fixed policy $\varphi\in\Ad$, the retailer optimally chooses the smallest wholesale price that satisfies the manufacturer's participation constraint. Hence, the participation constraint binds:
\[
\Pi^{\ti M}(\varphi,w)=\pi^{\ti M}.
\]
Using \eqref{eq:manufacturercost}, the corresponding wholesale price is
\[
w(\varphi)
=
c+\frac{K^{\ti M}\,\sigma_{\ti M}(\varphi)+\pi^{\ti M}}{d}.
\]
Substituting this expression into the retailer's payoff gives
\begin{align*}
\Pi^{\ti R}(\varphi,w(\varphi))
&= (p-c)d-\pi^{\ti M}
   -K^{\ti R}\,\sigma_{\ti I}(\varphi)
   -K^{\ti M}\,\sigma_{\ti M}(\varphi).
\end{align*}
The first two terms are independent of $\varphi$. Therefore, maximizing the retailer's payoff over replenishment policies is equivalent to minimizing
\[
K^{\ti R}\,\sigma_{\ti I}(\varphi)
+
K^{\ti M}\,\sigma_{\ti M}(\varphi).
\]
After normalizing by $K^{\ti M}>0$, the retailer's problem reduces exactly to
\[
{\cal C}(\varphi;\kappa)=\kappa\,\sigma_{\ti I}(\varphi)+\sigma_{\ti M}(\varphi),
\qquad
\text{where}
\qquad
\kappa=\frac{K^{\ti R}}{K^{\ti M}}.
\]
This is the cost criterion in \eqref{eq:cost_criterion}. Thus, the linear objective used in the paper is not only a reduced-form representation of supply-chain performance; it is the equivalent policy-selection problem of a strategic retailer that internalizes the manufacturer's participation constraint through the wholesale price.\smallskip

The same criterion also admits a centralized interpretation. Combining \eqref{eq:Cost_Retailer2} and \eqref{eq:manufacturercost}, the cumulative long-run expected per-period cost of the supply chain is
\begin{align*}
{\cal C}^{\ti C}(\varphi)
&= {\cal C}^{\ti R}(\varphi)+{\cal C}^{\ti M}(\varphi) \\
&= K^{\ti R}\,\sigma_{\ti I}(\varphi)+K^{\ti M}\,\sigma_{\ti M}(\varphi).
\end{align*}
Accordingly, for $\kappa=K^{\ti R}/K^{\ti M}$, minimizing ${\cal C}(\varphi;\kappa)$ is equivalent to minimizing total supply-chain inventory-related cost. More generally, allowing $\kappa$ to vary generates a family of objectives that captures different planning perspectives, including the limiting cases $\kappa\downarrow 0$ and $\kappa\uparrow\infty$, which correspond to minimizing exclusively the manufacturer's and retailer's costs, respectively. In this sense, the cost criterion traces the full spectrum of trade-offs between retailer inventory volatility, manufacturer forecastability, and overall supply-chain performance.\smallskip

\medskip

\begin{remark}[Competitive wholesale market]\label{rem:competitive_wholesale} {\sf
If the retailer cannot set the wholesale price---for instance, because $w$ is determined competitively in an upstream market---then the retailer takes $w$ as exogenous and chooses only its replenishment policy $\varphi\in\Ad$. In that case, the manufacturer's participation constraint restricts the set of feasible policies through
\[
(w-c)d-K^{\ti M}\sigma_{\ti M}(\varphi)\geq \pi^{\ti M}.
\]
Introducing a Lagrange multiplier for this constraint again produces an objective with a weighted trade-off between the retailer's inventory term $\sigma_{\ti I}(\varphi)$ and the manufacturer's forecast-error term $\sigma_{\ti M}(\varphi)$. Thus, even when the wholesale price is exogenous, the relaxed retailer problem has the same linear structure as the cost-minimization formulation above. {\scriptsize$\blacksquare$}}
\end{remark}

%% file: Management_Science/Appendix_Benchmark_Policies.tex
\section{Benchmark Policies}\label{app:benchmark_details}
\setcounter{equation}{0}
\renewcommand{\theequation}{C\arabic{equation}}

This appendix provides the analytical details underlying the benchmark comparison in \cref{sec:benchmark}. In particular, we report the values of the manufacturer's root MSFE, \(\sigma_{\ti M}\), the retailer's inventory volatility, \(\sigma_{\ti I}\), and the optimized benchmark costs for the Simple Moving Average (MA) and Exponential Smoothing (ES) policy classes.

Throughout this appendix, demand is i.i.d., so \(\psi(z)\equiv\mathds{1}\), and the total supply-chain cost of a policy \(\varphi\) is
\[
{\cal C}_{\iid}(\varphi;\kappa)
=
\kappa\,\sigma_{\ti I}(\varphi;\mathds{1})
+
\sigma_{\ti M}(\varphi;\mathds{1}).
\]

\medskip

{\bf Simple Moving Average (MA).}
Recall that the MA class is given by
\[
\Ad^{\ti{MA}}
:=
\Bigg\{
\varphi_{\tim N}^{\ti{MA}} \in \Ad
\;\colon\;
\varphi_{\tim N}^{\ti{MA}}(z)
=
\frac{1}{N+1}\sum_{n=0}^{N}z^n
=
\frac{1-z^{N+1}}{(N+1)(1-z)},
\quad
N\in\mathbb{N}_0
\Bigg\}.
\]
The coefficients of \(\varphi_{\tim N}^{\ti{MA}}\) are uniform over the finite window \(0,\ldots,N\). Thus, a unit demand shock is transmitted to the manufacturer gradually, with weight \(1/(N+1)\) in each of the next \(N+1\) periods.

Simple moving-average policies are invertible in the sense of \cref{dfn:invertible}. Therefore, by \cref{eq:invertible_orders}, the manufacturer's root MSFE is
\[
\sigma_{\ti M}(\varphi_{\tim N}^{\ti{MA}};\mathds{1})
=
\frac{1}{N+1}.
\]
The retailer's inventory volatility is
\[
\sigma_{\ti I}(\varphi_{\tim N}^{\ti{MA}};\mathds{1})
=
\sqrt{
\frac{(N+2)(2N+3)}{6(N+1)}
}.
\]
Hence, for a fixed \(N\), the total cost of the MA policy is
\[
{\cal C}_{\iid}(\varphi_{\tim N}^{\ti{MA}};\kappa)
=
\kappa
\sqrt{
\frac{(N+2)(2N+3)}{6(N+1)}
}
+
\frac{1}{N+1}.
\]

For each value of \(\kappa\), the best MA policy is obtained by minimizing over \(N\in\mathbb{N}_0\):
\begin{equation}\label{eq:CostMA}
{\cal C}_{\iid}^{\ti{MA}}(\kappa)
:=
\min_{N\in\mathbb{N}_0}
\left\{
\kappa
\sqrt{
\frac{(N+2)(2N+3)}{6(N+1)}
}
+
\frac{1}{N+1}
\right\}.
\end{equation}
This expression makes explicit the trade-off induced by the smoothing window \(N\). Increasing \(N\) reduces the manufacturer's root MSFE, since
\[
\sigma_{\ti M}(\varphi_{\tim N}^{\ti{MA}};\mathds{1})
=
\frac{1}{N+1},
\]
but increases the retailer's inventory volatility,
\[
\sigma_{\ti I}(\varphi_{\tim N}^{\ti{MA}};\mathds{1})
=
\sqrt{
\frac{(N+2)(2N+3)}{6(N+1)}
}.
\]
Thus, stronger smoothing improves upstream forecastability at the expense of downstream inventory stability.

\medskip

{\bf Exponential Smoothing (ES).}
Recall that the ES class is given by
\[
\Ad^{\ti{ES}}
:=
\Bigg\{
\varphi_{\tim \theta}^{\ti{ES}} \in \Ad
\;\colon\;
\varphi_{\tim \theta}^{\ti{ES}}(z)
=
(1-\theta)\sum_{n=0}^{\infty}(\theta z)^n
=
\frac{1-\theta}{1-\theta z},
\quad
\theta\in[0,1)
\Bigg\}.
\]
The ES impulse response has infinite support and geometrically decaying coefficients. A unit demand shock is transmitted with weight \((1-\theta)\theta^n\) after \(n\) periods.

Exponential smoothing policies are also invertible. By \cref{eq:invertible_orders}, the manufacturer's root MSFE is
\[
\sigma_{\ti M}(\varphi_{\tim \theta}^{\ti{ES}};\mathds{1})
=
1-\theta.
\]
The retailer's inventory volatility is
\[
\sigma_{\ti I}(\varphi_{\tim \theta}^{\ti{ES}};\mathds{1})
=
\frac{1}{\sqrt{1-\theta^2}}.
\]
Therefore, for a fixed \(\theta\), the total cost of the ES policy is
\[
{\cal C}_{\iid}(\varphi_{\tim \theta}^{\ti{ES}};\kappa)
=
\frac{\kappa}{\sqrt{1-\theta^2}}
+
1-\theta.
\]

For each value of \(\kappa\), the best ES policy is obtained by minimizing over \(\theta\in[0,1)\):
\begin{equation}\label{eq:CostES}
{\cal C}_{\iid}^{\ti{ES}}(\kappa)
:=
\min_{\theta\in[0,1)}
\left\{
\frac{\kappa}{\sqrt{1-\theta^2}}
+
1-\theta
\right\}.
\end{equation}
For an interior optimizer \(\theta^*\in[0,1)\), the first-order condition is
\[
\frac{\kappa \theta^*}{(1-(\theta^*)^2)^{3/2}}=1,
\]
or equivalently,
\[
\kappa^2(\theta^*)^2
=
(1-(\theta^*)^2)^3.
\]
As in the MA class, larger values of the smoothing parameter improve the manufacturer's forecastability but increase the retailer's inventory volatility:
\[
\sigma_{\ti M}(\varphi_{\tim \theta}^{\ti{ES}};\mathds{1})
=
1-\theta
\quad\text{decreases in }\theta,
\]
whereas
\[
\sigma_{\ti I}(\varphi_{\tim \theta}^{\ti{ES}};\mathds{1})
=
\frac{1}{\sqrt{1-\theta^2}}
\quad\text{increases in }\theta.
\]

\medskip

{\bf Myopic Policy as a Special Case.}
Under i.i.d. demand, both benchmark classes include the myopic policy. Setting \(N=0\) in the MA class gives
\[
\varphi^{\ti{MA}}_{0}(z)
=
\mathds{1},
\]
and setting \(\theta=0\) in the ES class gives
\[
\varphi^{\ti{ES}}_{0}(z)
=
\mathds{1}.
\]
Thus,
\[
\varphi^{\ti{MA}}_{0}(z)
=
\varphi^{\ti{ES}}_{0}(z)
=
\varphi^{\ti{MP}}(z)
=
\mathds{1},
\]
so that \(O_t=D_t\). By \cref{prop:myopic}, as \(\kappa\uparrow\infty\), the myopic policy is optimal. Consequently, both MA and ES inherit asymptotic optimality in the large-\(\kappa\) regime through their myopic special cases.

\medskip

{\bf Relative Performance.}
For a policy class \(\widehat{\Ad}\), define its relative performance by
\begin{equation}\label{eq:relerror_appen}
{\cal E}_{\iid}^{\widehat{\Ad}}(\kappa)
:=
\inf_{\varphi\in\widehat{\Ad}}
\frac{
{\cal C}_{\iid}(\varphi;\kappa)
}{
{\cal C}_{\iid}^{*}(\kappa)
}.
\end{equation}
For the MA and ES benchmarks, this corresponds to using the optimized costs in \eqref{eq:CostMA} and \eqref{eq:CostES}, respectively:
\[
{\cal E}_{\iid}^{\ti{MA}}(\kappa)
=
\frac{
{\cal C}_{\iid}^{\ti{MA}}(\kappa)
}{
{\cal C}_{\iid}^{*}(\kappa)
},
\qquad
{\cal E}_{\iid}^{\ti{ES}}(\kappa)
=
\frac{
{\cal C}_{\iid}^{\ti{ES}}(\kappa)
}{
{\cal C}_{\iid}^{*}(\kappa)
}.
\]

The expressions above are used to generate the relative-performance values reported in \cref{table:suboptbenchmarks_noinfonew}. They also clarify why MA and ES do not provide uniform approximation guarantees. Even after optimizing over their smoothing parameters, their relative performance deteriorates as \(\kappa\downarrow 0\):
\[
\lim_{\kappa\downarrow 0}
{\cal E}_{\iid}^{\ti{MA}}(\kappa)
=
\lim_{\kappa\downarrow 0}
{\cal E}_{\iid}^{\ti{ES}}(\kappa)
=
\infty.
\]
Consequently, neither the simple moving-average family nor the exponential-smoothing family constitutes an \(\alpha\)-approximation class of inventory policies for any finite \(\alpha\).

The failure occurs in the small-\(\kappa\) regime, where the supply-chain cost is dominated by the manufacturer's forecasting component. In this regime, generic smoothing rules such as MA and ES do not preserve enough useful demand information in the order stream. By contrast, the BN class uses a calibrated delay that enables the manufacturer to recover information more effectively through optimal forecasting, which explains its stronger performance in \cref{table:suboptbenchmarks_noinfonew}.